\documentclass[preprint]{elsarticle}
\usepackage[a-1b]{pdfx}
\journal{Computers \& Security} 

\usepackage[T1]{fontenc}
\usepackage{lmodern}
\usepackage{amsthm}
\usepackage{verbatim}
\usepackage{csquotes}
\usepackage{xspace}
\usepackage{soul}
\usepackage{multicol}
\usepackage{mathtools}
\usepackage[export]{adjustbox}
\usepackage{lscape}
\usepackage{amsmath}
\usepackage{amssymb}
\usepackage{tabularx}
\usepackage{array}
\usepackage[all]{xy}
\usepackage{thmtools}
\usepackage{thm-restate}

\usepackage[dvipsnames,svgnames]{xcolor}

\usepackage{tikz}
\usetikzlibrary{positioning,calc}
\usetikzlibrary{backgrounds}
\usetikzlibrary{arrows,shapes}
\usetikzlibrary{tikzmark} 
\usetikzlibrary{calc} 

\usetikzlibrary{fadings, shadings}
\usepackage{tikz-3dplot}

\usepackage{array, booktabs, longtable, caption}
\usepackage{multirow, makecell}

\usepackage{spi}
\usepackage{prooftree}
\usepackage{hyperref}

\usepackage{colortbl}
\usetikzlibrary{decorations.pathmorphing}
\usetikzlibrary{decorations.pathreplacing,tikzmark, arrows.meta}
\tikzset{snake it/.style={decorate, decoration=snake}}

\tikzstyle{causal}=[purple, dotted, -{Turned Square[open]}, thick] 
\tikzstyle{structural} = [red, dashed, ->>, thick] 
\tikzstyle{transition} = [->, thick] 

\newcommand*{\eg}{e.g.\@,\xspace}
\newcommand*{\Eg}{E.g.\@,\xspace}

\newcommand*{\ie}{i.e.\@,\xspace}
\newcommand*{\aka}{a.k.a.\@\xspace}

\renewcommand{\st}{s.t.\@\xspace}
\newcommand*{\resp}{resp.\@\xspace}
\newcommand*{\wrt}{w.r.t.\@\xspace}

\renewcommand{\aka}{a.k.a.\@\xspace}

\makeatletter
\newcommand*{\etc}{%
	\@ifnextchar{.}%
	{etc}%
	{etc.\@\xspace}%
}
\makeatother

\AtBeginDocument{%

}

\newcommand{\dlts}[2]{\updownlts^{#1}_{#2}}
\renewcommand{\dlts}[2]{\xrightarrow[\raisebox{.15em}{\protect{\(\scriptstyle #2\)}}]{#1}} 
\renewcommand{\lts}[1]{\xrightarrow{#1}}
\renewcommand{\pair}[2]{\left\langle#1, #2\right\rangle} 

\newcommand{\simi}[1]{\mathrel{\preceq_{#1}}}
\newcommand{\nsimi}[1]{\mathrel{\npreceq_{#1}}}

\newcommand{\bisimi}[1]{\mathrel{\sim_{\textnormal{\emph{#1}}}}}
\newcommand{\nbisimi}[1]{\mathrel{\not\sim_{\textnormal{\emph{#1}}}}}

\newcommand{\alia}{\alpha}
\newcommand{\alib}{\beta}
\newcommand{\alic}{\gamma}

\newcommand{\match}[1]{\mathopen{\left[ #1 \right]}}

\newcommand{\prefix}[1]{\textcolor{red}{#1}}

\renewcommand{\names}[2]{#2} 
\renewcommand{\env}{\vec{\alic}} 

\newcommand{\fal}[1]{\mathopen{\mathrm{fa}}\left(#1\right)} 
\newcommand{\aliroot}{\lambda} 
\newcommand{\lout}[2]{\co{#1}(#2)} %

\DeclareMathOperator{\id}{id} 
\DeclareMathOperator{\pok}{ok}

\newcommand{\Indy}{\smile} 
\newcommand{\notIndy}{\mathrel{\not\smile}} 
\newcommand{\sIndy}{\mathrel{I_{\ell}}}
\newcommand{\ESS}{\mathrel{\mathsf{S}}}

\definecolor{green(ncs)}{rgb}{0.0, 0.62, 0.42}
\newcommand{\diff}[1]{\textcolor{green(ncs)}{#1}}

\newtheorem{thm}{Theorem}
\newtheorem{lem}[thm]{Lemma}
\newdefinition{rmk}[thm]{Remark}
\newtheorem{defi}[thm]{Definition}
\newtheorem{example}[thm]{Example} 
\newtheorem{remark}[thm]{Remark} 
\newproof{pf}{Proof}

\newcommand{\boxm}[1]{\mathopen{\big[ #1 \big]}} 
\newcommand{\diam}[1]{\mathopen{\big\langle #1 \big\rangle}}
\newcommand{\ttt}{\texttt{true}}
\newcommand{\fff}{\texttt{false}}
\newcommand{\yields}{\supset}

\usepackage{stackengine}
\newcommand{\eboxm}[2]{\mathopen{\big[ \stackanchor[1.5pt]{\footnotesize $#1$}{\footnotesize $#2$} \big]}} 
\newcommand{\ediam}[2]{\mathopen{\big\langle \stackanchor[1.5pt]{\footnotesize $#1$}{\footnotesize $#2$} \big\rangle}}
\newcommand{\event}[2]{\big( \stackanchor[1.5pt]{\footnotesize $#1$}{\footnotesize $#2$} \big)}

\newcommand{\reader}{r}
\newcommand{\passport}{p}
\DeclareMathOperator{\getchallenge}{get}
\DeclareMathOperator{\error}{err}

\DeclareMathOperator{\ch}{ch}

\newcommand{\system}[1]{\textnormal{#1}} 

\newcommand{\SpecUK}{\system{Spec}_{\ch}}
\newcommand{\SystemUK}{\system{System}_{\ch}}
\newcommand{\MainUK}{\system{P}}
\newcommand{\SystemCSF}{\system{System}_{\getchallenge}}
\newcommand{\SpecCSF}{\system{Spec}_{\getchallenge}}

\newcommand{\SystemF}{\system{System}_{2}}
\newcommand{\SpecF}{\system{Spec}_{2}}

\newcommand{\SystemMin}{\system{System}_{\min}}
\newcommand{\SpecMin}{\system{Spec}_{\min}}

\newcommand{\Reader}{\emph{V}}

\newcommand{\locp}{\prefix{001}[]}
\newcommand{\locra}{\prefix{000}[]}
\newcommand{\locrb}{\prefix{0100}[]}
\newcommand{\locq}{\prefix{1000}[]}

\newcommand{\HPFM}{\mbox{HP-$\mathcal{F}\!\mathcal{M}$}}

\newcommand{\rulel}[1]{\textnormal{\textsc{#1}}}

\definecolor{clem}{HTML}{A30B37} 
\definecolor{ross}{HTML}{BABD8D} 
\definecolor{chri}{HTML}{EB6424} 
\definecolor{sjou}{HTML}{4C956C} 
\usepackage[normalem]{ulem}

\newcommand{\locA}{$\ell_1$}
\newcommand{\locB}{$\ell_2$}
\newcommand{\locC}{$\ell_3$}
\newcommand{\locD}{$\ell_4$}

\renewcommand\appendixautorefname[1]{}

\begin{document}

\begin{frontmatter}
\title{%
Unlinkability and History Preserving Bisimilarity
}

\author[1]{Cl{\'e}ment Aubert}
\ead{caubert@augusta.edu}

\author[2]{Ross Horne\texorpdfstring{\corref{cor1}}{}}
\ead{ross.horne@strath.ac.uk}

\author[3]{Christian Johansen}
\ead{christian.johansen@ntnu.no}

\author[4]{Sjouke Mauw}
\ead{sjouke.mauw@uni.lu}

\cortext[cor1]{Corresponding author} 

\affiliation[1]{
	organization={School of Computer and Cyber Sciences}, 
	addressline={Augusta University},
	city={Augusta},
	postcode={30904},
	state={GA}, 
	country={USA}
}

\affiliation[2]{
	organization={Computer and Information Sciences},
	addressline={University of Strathclyde},
	city={Glasgow},
	country={United Kingdom}
}

\affiliation[3]{
	organization={Department of Information Security and Communication Technology},
	addressline={Norwegian University of Science and Technology},
	city={Gj\o{}vik},
	country={Norway}
}

\affiliation[4]{
	organization={University of Luxembourg},
	country={Luxembourg}
}


\begin{abstract}
	An ever-increasing number of critical infrastructures rely heavily on the assumption that security protocols satisfy a wealth of requirements.
	Hence, the importance of certifying \eg privacy properties using methods that are better at detecting attacks can hardly be overstated.
	This paper scrutinises the \enquote{unlinkability} privacy property using relations equating behaviours that cannot be distinguished by attackers.
	Starting from the observation that some reasonable design choice can lead to formalisms missing attacks, we draw attention to a classical concurrent semantics accounting for relationship between past events, and show that there are concurrency-aware semantics that can discover attacks on all protocols we consider.
	More precisely, we focus on protocols where trace equivalence is known to miss attacks that are observable using branching-time equivalences.
	We consider the impact of three dimensions: design decisions made by the programmer specifying an unlinkability problem (style), semantics respecting choices during execution (branching-time), and semantics sensitive to concurrency (non-interleaving), and discover that reasonable styles miss attacks unless we give attackers enough power to observe choices and concurrency.
	Our main contribution is to draw attention to how a popular concurrent semantics -- history-preserving bisimilarity -- when defined for the non-interleaving applied \(\pi\)-calculus, can discover attacks on all protocols we consider, regardless of the choice of style.
	Furthermore, we can describe all such attacks using a novel modal logic that is hence suitable to formally certify attacks on privacy properties.
	This study highlights the threats posed by relying exclusively on tools implementing coarser semantics for protocol verification, and justifies in a very precise sense why security practitioners should account for history between past events to build reliable tools.
\end{abstract}
\end{frontmatter}

\section{Introduction}
\label{sec:intro}

Underspecified or poorly implemented communication standards can lead to catastrophic failures, privacy breaches and other undesirable consequences.
Both business managers and protocol engineers are equally aware of the possibility that resulting vulnerabilities may be exploited, and are looking for formal guarantees of their non-existence.
Despite decades of work and numerous advances, completely foolproof methods to verify essential security and privacy properties are still sought after.
To discuss and address precisely this general problem, our paper focuses on one particular property, \enquote{unlinkability}, for one particular protocol, the Basic Access Control (BAC) protocol.
This well-studied property has seen variations in its representation, and our first contribution is to 
systematically identify and categorise those \enquote{stylistic differences} -- as explained in \autoref{sec:tech_preliminry_intro} -- to be able to adequately compare how different approaches fare in detecting attacks.
This foundational study allows us to sharpen and clarify our problem statement: 
\emph{Why have existing protocol verification tools produced false 
	positives (\ie  insecure protocols being verified as safe) and reached 
	fundamentally different conclusions about the privacy properties of 
	this protocol?}

Answering this question requires first to adopt a non-interleaving structural operational semantics, to represent specifications and systems in a concise, uniform and abstract way, as \emph{processes} (\aka \emph{agents}).
Our take on the applied $\pi$-calculus makes it easier to track \emph{events} and their origin locations, while maintaining usual features such as the ability to model Dolev-Yao attackers.
We further leverage (bi)simulations, used to represent competing positions in a game in the presence of intruders\footnote{%
	While simulation sets once and for all which process is leading the game, \emph{bi}simulation allows changes in who is leading. 
}\,: this classical tool equates processes that may be internally different, provided that they are indistinguishable from the attacker's perspective.
The classical methodology ensures that if a system and its specification, both represented as processes, cannot be distinguished by a (bi)simulation, then the property of interest is guaranteed, since the attacker cannot tell apart the implementation and the specification.
If, on the other hand, there exists a distinguishing strategy, then a modal logic capturing said (bi)simulation can be used to forge a concrete attack.
This tooling is the second variable of our method: different (bi)simulations will capture attackers with varying capacities, more distinguishing relations corresponding to more powerful attackers.
The third, potentially surprising, element required to conduct our study is the acknowledgement that seemingly unimportant stylistic differences in how specifications are represented can greatly impact existing protocol verification tools.
While it is agreed that those differences \emph{should not} matter, our method highlights precisely in which sense \emph{they do} for existing tools, resulting in protocols being flagged as safe or unsafe depending on the modelling technique chosen rather than due to features of the protocol itself.

Our first finding is that \emph{existing protocol verification tools miss attacks because they equate processes that attackers can actually distinguish}, and that \emph{this observation was obfuscated by minor stylistic differences}, resulting in tools disagreeing on the privacy provided by seemingly identical protocols.
Bridging this knowledge gap allows us to better navigate this research landscape and to improve our fundamental understanding of unlinkability.
Our second finding establishes that the so called \emph{history-preserving bisimilarity} adequately models attackers capable of forging attacks on all existing implementations by taking into account relations between past events.
This \enquote{just powerful enough} bisimulation allows us to ignore unimportant stylistic differences when modelling the unlinkability of security protocols.
To the best of our knowledge, this is the first time that a compelling security application has emerged for this relation, which was originally developed in the early days of concurrency theory as a process semantics to better capture concurrency.

Our contributions go beyond the case study of unlinkability for the BAC protocol by illustrating more generally how existing protocol verification tools are inherently limited by their design choices.
By illustrating both technically and intuitively how attackers can leverage information about past events to forge attacks that neither the so-called \emph{linear-time} nor \emph{interleaving} bisimilarities 
can represent, we also address an interesting challenge to the community: \emph{can existing tools adopt our \enquote{history preserving} view, so that they can  become insensible to stylistic changes?}.
We believe that computer security practitioners will find it critical to re-evaluate the design choices of their verification methods based on our findings, and sketch possible avenues to improve protocol verification tools so that they can account for attacks forged using knowledge about past events and their relations.

\paragraph{Structure of the paper}
\label{par:structure_of_paper}
In \autoref{sec:lit-review} we give a short literature review for the two main areas of this paper, namely for unlinkability and bisimulations of the type we employ; we relate to more specific works throughout the paper in the respective places where these works are relevant.
We then give in \autoref{sec:tech_preliminry_intro} a preliminary technical presentation of our results, meant for readers to understand key observations and results before consuming the details in later sections.
We revisit in \autoref{sec:SOS} the non-interleaving semantics for the applied $\pi$-calculus~\cite{Aubert2022e}.
\autoref{sec:interleaving} presents formal machinery and results 
explaining the effectiveness of interleaving semantics when evaluating unlinkability,
in particular, by proving which formulations of unlinkability 
require $i$-simulation and which require more distinguishing power 
to discover attacks.
The reader just wishing to know the most novel highlights
can skip forward to \autoref{sec:HP} 
where HP-simulation (\autoref{def:hp-sim}),
and its characteristic modal logic (\autoref{fig:modal-2}),
are applied in the proof of
\autoref{thm:min-HP}
to obtain a description of an attack on unlinkability which 
can not be discovered using
semantics less powerful than HP-bisimilarity.
\autoref{sec:protocols} serves to explain
further the context of the contributions in \autoref{sec:interleaving} and \ref{sec:HP},
by tightening the scope of these observations.
These results are finally discussed in a cohesive manner in \autoref{sec:discussion}, wrapping up the multiple results of this paper in a cohesive final picture.

\section{Literature Review}
\label{sec:lit-review}
We provide a brief review of the literature on the two main topics of
this paper: unlinkability of the BAC protocol and the types of
bisimulation used for protocol verification. More specific related
works are discussed in context, at the points in the paper where they
are most relevant.

\subsection{Unlinkability of the BAC protocol}

The Basic Access Control (BAC) protocol, as defined in the ICAO 9303 standard~\cite{international1996machine}, describes how machine-readable passports should be implemented and their communication protocols.
Implementations are supposed to enforce various privacy properties~\cite{10.1145/2825026}, among which \emph{unlinkability}~\cite[14.3]{CC}.
This particular property has been studied with a wealth of formal
methods tools, including process algebras~\cite{Arapinis2010}, the ProVerif cryptographic protocol verifier~\cite{Hirschi2016,Hirschi2019}, DEEPSEC~\cite{Cheval2018}, or bisimilarity equivalence technique~\cite{Filimonov2019,Horne2021}. 
Other tools of interest to study protocols formally include Tamarin~\cite{10.1145/2810103.2813662}, Maude-NPA~\cite{10.1007/978-3-319-11851-2_11}, Akiss~\cite{Chadha2016}, SAPIC~\cite{Kremer2016}, SAT-Equiv~\cite{Cortier2017} and SPEC~\cite{Tiu2016}.

\subsection{Branching-time and non-interleaving behavioural equivalences}

The study of the discriminating power of behavioural equivalences \wrt security analysis has been an active research subject for nearly three decades, since the landmark studies on \enquote{the linear time-branching time spectrum}~\cite{Glabeek1989,Glabbeek2001,Glabbeek1993}.
In the security domain, interleaving bisimilarity has been proposed to establish non-interference properties~\cite{DBLP:conf/fosad/FocardiG00,doi:10.3233/JCS-1994/1995-3103,DBLP:conf/apn/Gorrieri20,Ryan2001}.
Prior to the current work, coarser equivalences than bisimilarity have been sufficient to express attacks on security properties of protocols~\cite{DBLP:conf/fosad/FocardiGM02}.
Prior work on vulnerabilities in the BAC protocol using the applied \(\pi\)-calculus~\cite{Abadi2018,Aubert2022e} 
required only the power of similarity~\cite{Filimonov2019,Horne2021}.

ST-bisimilarity~\cite{Johansen2016} and HP-bisimilarity~\cite{Bednarczyk1991,Rabinovich1988,Aubert2020b} have also been redefined for the applied \(\pi\)-calculus and conjectured to be a useful toolbox to study security properties of protocols~\cite{Aubert2022k}. Non-interleaving semantics have previously been proposed for investigating non-interference~\cite{DBLP:conf/apn/Gorrieri20}. The current work is however the first to apply such non-interleaving semantics to a formulation of unlinkability of a protocol, for which HP-bisimilarity finds attacks that interleaving bisimilarity misses.

\section{Technical Preliminaries}
\label{sec:tech_preliminry_intro}

Unlinkability is a strong privacy property,
\textcquote[14.3]{CC}{\emph{\textins{ensuring} that a user may make multiple uses of resources or services without others being able to link these uses together}}.
In the context of ePassports~\cite{ICAO}, unlinkability ensures that an ePassport holder may make use of their ePassport multiple times without a third party
being able to link these uses together.
Unlinkability has been studied intensely, particularly in the setting of the Basic Access Control (BAC) protocol, which is used to authenticate ePassports when passing a passport control reader machine~\cite{Arapinis2010,Hirschi2016,Cheval2018,Hirschi2019,Filimonov2019,Horne2021}. Such results have, however, a much wider application than BAC, \eg to the more recent PACE protocol for ePassports, or to RFID protocols. Since the BAC protocol has been the subject of quite some debate, we use it to illustrate our motivation.

The BAC protocol involves a prover (\ie an ePassport)
and a verifier (\ie an ePassport reader).
An established approach to modelling an unlinkability property
is to check whether a system that allows the same prover to be used
multiple times (\eg due to passing through multiple checkpoints)
is equivalent to a specification where each prover is used at most
once.

The verification expert may choose to work with different equivalences, 
reflecting the observational powers of an external attacker.
One distinction is between trace equivalences and (bi)similarities, both being well studied verification methods implemented in a multitude of tools.

In 2010, the applied $\pi$-calculus was used~\cite{Arapinis2010} to model and prove a notion of unlinkability of the BAC protocol based on bisimilarity, but a bug in the ProVerif tool used for this purpose rendered the proof unacceptable. 
This original formulation also made use of a constant \enquote{$\getchallenge$} message at the beginning of the protocol to mark the creation of a new session.
In 2016, a proof was provided for unlinkability without such a \enquote{$\getchallenge$} message~\cite{Hirschi2016,Hirschi2019}, but that proof makes use of trace equivalence instead of bisimilarity as previously done~\cite{Arapinis2010}.
In 2018, 
an attack was discovered using the DEEPSEC tool~\cite{Cheval2018}
on a scheme for unlinkability
differing from the one described above~\cite[Sect.~6]{Horne2021},
and, importantly,
for a formulation of unlinkability where only a finite number of
sessions are allowed.
In 2019, practical attacks were discovered~\cite{Filimonov2019} using the original 2010 bisimilarity-based formulation of unlinkability.
More recently, cleaner
attack formulas have been exposed~\cite{Horne2021}, leveraging a formulation that makes use of a fresh channel for each endpoint 
(\ie each participant in each session).

What may be surprising in the timeline above
is that a proof of unlinkability in 2016 appears to be contradicted by the attack discovered in 2019.
The decisive difference between the 2016 and 2019 results is 
the use of trace equivalence rather than the classic bisimilarity (which we call interleaving bisimilarity, and abbreviate as $i$-bisimilarity), as recently discussed~\cite{Filimonov2019,Horne2021}.
That is, the difference is in the so-called linear-time/branching-time spectrum~\cite{glabbeek90concur}.
The essential role of branching-time comes from the fact that attack strategies involve a choice between receiving different inputs
at a particular moment 
by some participant in the protocol (usually the ePassport).
Moreover, interleaving similarity ($i$-similarity), which is the one-sided relaxation of $i$-bisimilarity where the same
player always leads throughout the 
 game, suffices to discover this attack for most formulations of unlinkability,
since the above-mentioned choice between inputs can be dictated by one player.

\paragraph{Synopsis}
We work in the setting of applied $\pi$-calculus, whose syntax is defined in \autoref{figure:syntax} on page~\pageref{figure:syntax}.
We appeal to the reader's patience to allow us to use this syntax 
in the rest of this technical preliminary to explain the three stylistic choices distilled from related work.
These stylistic differences are
summarised in \autoref{table:stylistic-diff},
where we indicate the presence of the attack with respect to $i$-similarity,
regardless of the equivalence used in the respective paper.
The condensed information from \autoref{table:stylistic-diff} is explained in the rest of this synopsis.

\begin{table}[t]
\caption{
In this table, we fix $i$-simulation as 
the distinguishing equivalence employed,
and summarise how stylistic differences impact the ability to find attacks.
}
\label{table:stylistic-diff}
\centering
\begin{tabular}{|l|l|l|l|l|l}
	Style \& inspiration
	& First message & Channels & Bounded & Attack 
	\\ \hline
	Minimal (min)~\cite{Hirschi2019}
	&
	nonce
	&
	single
	&
	no
	&
	\textbf{no}
	\\
	W/ $\getchallenge$ ($\getchallenge$)~\cite{Arapinis2010}
	&
	\textbf{constant $\getchallenge$}
	&
	single
	&
	no
	&
	yes
	\\
	W/ endpoints ($\ch$)~\cite{Horne2021}
	&
	nonce
	&
	\textbf{endpoints}
	&
	no
	&
	yes
	\\
	Bounded ($2$)~\cite{Cheval2018}
	&
	nonce
	&
	single
	&
	\textbf{yes}
	&
	yes
\end{tabular}
\end{table}

As standard when modelling protocols in the applied $\pi$-calculus, the control flow of inputs and outputs that an ePassport (the prover) and the ePassport reader (the verifier) perform are represented as processes denoted as $\MainUK(c,d,ke,km)$ and $\Reader(c,d,ke,km)$, respectively. 
We delay the definition of these processes until later in the paper on
page~\pageref{eq:def-prover},
since we aim to draw attention to the various definitions of unlinkability here, with the BAC protocol serving merely as a well-known running example. It is sufficient at this point to know that each represents the inputs and outputs of one participant within one session and the following two things about the parameters. 
Firstly, $c$ represents a channel used for all outputs within the given session of the prover or verifier, while $d$ is used for all inputs. 
Secondly, there are two long-term keys $ke$ and $km$, representing keys built into the ePassport (the prover),
the same two keys also appear in verifier processes 
into which an ePassport holder has loaded these keys (typically, via an OCR session that the attacker cannot intercept).

The various models of unlinkability can then be captured by 
an equivalence problem between a system and specification as in
\autoref{thm:system-min-interleaving} below, which, like all the other theorems stated in this preliminary, is proven later in the paper.
Here the system allows the same ePassport to be used 
against multiple verifiers, all sharing the same long-term keys.
The innermost replication (\eg the \enquote{\(\bang\)} symbol) captures the generation of multiple
instances of sessions using the same ePassport, and hence the same
long-term keys.
This is depicted in \autoref{fig:idea}, where
the left side of the choice in the System
depicts a single session with one ePassport used
across session with two verifiers, whereas the right side of the
choice indicates that it is also possible that two distinct ePassports
are used, one for each verifier.
The \enquote{\(\sim\)} symbol between them indicates that the unlinkability property holds if they are bisimilar, \eg they cannot be distinguished using a bisimilarity.
At the same time, the Specification explicitly supposes a scenario where each ePassport is only ever
used once, 
\ie against a single verifier, 
in one session (as the inner replication operator is missing) and hence is certainly impossible to track.
This formulation
from~\cite{Hirschi2016,Hirschi2019}, with neither a leading $\getchallenge$
nor fresh channel name for each of the two participants in each session (called endpoints), 
has no attack with respect to $i$-similarity (denoted $\simi{i}$
), \ie the two processes are equated.

\begin{figure}[t]
	{
		\centering
		\includegraphics[width=.9\linewidth]{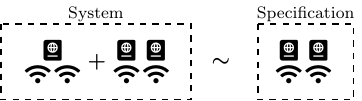}
	}
	\caption{The unlinkability property, illustrated using ePassports (top) and readers (bottom).}
	\label{fig:idea}
\end{figure}

\begin{thm}[Minimal]
	\label{thm:system-min-interleaving}
$
\SystemMin
\simi{i}
\SpecMin
$, 
where
\[
\begin{array}{rl}
\SystemMin \triangleq & \bang\mathopen{\nu ke, km.\bang}\left( \Reader(c,d,ke,km) \cpar \MainUK(c,d,ke,km) \right)
\\\\
\SpecMin \triangleq & \bang\mathopen{\nu ke, km.}\left( \Reader(c,d,ke,km) \cpar \MainUK(c,d,ke,km) \right)
\end{array}
\]
\end{thm}
The key observation behind the above theorem
is that the input and output actions of 
any successfully authenticating ePassport in the system,
can be imitated in the specification by starting a fresh reader session, which has an indistinguishable communication pattern from the perspective of an attacker.
The formal proof constructs a relation defining precisely this strategy.

The original 2010 formulation~\cite{Arapinis2010} adds a special message, a \enquote{$\getchallenge$} constant, sent by the reader and received by the ePassport,
as indicated in red below (compare the difference in process terms with those in \autoref{thm:system-min-interleaving}).

\begin{restatable}[With \enquote{$\getchallenge$}]{thm}{systemget}
	\label{thm:system-get}
$\SystemCSF
\nsimi{i}
\SpecCSF
$,  where
\[
\begin{array}{rl}
	\SystemCSF \triangleq & \bang\mathopen{\nu ke, km.\bang}\left( {\color{red} \cout{c}{\getchallenge}.}\Reader(c,d,ke,km) \cpar  {\color{red} \mathopen{d(x).}
		\match{ x = \getchallenge }}\MainUK(c,d,ke,km) \right)
	\\[1.6em]
	\SpecCSF \triangleq & \bang\mathopen{\nu ke, km.}\left( {\color{red} \cout{c}{\getchallenge}.}\Reader(c,d,ke,km) \cpar  {\color{red} \mathopen{d(x).}
		\match{ x = \getchallenge }}\MainUK(c,d,ke,km) \right)
\end{array}
\]
\end{restatable}

The above theorem states that $i$-similarity distinguishes these processes.
The trick is that the initial message acts as a proxy for distinguishing between an existing session and the creation of a new session.
This approach turns out to be sufficient to reveal attacks in the BAC protocol, although it does not make it clear which session each message belongs to.

The 2021 formulation~\cite{Horne2021} uses a fresh 
channel
for each session participant to model explicitly that an attacker can tell the difference between two concurrent sessions each with a separate verifier. Such an attack is demonstrated to be practical, 
\eg  by standing next to the verifier device or observing the signature of the underlying wireless transport protocol session.
This stylistic difference in the formulation of the BAC protocol is again highlighted in red below (compared with \autoref{thm:system-min-interleaving}).
Intuitively, the channel $c$ is new for each participant in each session, and this is being sent out (on the two channels $\reader$ and $\passport$) for anyone that wants to communicate with the respective participant (hence the parametrization of both processes $\Reader$ and $\MainUK$ with the new channel name $c$).

\begin{restatable}[With endpoints]{thm}{systemend}
\label{thm:system-end}
$\SystemUK \nsimi{i} \SpecUK$, where
\[
\begin{array}{rl}
\SystemUK \triangleq & \bang\mathopen{\nu ke,  km.}\bang\left( {\color{red} \nu c.\cout{\reader}{c}.}\Reader({\color{red}c},{\color{red}c},ke,km) \cpar {\color{red} \nu c.\cout{\passport}{c}.}\MainUK({\color{red}c},{\color{red}c},ke,km) \right)
\\[1.7em]
\SpecUK \triangleq & \bang\mathopen{\nu ke,  km.}\left( {\color{red} \nu c.\cout{\reader}{c}.}\Reader({\color{red}c},{\color{red}c},ke,km) \cpar {\color{red} \nu c.\cout{\passport}{c}.}\MainUK({\color{red}c},{\color{red}c},ke,km) \right)
\end{array}
\]
\end{restatable}

Again, $i$-similarity is enough to discover attacks on this variant of BAC.
Since similarity is coarser than bisimilarity, they are certainly not related by $i$-bisimilarity, denoted $\bisimi{i}$.
To present these attacks we will define one modal logic as the testing framework for each of the (bi)similarities we consider, and then give on page~\pageref{proof:systemend} the modal formula that represents the attack (\ie the strategy that the system can execute, but that the specification cannot).

The clash between \autoref{thm:system-min-interleaving} and Theorems~\ref{thm:system-get} and \ref{thm:system-end}
is awkward from a protocol modelling perspective, since $\getchallenge$
is not really part of the cryptographic protocol, exactly like the artificial endpoints from \autoref{thm:system-end}.
We would expect to be able to drop such leading constant messages or fresh channels, considering them to be just
a part of the underlying setup.

Somewhat more curious is that when bounding the number of sessions, an attack can be discovered, as opposed to the formulation of unlinkability in \autoref{thm:system-min-interleaving},
where neither constant messages nor fresh endpoints were employed.
This fact is captured by the theorem below, which allows at most 2 ePassport and 2 reader sessions by dropping the replication operators.
In the case of the specification, both sessions use different keys,
whereas for the system it is possible to use the same keys in multiple prover or verifier sessions.

\begin{restatable}[Bounded]{thm}{systemfinite}
	\label{thm:system-finite}
	$\SystemF
	\nsimi{i}
	\SpecF
	$,
	where
	\[
	\begin{array}{rl}
		\SystemF \triangleq & \mathopen{\nu ke_1, km_1, ke_2, km_2.}\big( \\
		\MoveEqLeft[-2] \left( \Reader(c,d,ke_1,km_1) + \Reader(c,d,ke_2,km_2) \right) \cpar
		\\\MoveEqLeft[-2]
		\left( \MainUK(c,d,ke_1,km_1) + \MainUK(c,d,ke_2,km_2) \right) \cpar\\
		\MoveEqLeft[-2] \left( \Reader(c,d,ke_1,km_1) + \Reader(c,d,ke_2,km_2) \right) \cpar\\
		\MoveEqLeft[-2] \left( \MainUK(c,d,ke_1,km_1) + \MainUK(c,d,ke_2,km_2) \right) \big)\\[1.5em]
		\SpecF \triangleq & \mathopen{\nu ke_1, km_1, ke_2, km_2.}\big( \\
		\MoveEqLeft[-2] \Reader(c,d,ke_1,km_1) \cpar \MainUK(c,d,ke_1,km_1)  \cpar \\
		\MoveEqLeft[-2] \Reader(c,d,ke_2,km_2) \cpar \MainUK(c,d,ke_2,km_2) \big)
	\end{array}
	\]
\end{restatable}

We consider this clash between Theorems~\ref{thm:system-min-interleaving} and \ref{thm:system-finite}
difficult to accept.
Moreover, it makes the verification task difficult, if not impossible, as one needs to decide whether to write models using the replication (going for the unbounded case) or to write many versions of the protocol, one for each number of sessions, where the above theorem is only the case for 2.

In \autoref{sec:HM} we will give instantiations for the Verifier and Prover
processes and provide proofs for Theorems \ref{thm:system-get},
\ref{thm:system-end}, and \ref{thm:system-finite}.

\paragraph{Contributions}
In this work, we take two steps in order to obtain a more coherent picture regarding branching-time attacks,
independently of the style of modelling unlinkability.
Firstly, we observe that, while $i$-similarity could not detect attacks on the minimal style employed in \autoref{thm:system-min-interleaving},
if instead we employ $i$-bisimilarity while keeping the style the same, then attacks on BAC can be detected.
This is a contribution of this work, since we are not aware of any formulation of a privacy problem in the literature for which it is proven that bisimilarity is necessary and similarity is insufficient.

Going further, for a more comprehensive methodology,
we show that it is necessary to 
look beyond even $i$-bisimilarity
into the non-interleaving spectrum.
The attack patterns explored in this work, demanding branching-time equivalences,
are not unique to the BAC protocol.
They can be applied across the board to many authentication protocols, such as the more recent PACE protocol for ePassports, as explained in~\cite{Horne2021}.
In this work, we draw attention to a more minimal historical predecessor of the BAC protocol
devised for RFID tags called the Feldhofer protocol~\cite{Feldhofer2004}, where $i$-similarity discovers attacks in line with Theorems~\ref{thm:system-get}--\ref{thm:system-finite},
yet, surprisingly, 
even 
$i$-bisimilarity
cannot discover attacks for a minimal formulation of unlinkability.

Our second contribution to privacy is to show that a non-interleaving equivalence, history-preserving bisimilarity (denoted $\bisimi{HP}$),
is however sufficient to reveal the unlinkability attacks on the Feldhofer protocol without making use of any of the three styles of protocol manipulation that Theorems~\ref{thm:system-get}--\ref{thm:system-finite} employ.
Thus our main aim is to resolve the observed discrepancy between the
various theorems, leading to the following revision of 
\autoref{thm:system-min-interleaving} which states that
history-preserving bisimilarity can indeed distinguish the two
processes without recourse to specific stylistic features.

\begin{restatable}{thm}{systemHP}
\label{thm:min-HP}
$
\SystemMin
\nbisimi{HP}
\SpecMin
$, where $\SystemMin$ and $\SpecMin$ as in \autoref{thm:system-min-interleaving}.
\end{restatable}
Furthermore, the above theorem applies
to all protocols we have considered, notably Feldhofer and BAC,
suggesting that HP-bisimilarity is a robust choice of semantic equivalence for 
unlinkability problems. 

The body of this paper is dedicated to introducing and explaining these notions of equivalence
and evaluating them on the examples from this preliminary.
The formal machinery introduced is, of course, a third contribution,
notably we introduce a novel modal logic for describing attacks whenever unlinkability
fails with respect to HP-bisimilarity.

\section{Background: a non-interleaving structural operational semantics}
\label{sec:SOS}

This section reviews the definitions of non-interleaving applied $\pi$-calculus~\cite{Aubert2022e} necessary to understand the current paper, particularly 
how its structural operational semantics supports a notion
of independence (\ie concurrency) that we use to define the
various non-interleaving equivalences studied in this paper.

\subsection{Preliminaries: syntax, located aliases and events}

Notions for manipulating syntax with name binders are standard, hence presented in brief.
The novelties are the notion of \emph{located alias} in \autoref{def:syntax}, which is an alias variable prefixed by a \emph{location} string representing a component in a parallel composition, and \emph{location labels}, which are used to define \emph{events} in \autoref{def:location-and-events}.
In addition, we give our rationale for selecting a minimal notion of structural congruence. 

\begin{figure}
	\begin{multicols}{2}
	\begin{align*}
		\shortintertext{\textsc{Processes:}}
		P, Q \Coloneqq &&& 0 \tag*{deadlock} \\
		\mid &&& \mathopen\nu x. P \tag*{new} \\
		\mid &&& P \cpar Q \tag*{parallel} \\
		\mid &&& G \tag*{guarded process} \\
		\mid &&& \bang P \tag*{replication} \\
		\mid &&& \mathopen{\left[M = N\right]}P \tag*{match} 
	\end{align*}
	
	\begin{align*}
		\shortintertext{\textsc{Guarded processes:}}
		G, H \Coloneqq&&& \cin{M}{x}.P \tag*{input prefix} \\
		\mid &&&  \cout{M}{N}.P \tag*{output prefix} \\
		\mid &&& \mathopen{\left[M = N\right]}G \tag*{match} \\
		\mid &&& G + H \tag*{choice} \\
		\mid &&& \mathopen\nu x. G \tag*{new} 
	\end{align*}
	
	\begin{align*}
		\shortintertext{\rulel{Extended processes:}}
		A, B \Coloneqq &&& \sigma \cpar P \tag*{active process} \\ 
		\mid &&& \mathopen{\nu x.} A\tag*{new}
	\end{align*}
	
	\begin{align*}
		\shortintertext{\textsc{Messages:}  for \(f\in \Sigma\),}
		M, N \Coloneqq &&& x \tag*{variable} \\
		\mid &&&  \alia \tag*{alias} \\
		\mid &&& f(M_1, \hdots, M_n) \tag*{function} 
	\end{align*}
	
	\begin{align*}
		\shortintertext{\textsc{Early action labels:}}
		\pi  \Coloneqq &&&  M\,N\tag*{free input} \\ 
		\mid &&& \co{M}(\alia)\tag*{output} \\
		\mid &&& \tau\tag*{interaction} 
	\end{align*}
\end{multicols}
	\caption{Syntax of extended processes with guarded choices.
	}\label{figure:syntax}
\end{figure}

\begin{defi}[Syntax]
	\label{def:syntax}
All variables $x$, $y$ and $z$ are in the same syntactic category, but are
distinct from \emph{located aliases}, which are ranged over by $\alia$, $\alib$,
$\alic$. A located alias consists of an alias variable, say $\lambda$,
prefixed with a string $\prefix{s} \in \left\{ \prefix{0}, \prefix{1} \right\}^*$,
\ie $\alia = \prefix{s}\lambda$.
Messages, processes, and extended processes are defined according to the grammar in \autoref{figure:syntax}, given a signature of function symbols $\Sigma$, and omitting the deadlock process when guarded by a prefix (\eg we write \(\cout{M}{x}\) for \(\cout{M}{x}.0\)).
Sequences of names $\nu \vec{x}. P$ abbreviate multiple name binders defined inductively such that $\nu \epsilon. P = P$ and $\nu x, \vec{y}. P = \nu x. \nu \vec{y}.P$, where $\epsilon$ is the empty sequence (later on also used to denote the empty string).

\emph{Substitutions} $\sigma$, $\theta$ or $\rho$ are functions on located aliases, whose  
domains and ranges ignore elements mapped to themselves: $\dom{\sigma} = \left\{ \alia \colon \sigma(\alia) \neq \alia \right\}$ and $\ran{\sigma} = \left\{\sigma(\alia) \colon \alia \in \dom{\sigma} \right\}$.
For \eg \(\sigma(\prefix{0}x) = \prefix{1}y\) and \(\sigma (\alib) = \alib\) if \(\alib \neq \prefix{0}x\) we have \(\dom{\sigma} = \{\prefix{0}x\}\) and \(\ran{\sigma} = \{\prefix{1}y\}\) and write \(\sigma = \sub{\prefix{1}y}{\prefix{0}x}\).
The \emph{identity substitution} is denoted \(\id\), composition $\sigma \circ \theta$, and we write substitutions in suffix form, \eg $(x \sigma)\theta = x (\sigma \circ \theta) = \theta(\sigma(x))$.

The substitutions appearing in \emph{extended processes}
are called \emph{active substitutions} and map aliases in their finite domain to messages containing no aliases.
We assume a \emph{normal form} for extended processes, where aliases do not appear in processes.
\end{defi}

Working with applied $\pi$-calculus 
has the advantage of being close to languages used in tools such as ProVerif~\cite{Blanchet2016}, DEEPSEC~\cite{Cheval2019}, Akiss~\cite{Chadha2016}, SAPIC~\cite{Kremer2016}, SAT-Equiv~\cite{Cortier2017} and SPEC~\cite{Tiu2016}.
We however work with guarded choices.

We now define some syntactical notions to manipulate variables, and give some intuitions on how to read messages after \autoref{figure:active}.

\begin{defi}[Freshness, $\alpha$-conversion, etc.]
	A variable \(x\) (\resp an alias \(\alia\)) is \emph{free in a message \(M\)} if \(x \in \fv{M}\) (\resp \(\alia \in \fal{M}\)) for
	\begin{align*}
	\fv{f(M_1, \ldots M_n)} &= \cup_{i = 1}^n \fv{M_i}
	&&&
	\fv{x} &= \left\{ x \right\}
	&&&
	\fv{\alia} &= \emptyset
	\\
	\fal{f(M_1, \ldots M_n)} &= \cup_{i = 1}^n \fal{M_i}
	&&& 
	\fal{x} &= \emptyset
	&&&
	\fal{\alia} &= \left\{ \alia \right\}.
	\end{align*}
	The \(\mathrm{fv}\) function extends in the standard way to (extended) processes, letting $\fv{\nu x.P} = \fv{P}\setminus\left\{x\right\}$
	and $\fv{M(x).P} = \fv{M} \cup \left(\fv{P}\setminus\left\{x\right\}\right)$, and similarly for \(\fv{\nu x.A} = \fv{A}\setminus\left\{x\right\}\) and $\fv{\sigma \cpar P} = \fv{\ran{\sigma}} \cup \fv{P}$.
The functions for free variables and free aliases extend to labels as follows:
\begin{gather*}
\fv{M\,N} = \fv{M}\cup\fv{N} 
\quad
\fv{\co{M}(\alia)} = \fv{M} 
\quad
\fv{\tau} = \emptyset
\\
\fal{M\,N} = \fal{M}\cup\fal{N} 
\quad
\fal{\co{M}(\alia)} = \fal{M} 
\quad
\fal{\tau} = \emptyset.
\end{gather*}

	We say a variable $x$ is \emph{fresh for a message \(M\) (\resp process \(P\), extended process \(A\))}, written \(\isfresh{x}{M}\) (\resp \(\isfresh{x}{P}\), \(\isfresh{x}{A}\)) 	whenever $x \notin \fv{M}$ (\resp \(x \notin \fv{P}\), \(x \notin \fv{A}\)), and similarly for aliases.
	Freshness extends point-wise to lists of entities, \ie $\isfresh{x_1,x_2, \ldots x_m}{M_1, M_2, \ldots, M_n}$, denotes the conjunction of all $\isfresh{x_i}{M_j}$ for all $1 \leq i \leq m$ and $1 \leq j \leq n$.

	We define \emph{$\alpha$-conversion} (denoted $\equiv_\alpha$) for variables only (not aliases, which are fixed \enquote{addresses}) as the least congruence 
	such that,
	whenever $\isfresh{z}{\nu x.P}$, we have
	$\nu x.P \mathrel{\equiv_\alpha} \nu z.(P\sub{x}{z})$
	and $M(x).P \mathrel{\equiv_\alpha} M(z).(P\sub{x}{z})$.
	Similarly, for extended processes, we have the least congruence such that, whenever $\isfresh{z}{\nu x.A}$, we have
	$\nu x.A \mathrel{\equiv_\alpha} \nu z.(A\sub{x}{z})$.
	Restriction is such that $\sigma\mathclose{\restriction_{\vec{\alpha}}}(x) = \sigma(x)$ if $x \in \vec{\alpha}$ and $x$ otherwise.

	\emph{Capture-avoiding substitutions} are defined for processes such that
	\begin{align*}
		(M(x).P)\sigma \mathrel{\equiv_\alpha} M\sigma(z).P\sub{x}{z}\sigma && \text{ and } && (\nu x.P)\sigma \mathrel{\equiv_\alpha} \nu z.P\sub{x}{z}\sigma
	\end{align*} for $\isfresh{z}{\ran{\sigma}, \nu x.P}$.
	For extended processes, it is defined such that 
	\begin{align*}
		(\nu x.A)\mathclose{\sigma} \mathrel{\equiv_\alpha} \nu z.(A\mathclose{(\sub{x}{z} \circ \sigma)}) && \text{ and } && (\theta \cpar P)\sigma = ({\theta \circ \sigma}\mathclose{\restriction_{\dom{\theta}}} \cpar P\sigma)
	\end{align*}
	for $\isfresh{z}{\ran{\sigma}, \nu x.A}$.
	A substitution $\sigma$ extends to labels such that $(M\,N)\sigma = M\sigma\,N\sigma$ and $(M(\alia))\sigma = M\sigma(\alia\sigma)$, and $\tau\sigma = \tau$.
\end{defi}

{
\newcommand{\tablespace}{\\[3.4em]} 
	\begin{figure*}
	\begin{tabular}{c@{\hskip 2.5em} c}
		\begin{prooftree}
			M \mathrel{=_E} K
			\qquad
			\justifies
			\names{\env}{
				\mathopen{\cin{K}{x}.}P 
				\dlts{M\,N}{
					[]
				}
				{\id \cpar P\sub{x}{N}} }
			\using
			\mbox{\textsc{Inp}}
		\end{prooftree}
		&
		\begin{prooftree}
			M \mathrel{=_E} K
			\justifies
			\names{\env}{
				\cout{K}{N}.P 
				\dlts{\co{M}(\aliroot)}{
					[]
				}
				{\sub{\aliroot}{N}} \cpar P }
			\using
			\mbox{\textsc{Out}}
		\end{prooftree}
		\tablespace
		\begin{prooftree}
			\names{\env}{
				P}
			\dlts{\pi}{u}
			\mathopen{\nu \vec{x}.}\left( \sigma \cpar R \right)
			\quad
			\isfresh{\vec{x}}{Q}
			\justifies
			\names{\env}{
				P \cpar Q
				\dlts{\pi}{\prefix{0}u}
				\mathopen{\nu \vec{x}.}\left( \sigma \cpar R \cpar Q \right)
			}
			\using
			\mbox{\textsc{Par-L}}
		\end{prooftree}
		&
		\begin{prooftree}
			\names{\env}{
				Q}
			\dlts{\pi}{u}
			\mathopen{\nu \vec{x}.}\left( \sigma \cpar R \right)
			\quad
			\isfresh{\vec{x}}{P}
			\justifies
			\names{\env}{
				P \cpar Q
				\dlts{\pi}{\prefix{1}u}
				\mathopen{\nu \vec{x}.}\left( \sigma \cpar P \cpar R \right)
			}
			\using
			\mbox{\textsc{Par-R}}
		\end{prooftree}
		\tablespace
		\begin{prooftree}
			\names{\env}{
				P\sub{x}{z}
				\dlts{\pi}{
					u 
				}
				A
			}
			\qquad
			\isfresh{z}{\pi, \nu{x}.{ P }}
			\justifies
			\names{\env}{
				\nu{x}.{ P }
				\dlts{\pi}{
					u 
				}
				\mathopen{\nu{z}.}A
			}
			\using
			\mbox{\rulel{Ext}}
		\end{prooftree}
		&
		\begin{prooftree}
			\names{\env}{
				A
				\dlts{\pi}{
					u 
				}
				B
			}
			\qquad
			\isfresh{x}{{\pi} }
			\justifies
			\names{\env}{
				\nu{x}.{ A }
				\dlts{\pi}{
					u 
				}
				\mathopen{\nu{x}.}B
			}
			\using
			\mbox{\rulel{Res}}
		\end{prooftree}
		\tablespace
		\multicolumn{2}{c}{
			\begin{prooftree}
				\names{\env}{
					P}
				\dlts{\co{M\sigma}(\lambda)}{\prefix{s}[s']} 
				\mathopen{\nu \vec{x}.}\left( {\sub{\lambda}{N}} \cpar Q \right)
				\quad
				\isfresh{\vec{x}}{\ran{\sigma}}
				\quad
				\fal{M} \subseteq \dom{\sigma}
				\quad
				\isfresh{\prefix{s}\lambda}{\dom{\sigma}}
				\justifies
				\names{\env}{
					\sigma \cpar P}
				\dlts{\co{M}(\prefix{s}\lambda)}{\prefix{s}[s']}
				\mathopen{\nu \vec{x}.}\left( \sigma \circ {\sub{\prefix{s}\lambda}{N}} \cpar Q \right)
				\using
				\rulel{Alias-out}
		\end{prooftree}}
		\tablespace
		\multicolumn{2}{c}{
			\begin{prooftree}
				\names{\env}{
					P}
				\dlts{\pi\sigma}{u} 
				\mathopen{\nu \vec{x}.}\left( \id \cpar Q\right)
				\quad
				\isfresh{\vec{x}}{\ran{\sigma}}
				\quad
				\fal{\pi} \subseteq \dom{\sigma}
				\justifies
				\names{\env}{
					\sigma \cpar P}
				\dlts{\pi}{u}
				\mathopen{\nu \vec{x}.}\left( \sigma \cpar Q \right)
				\using
				\rulel{Alias-free}
			\end{prooftree}
		}
		\tablespace
		\multicolumn{2}{c}{
			\begin{tabular}{c@{\hskip 2em} c@{\hskip 2em} c}
				\begin{prooftree}
					\names{\env}{ G } \dlts{\pi}{[t]} A
					\justifies
					\names{\env}{ G + H } \dlts{\pi}{[0t]} A
					\using
					\mbox{\rulel{Sum-L}}
				\end{prooftree}
				&
				\begin{prooftree}
					\names{\env}{ H } \dlts{\pi}{[t]} A
					\justifies
					\names{\env}{ G + H } \dlts{\pi}{[1t]} A
					\using
					\mbox{\rulel{Sum-R}}
				\end{prooftree}
		& 
		\begin{prooftree}
			\names{\env}{ G } \dlts{\pi}{u} A
			\qquad
			M \mathrel{=_E} N
			\justifies
			\names{\env}{ {\mathopen{\left[M=N\right]}{G} }} \dlts{\pi}{u} 
			{A}
			\using
			\mbox{\textsc{Mat}}
		\end{prooftree}
			\end{tabular}
		}
		\tablespace
		\multicolumn{2}{c}{
				\begin{prooftree}
					\names{\env}{  P \cpar \bang P 
						\dlts{\pi}{u}
						A }
					\justifies
					\names{\env}{ \bang P 
						\dlts{\pi}{u}
						A }
					\using
					\mbox{\textsc{Bang}}
				\end{prooftree}
		}
		\tablespace
		\multicolumn{2}{c}{
			\begin{prooftree}
				\names{\env}{ P }
				\dlts{\co{M}(\aliroot)}{\ell_0} 
				\mathopen{\nu \vec{y}.}\left( {\sub{\aliroot}{N}} \cpar P' \right)
				\qquad
				\names{\env}{  Q }
				\dlts{M\,N}{\ell_1}
				\mathopen{\nu \vec{w}.}\left(\id \cpar Q' \right)  
				\qquad
				\isfresh{\vec{y}}{Q}
				\qquad
				\isfresh{\vec{w}}{P, \vec{y}}
				\justifies
				\names{\env}{  P \cpar Q}
				\dlts{\tau}{(\prefix{0}\ell_0, \prefix{1}\ell_1)}
				\mathopen{\nu \vec{y}, \vec{w}.}\left(\id \cpar P' \cpar Q' \right) 
				\using
				\mbox{\rulel{Close-L}}
			\end{prooftree}
		}
		\tablespace
		\multicolumn{2}{c}{
			\begin{prooftree}
				\names{\env}{ P }
				\dlts{M\,N}{\ell_0}
				\mathopen{\nu \vec{y}.} \left(\id \cpar P' \right)  
				\qquad
				\names{\env}{  Q }
				\dlts{\co{M}(\aliroot)}{\ell_1} 
				\mathopen{\nu \vec{w}.}\left(  {\sub{\aliroot}{N}} \cpar Q' \right)
				\qquad
				\isfresh{\vec{w}}{P}
				\qquad
				\isfresh{\vec{y}}{Q, \vec{w}}
				\justifies
				\names{\env}{  P \cpar Q}
				\dlts{\tau}{(\prefix{1}\ell_1, \prefix{0}\ell_0)}
				\mathopen{\nu \vec{y}, \vec{w}.}\left( \id \cpar P' \cpar Q' \right)
				\using
				\mbox{\rulel{Close-R}}
			\end{prooftree}
		}
	\end{tabular}
	
	\caption{
		An \emph{early} non-interleaving structural operational semantics.
	}\label{figure:active}
\end{figure*}
}

\autoref{figure:active} defines the structural operational rules for deriving transitions.
Transitions are early in the sense that for each input there is a separate label for each choice of input message.
Since the LATS models generated by this SOS have extended process terms as the states,
we are ultimately 
interested in transitions between extended processes, and as a result the rules \rulel{Res}, \rulel{Alias-Out}, \rulel{Alias-Free} generally conclude the derivation trees\footnote{%
	Note, however, that we sometimes omit the active substitution when it is \(\id\), and will abuse the notation by writing \(P\) for \(\id \cpar  P\).
	Examples will always spell out the active substitution, for clarity.
}.
The two \rulel{Alias} rules check that private names do not appear directly in the labels and instantiate the aliases. 
The rest of the rules are applied to processes, higher up in the derivation tree.
A major strength of the applied $\pi$-calculus is that our results
hold independently of the choice of \emph{equational theory}, \ie of a set of equations, denoted $E$, that can be applied to messages.
In our examples we use a standard equational theory for pairs and symmetric encryption,
defined as $\fst{\pair{M}{N}} \mathrel{=_E} M$, $\snd{\pair{M}{N}} \mathrel{=_E}
N$, and $\dec{\enc{M}{K}}{K} \mathrel{=_E} M$.

The most important departure from established formulations of SOS for the applied $\pi$-calculus
is that we label transitions with \emph{events} rather than actions.
An event, in addition, indicates the location(s) of the input(s) and output(s) involved in the action.
When denoting transitions, the event is decomposed, with the action appearing above the transition's arrow and the locations below.
\begin{defi}[Location and events]
	\label{def:location-and-events}
	A \emph{location} \(\ell\) is of the form $\prefix{s}[t]$, where $\prefix{s} \in \left\{\prefix{0},\prefix{1}\right\}^*$
	and $t \in \left\{{0},{1}\right\}^*$. 
	If $\prefix{s}$ or \(t\) is empty, we omit it (hence, we write \(\prefix{\epsilon}[\epsilon]\) as \([]\)).
	A \emph{location label} $u$ is either a location \(\ell\) or a pair of locations \((\ell_0, \ell_1)\), and we let $\prefix{c}(\ell_0, \ell_1) = (\prefix{c}\ell_0, \prefix{c}\ell_1)$ for $\prefix{c} \in \left\{\prefix{0},\prefix{1}\right\}$.
An \emph{event} $\event{\pi}{u}$ is a pair consisting of an 
action label $\pi$ 
(as defined in \autoref{figure:syntax}) 
and a location label $u$.
\end{defi}

A \emph{free input} action label $M\,N$ 
in rule~\rulel{Inp}
indicates the message $N$ input on channel $M$.
An \emph{output} action label $\co{M}(\aliroot)$ 
in rule~\rulel{Out}
indicates that on channel $M$ is output a message that gets assigned the alias $\aliroot$, which can then be used to refer to that message in the future by an attacker observing the network.
That alias can also be used to establish a direct communication, where an input and output interact as in rules~\rulel{Close-*}, where the action label $\tau$ is used to denote an internal communication, and the observer cannot intercept the aliases used for this communication.
The only departure compared to the actions of a modern formulation of the applied $\pi$-calculus~\cite{Horne2021,Abadi2018} is the use of \emph{located} aliases.

The first part $\prefix{s}$ of a location $\prefix{s}[t]$ is a location prefix indicating which \rulel{Par} rules are applied,
while the second part $[t]$ indicates which branch is chosen by the \rulel{Sum} rules.
The label unambiguously picks out a specific input or output action in the syntax tree of a process with guarded choices.
When a transition is labelled with a pair of locations, the pair records the locations involved in a
$\tau$-transition, where the first is the location of an output transition
and the second is the location of an input transition that interacts with the output.
Note that both \rulel{Close} rules ensure that the first location is always the output,
which helps make some proofs more syntax directed.

\begin{example}\label{ex:pok}
Consider the process 
 $\id \cpar \nu m, n . (\cout{a}{\pair{m}{n}}\, \cpar \cin{m}{x}.[x = n] \cout{\pok}{\pok})$,
 where the empty active substitution (\eg \(\id\)) indicates that nothing has been output yet.
 Below is the derivation tree for an output transition of this extended process.
 
\noindent
\begin{centering}
\includegraphics[max width=\linewidth-.5em]{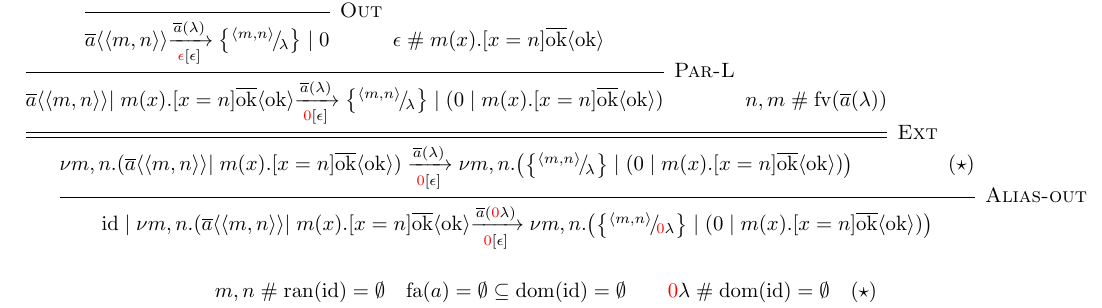}
\end{centering}

For readability, we sometimes add parenthesis around processes and apply the \rulel{Ext} or \rulel{Res} rules \enquote{in batch}, as shown in these example derivation trees (with a double line). 

In the derivation above, the private names \(m\) and \(n\) never occur in the labels, something that is guaranteed by the \rulel{Ext} rules (the \rulel{Res} rule also guarantees this property).
However, an attacker could intercept an output, gain knowledge of the alias  $\prefix{0}\lambda$ and use it in subsequent labels to indirectly refer to the private name \(m\) and \(n\): this \enquote{redirection} is represented using the active substitution $\sub{\prefix{0}\lambda}{\pair{m}{n}}$.
The \rulel{Alias-free} rule switches perspective by instantiating the aliases, \eg in the input transition below
we may refer to the private name $m$
by using the message $\fst{\prefix{0}\lambda}$ on the label (assuming some message theory $E$ featuring equation $\fst{\pair{m}{n}} \mathrel{=_E} m$).

\begin{centering}
	\includegraphics[max width=\linewidth-.3em]{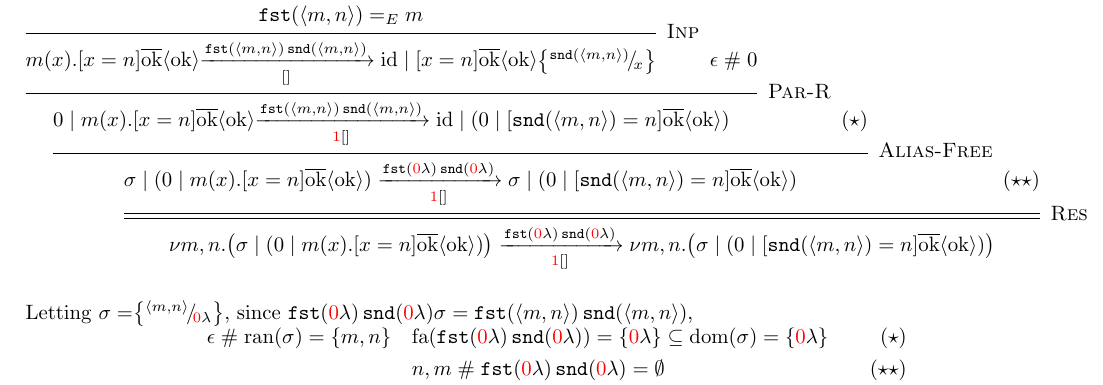}
\end{centering}

This helps explain why input and output prefixes $\cin{M}{x}.P$ and $\cout{M}{N}.P$, respectively, indicate the channel as a term $M$, which is the modern approach to the applied $\pi$-calculus~\cite{Abadi2018}.
\end{example}

Additional examples of transitions, along with their derivation trees, can be found in our companion paper~\cite[Appendix A]{Aubert2022e}.

\subsection{The independence relation}
\label{sec:ind}

A key reason for introducing event-labelled transitions is that it eases defining what it means for two transitions to be independent (\ie their events must be independent).
Independence of events is the key technological innovation
we use to define non-interleaving semantics and so-called concurrency diamonds.

Firstly, we define structural independence (\autoref{def:prefix-location}), which ensures that two independent events must occur in different locations.
This is done by making sure that the location prefixes occurring in both events are not (string) prefixes of each other.

\begin{defi}[Structural independence]
	\label{def:prefix-location}
	The \emph{structural independence}, $\sIndy$, is a 
	relation on location labels
defined as follows, where $\prefix{s} \in \left\{\prefix{0}, \prefix{1}\right\}^*$.
\begin{gather*}
\begin{prooftree}
\justifies
\prefix{s0}\ell_0 \sIndy \prefix{s1}\ell_1
\end{prooftree}
\qquad
\begin{prooftree}
\justifies
\prefix{s1}\ell_0 \sIndy \prefix{s0}\ell_1
\end{prooftree}
\qquad
\begin{prooftree}
\ell_0 \sIndy u
\qquad
\ell_1 \sIndy u
\justifies
\left( \ell_0, \ell_1 \right) \sIndy u
\end{prooftree}
\qquad
\begin{prooftree}
u \sIndy \ell_0
\qquad
u \sIndy \ell_1
\justifies
u \sIndy \left( \ell_0, \ell_1 \right)
\end{prooftree}
\end{gather*}
\end{defi}

\begin{example}
	In
\(
\id \cpar (
{\cout{a}{a}} \cpar \cout{b}{b}.({\cout{c}{c}} \cpar \cout{d}{d}))
\) the locations of the four output events are: \(\prefix{0}, \prefix{1}, \prefix{10}\) and \(\prefix{11}\),  respectively.
The output on channel $a$ is structurally independent from all other outputs.
The output on channel $b$ is not structurally independent with respect to the outputs
on channels $c$ or $d$, both having \prefix{1} as prefix.
However, these two are independent, since neither
\prefix{10} nor \prefix{11} are prefixes of each other.
\end{example}

From the way we have designed our SOS, if an event is a $\tau$-transition, then 
all locations involved must be structurally independent.
But, in addition to structural independence, we also need to detect \emph{link dependence} \ie whether an output influences another event.

\begin{defi}[Independence of events]
	\label{def:independence}
	The independence relation $\Indy$ on events
	is 
	such that 
	$\event{\pi_0}{u_0} \Indy \event{\pi_1}{u_1}$
	whenever:
	$u_0 \mathrel{I_{\ell}} u_1$
	and also,
	if $\pi_i = \co{M}(\alia)$
	then 
	$\isfresh{\alia}{\pi_j}$,
	with $i\neq j \in \{0,1\}$.
\end{defi}

\begin{example}
Observe that \autoref{ex:pok} contains a link dependency.
The two events \(\event{\co{a}(\prefix{0}\lambda)}{\prefix{0}[]}\) and \(\event{\fst{\prefix{0}\lambda}\snd{\prefix{0}\lambda}}{\prefix{1}[]}\) are structurally independent, since $\prefix{0}[] \mathrel{I_{\ell}} \prefix{1}[]$,
but are not link independent, since \(\prefix{0}\lambda\) is not fresh in \(\fst{\prefix{0}\lambda}\,\snd{\prefix{0}\lambda}\). 
\end{example}

Independence is deliberately only well-defined for events that can be enabled in the same state (\ie extended process) or in consecutive states.
To see this observe that the process in \autoref{ex:pok} can make a third transition labelled with event $\event{ \lout{\pok}{\prefix{1}\lambda_1} }{ \prefix{1}[] }$.
That event must happen after 
 \(\event{ \fst{\prefix{0}\lambda}\, \snd{\prefix{0}\lambda} }{ \prefix{1}[] }\),
due to structural dependence, which itself happens after
 \(\event{ \co{a}(\prefix{0}\lambda) }{ \prefix{0}[]}\), due to the link dependency above.
 Yet \(\event{ \co{a}(\prefix{0}\lambda) }{ \prefix{0}[] }\) and
 $\event{ \cout{\pok}{\prefix{1}\lambda_1} }{ \prefix{1}[] }$
 are independent according to \mbox{\autoref{def:independence}}.
However, since these events can never be enabled in consecutive states, this is not a limitation for the theory developed in this work.
Indeed it is a bonus, since independence becomes easy to compute, as we may ignore transitive dependencies.
If transitive dependencies are required, they may be computed by unfolding our LATS (similarly to transition system with independence) into suitable  event structures and Petri nets~\cite{Mukund1992,Joyal1993}.

\section{Interleaving semantics and attack formulas for each style}
\label{sec:interleaving}

This section introduces a formulation of an interleaving semantics
for the applied $\pi$-calculus.
The purpose of this section is to explain in detail how 
interleaving semantics
provide the results for our benchmark unlinkability properties
described in \autoref{sec:tech_preliminry_intro}.
We provide attacks using a modal logic characterising interleaving semantics.
We employ the SOS from \autoref{sec:SOS},
which is capable also of defining non-interleaving semantics,
so that results in this section can be compared more easily to results in subsequent sections.

The section also contains a new result showing that there is a particular formulation 
of unlinkability for the BAC protocol 
where an attack can be
detected using bisimilarity (\autoref{thm:system-min-silent}).
In later sections, we will show that 
the given attack is not detectable using similarity (\autoref{thm:sim-proof}).
Although it is well-known that bismilarity is finer than 
similarity, it may be unexpected that the additional power offered by bisimilarity is necessary to find attacks for some formulations of unlinkability.

\subsection{Interleaving semantics in the presence of located aliases}

Interleaving semantics are covered extensively in related work~\cite{Horne2021,Abadi2018}, and a companion paper explains how they are obtained from our non-interleaving operational semantics~\cite{Aubert2022k}.
Hence, we keep this section brief, as its goal is to define our interleaving similarity (\autoref{def:i-sim}).

All equivalences for the applied $\pi$-calculus make use of \emph{static equivalence}.
Static equivalence ensures that
an attacker cannot distinguish a snapshot of two extended processes
by performing a test using only intercepted messages that have been output in the past,
represented by the aliases \emph{extruded} by outputs.
In an extended process, the active substitution
records the information available to the attacker intercepting outputs on public channels.
The attacker can then combine that information in various ways to see whether the data that it has intercepted passes a test,
\eg hashing the first message and checking whether it is equal to the second message.

We find it insightful to express static equivalence by making use of the following satisfaction relation,
which represents the tests that may be performed.
\begin{defi}[Satisfaction]\label{def:equality}
	A satisfaction relation $\vDash$ between extended processes and messages equalities is defined inductively as follows.
	\begin{gather*}
	\begin{prooftree} 
	A\sub{x}{y} \vDash M = N
	\qquad
	\isfresh{y}{\nu x.A, M, N}
	\justifies
	\nu x.A \vDash M = N
	\end{prooftree}
\qquad
	\begin{prooftree}
	M\mathclose{\theta} \mathrel{=_E} N\mathclose{\theta}
	\justifies
	\theta \cpar P \vDash M = N
	\end{prooftree}
	\end{gather*}
We write \(\nvDash\) for the complement of \(\vDash\).
\end{defi}

The name management above ensures that the private names in an extended process
do not appear directly in $M$ or $N$,
leaving only the possibility of 
using aliases from the domain of the active substitution in $M$ and $N$ to indirectly refer to 
private names. 
In other words, $M$ and $N$ are recipes that must produce the same message, up to the equational theory $E$,
given the information recorded in the active substitution of the extended process.

\begin{example}
As a simple example we have the following.

\noindent\includegraphics[max width=\linewidth]{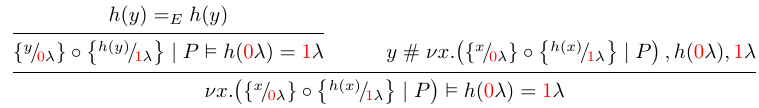}
\end{example}

Now we can make a generic point about all reasonable notions of equivalence based on our structural operational semantics~\cite{Aubert2022e}.
Recall that each alias has a location prefix, allowing each location to have its unique pool of aliases, thus ensuring that the choice of alias is localised and not impacted by choices of aliases made by concurrent threads.
For example, the following process has two transitions, labelled with $\event{ \co{a}(\prefix{0}\lambda) }{ \prefix{0}[] }$ and $\event{ \co{b}(\prefix{1}\lambda) }{ \prefix{1}[] }$: 

\begin{tikzpicture}[baseline=(current  bounding  box.center)]
	\node (orig) {\(\id \cpar \mathopen{\nu x.}\left({\cout{a}{x}} \cpar {\cout{b}{h(x)}}\right)\)};
	\node [above right = .1cm and 1.5cm of orig](left) {\(\mathopen{\nu x.}\left({\sub{\prefix{0}\lambda}{x}} \cpar 0 \cpar {\cout{b}{h(x)}} \right)\)};
	\draw [->] ($(orig.north east) + (-.3, 0)$) .. controls (orig.east |- left) ..
	node[above, sloped, pos=.9]{\(\co{a}(\prefix{0}\lambda)\)}
	node[below, sloped, pos=.9]{\(\prefix{0}[]\)}
	(left.west);
	\node [below right = .1cm and 1.5cm of orig](right) {\(\mathopen{\nu x.}\left({\sub{\prefix{1}\lambda}{h(x)}} \cpar {\cout{a}{x}} \cpar 0 \right)\)};
	\draw [->] ($(orig.south east) + (-.3, 0)$) .. controls (orig.east |- right) ..
	node[above, sloped, pos=.9]{\(\co{b}(\prefix{1}\lambda)\)}
	node[below, sloped, pos=.9]{\(\prefix{1}[]\)}
	(right.west);
\end{tikzpicture}

Clearly, any reasonable semantics should equate the above process with the one below, where the only difference is that the parallel processes $\cout{a}{x}$ and $\cout{b}{h(x)}$ have been permuted (\ie exchanged their locations).

\begin{tikzpicture}[baseline=(current  bounding  box.center)]
	\node (orig) {\(\id \cpar \mathopen{\nu x.}\left({\cout{b}{h(x)}} \cpar {\cout{a}{x}}\right)\)};
	\node [above right = .1cm and 1.5cm of orig](left) {\(\mathopen{\nu x.}\left({\sub{\prefix{1}\lambda}{x}} \cpar {\cout{b}{h(x)}} \cpar 0 \right) \)	};
	\draw [->] ($(orig.north east) + (-.3, 0)$) .. controls (orig.east |- left) ..
	node[above, sloped, pos=.9]{\(\co{a}(\prefix{1}\lambda)\)}
	node[below, sloped, pos=.9]{\(\prefix{1}[]\)}
	(left.west);
	\node [below right = .1cm and 1.5cm of orig](right) {\(\mathopen{\nu x.}\left({\sub{\prefix{0}\lambda}{h(x)}} \cpar 0 \cpar {\cout{a}{x}} \right)\)};
	\draw [->] ($(orig.south east) + (-.3, 0)$) .. controls (orig.east |- right) ..
	node[above, sloped, pos=.9]{\(\co{b}(\prefix{0}\lambda)\)}
	node[below, sloped, pos=.9]{\(\prefix{0}[]\)}
	(right.west);
\end{tikzpicture}

Notice that the events labelling the transitions differ only in the prefix string $\prefix{0}$ or $\prefix{1}$, but that this change impacts the domain of the active substitutions.
Therefore, when defining any notion of equivalence using this operational semantics, we must keep track of a bijection between aliases in the domain of the active substitution of each extended process. This allows for differences in prefixes and makes the particular choice of alias irrelevant when performing equivalence checking.
The following function is convenient to 
pick out the domain of an active substitution.
The domain remembers the set of aliases that have already been extruded.
\begin{defi}
	We extend the domain function to 
	extended processes such that
	$\dom{\mathopen{\nu \vec{x}.}\left( \theta \cpar P \right)} = \dom{\theta}$.
\end{defi}
We make use of bijection $\rho \colon \dom{A} \rightarrow \dom{B}$ 
in all notions of equivalence based on our structural operational semantics.
The following defines the notions of 
\enquote{interleaving similarity} and bisimilarity.
\begin{defi}[Interleaving similarity]\label{def:i-sim}
	A relation $\mathcal{R}$ is an \emph{interleaving similarity} (\emph{$i$-simulation})
	whenever if $A \mathrel{\mathcal{R}^{\rho}} B$,
	such that $\rho \colon \dom{A} \rightarrow \dom{B}$ is a bijection,
	then:
	\begin{itemize}
		\item If $A \dlts{\pi}{u} A'$ and $\pi$ is not an output, then
		there exist $B'$, $u'$, \st 
		$B \dlts{\pi\rho}{u'} B'$, 
		and $A' \mathrel{\mathcal{R}^{\rho}} B'$.
		\item
		If $A \dlts{\co{M}(\alpha)}{u} A'$, 
		 then there exist $B'$, $u'$, and $\alpha'$ \st 
		$B \dlts{\co{M\rho}(\alpha')}{u'} B'$, and $A' \mathrel{\mathcal{R}^{\rho\left[\alpha \mapsto \alpha'\right]}} B'$.
		\item $A \vDash M = N$ iff $B \vDash M\rho = N\rho$.
	\end{itemize}
	We say that an extended process $A$ \emph{is $i$-simulated by} $B$, and write $A \simi{i} B$, whenever
	there exists an $i$-simulation $\mathcal{R}$ such that $A \mathrel{\mathcal{R}^{id}} B$.
	If, in addition, the relation $\mathcal{R}$ is symmetric (defined such that $A \mathrel{\mathcal{R}^{\rho}} B$ iff $B \mathrel{\mathcal{R}^{\rho^{-1}}} A$), then $A$ and $B$ are $i$-bisimilar, written $A \bisimi{i} B$. 
\end{defi}

In the second clause above, notice how the bijection $\rho$ is updated,
where
$\rho\left[ \alpha \mapsto \alpha' \right]  : \dom{A'} \rightarrow \dom{B'}$
is
such that $\rho\left[ \alpha \mapsto \alpha' \right](\alpha) = \alpha'$ and 
otherwise $\rho\left[ \alpha \mapsto \alpha' \right](\beta) = \rho(\beta)$.
Therefore,
since
$\rho : \dom{A} \rightarrow \dom{B}$ is a bijection
and $\isfresh{\alpha}{\dom{A}}$ and $\isfresh{\alpha'}{\dom{B}}$,
as guaranteed by the SOS,
then
$\rho\left[ \alpha \mapsto \alpha' \right]$ is a bijection.
Also, notice that the above definition disregards the locations,
which plays an important role in subsequent sections when non-interleaving semantics are introduced.

\begin{example}	
Now consider again the extended processes examined before the definition
above
\begin{align*}
	\id \cpar \mathopen{\nu x.}\left({\cout{a}{x}} \cpar {\cout{b}{h(x)}}\right) && \text{ and } && \id \cpar \mathopen{\nu x.}\left({\cout{b}{h(x)}} \cpar {\cout{a}{x}}\right)\text{.}
\end{align*}
They are $i$-bisimilar, \ie there exists a symmetric $i$-simulation that relates them.
This bisimulation involves building up a bijection $\rho$
between aliases in each corresponding transition,
such that $\rho \colon \prefix{0}\lambda \mapsto \prefix{1}\lambda$
and
$\rho \colon \prefix{1}\lambda \mapsto \prefix{0}\lambda$.
Observe also that the final states these processes reach
are 
\begin{align*}
A = \mathopen{\nu x.}\left( {\sub{\prefix{0}\lambda}{x}}\circ{\sub{\prefix{1}\lambda}{h(x)}} \cpar 0 \cpar 0 \right)
&& \text{and} &&
B = \mathopen{\nu x.}\left( {\sub{\prefix{1}\lambda}{x}}\circ{\sub{\prefix{0}\lambda}{h(x)}} \cpar 0 \cpar 0 \right)\text{.}
\end{align*}

Since
$A \vDash h(\prefix{0}\lambda) = \prefix{1}\lambda$,
we also want this test to be satisfied by $B$, modulo the alias substitution $\rho$ that has been built by the $i$-bisimilarity, \ie
$B \vDash (h(\prefix{0}\lambda))\rho = (\prefix{1}\lambda)\rho$:
\begin{align*}
	(h(\prefix{0}\lambda)) \rho &= h((\prefix{0}\lambda)\rho) = h (\prefix{1}\lambda)&&& (\prefix{1}\lambda)\rho &= \prefix{0}\lambda\\
\shortintertext{Since we are working with respect to \(B\) and the active substitution \(\sigma = \sub{\prefix{1}\lambda}{x}\circ{\sub{\prefix{0}\lambda}{h(x)}}\) in \(B\) is such that}
(h(\prefix{1}\lambda)) \sigma &= h((\prefix{1}\lambda)\sigma) = h(x) &&& (\prefix{0}\lambda) \sigma &= h(x)
\end{align*}
by definition of satisfaction of equality in \(B\), $B \vDash (h(\prefix{0}\lambda))\rho = (\prefix{1}\lambda)\rho$ indeed holds.

Notice that it is necessary to apply $\rho$ to the messages when checking that
equality tests are preserved, and that it must be applied before the active substitution.
\end{example}

One may ask whether it is possible to simply have a permutation of location prefixes, keeping alias variables the same.
Such an approach would not be sufficiently flexible to capture relations such as those in \autoref{ex:alias_substitution} below.
In these examples observe that proving their relation requires a simulation involving bijections 
mapping two aliases with the same location prefix to two aliases with different location prefixes.
This explains why we employ a bijection between aliases rather than a bijection between locations.
\begin{example}\label{ex:alias_substitution}
Consider the following two interleaving similarities:
\begin{align}
	\id \cpar ( \mathopen{\nu x.}\left({\cout{b}{h(x)}.\cout{a}{x}}\right)) & \simi{i} \id \cpar ( \mathopen{\nu x.}\left({\cout{b}{h(x)}} \cpar {\cout{a}{x}}\right)) \label{example:sim-1}
	\shortintertext{and}
	\id \cpar ( \mathopen{\nu x.}\left({\cout{a}{x}} \cpar {\cout{x}{h(x)}}\right)) & \simi{i} \id \cpar (\mathopen{\nu x.}\left({\cout{a}{x}.\cout{x}{h(x)}}\right)) \label{example:sim-2}
	\text{.}
\end{align}
In both examples, observe that on one side there are two locations, and on the other there is only one.
To prove (\ref{example:sim-1}) we construct a simulation involving a bijection over aliases $\rho$ such that
\begin{align*}
	\rho \colon \lambda \mapsto \prefix{0}\lambda && \text{ and } && \rho \colon \lambda' \mapsto \prefix{1}\lambda
\end{align*}
and (\ref{example:sim-2}) requires a bijection over aliases $\rho$ such that
\begin{align*}
	\rho \colon \prefix{0}\lambda \mapsto \lambda && \text{ and } && \rho \colon \prefix{1}\lambda \mapsto \lambda'\text{.}
\end{align*}

As we can see, aliases containing only the prefix string \(\prefix{\epsilon}\) are mapped to aliases containing the prefix string \(\prefix{0}\) or \(\prefix{1}\) in the first case, while aliases containing the prefix strings \(\prefix{0}\) and \(\prefix{1}\) are both mapped to aliases containing the prefix string \(\prefix{\epsilon}\) in the second case.
\end{example}

\subsection{Describing attacks using modal logics}\label{sec:modal}

In this section, we show how to use modal logics to describe attacks that are detected by interleaving semantics on our benchmark unlinkability problems.
According to \autoref{def:i-sim}, to prove that two extended processes are bisimilar, we need to define a suitable relation over extended processes, indexed by an alias substitution, and check that the relation is a bisimulation.
In contrast, to prove that two extended processes are \emph{not} bisimilar, 
one needs to exhibit a behaviour that one process has but not the other. Modal logics are well suited for modelling behaviours of transition systems, especially in the interleaving case. Therefore, we need to find a modal logic formula that
holds for one process but not for the other; 
we then say that this formula describes an attack.

The modal logic called \emph{classical $\FM$} in~\cite{Horne2021}, characterises interleaving bisimilarity for the 
applied $\pi$-calculus, and was introduced to characterise attacks on BAC.
In this section, we present a formulation of this modal logic that accounts for the locations of aliases.
In later sections, we will adapt this logic to characterise non-interleaving bisimilarities.
Here, we present examples of formulas describing attacks on formulations of the
 BAC protocol.

\begin{figure}
	\begin{multicols}{2}
		\noindent\begin{align*}
			\phi, \psi \Coloneqq &&& \ttt \tag*{top} \\
			\mid &&& M = N \tag*{equality} \\
			\mid &&& \phi \wedge \psi \tag*{conjunction} \\
			\mid &&& \neg \phi \tag*{negation} \\
			\mid &&& \diam{\pi}\phi \tag*{diamond} 
		\end{align*}
		\begin{align*}
			\fff &\triangleq \neg \ttt\\
			M \neq N &\triangleq \neg M = N\\
			\phi \vee \psi &\triangleq \neg\left(\neg \phi \wedge \neg \psi \right) \\
			\phi \yields \psi &\triangleq \neg \phi \vee \psi\\
			\boxm{\pi}\phi &\triangleq \neg \diam{\pi}\neg \phi
		\end{align*}
	\end{multicols}
	\caption{Formulas of $\FM$ (left) and their common abbrevations (right)}\label{fig:FM-formulas}
\end{figure}

\begin{defi}
	\label{def:FM-syntax}
The syntax of $\FM$ is given by the grammar on the left in \autoref{fig:FM-formulas} 
and common abbreviations are given on the right,
where $\pi$ is either of the form
$\tau$, $M\,N$ or $\co{M}(x)$, and
where $M$, $N$ are messages and $x$ is a variable.

Formulas $\diam{\co{M}(x)}\phi$ and
$\boxm{\co{M}(x)}\phi$
bind occurrence of $x$ in $\phi$.
The \emph{simulation} fragment of $\FM$ is the fragment of $\FM$ where negation is applied to equality only.
\end{defi}

Notice that in the logic, compared to the process calculus, a variable rather than 
an alias is used for bound outputs.
This, we will see, is to allow the variable to be instantiated with a specific alias
when evaluating the formula,
as determined by the process.

Inference rules for the satisfaction relation are presented in \autoref{fig:modal-1}.
The semantics for conjunction and negation are standard, whereas for equality we take it from \autoref{def:equality}.
The diamond modality is where creativity was required. 
Observe in \autoref{fig:modal-1} that, below the horizontal line, we define the classical negation of formulas  in terms of failure of satisfaction.
The meaning of the set comprehension in the rule for the failure of diamond is that all branches in the set must be explored.

\begin{figure}
	{
		\newcommand{\tablespace}{\\[3em]} 
		\newcommand{\hypospace}{\hspace{1.5em}} 
		\newcommand{\columnspace}{1.5em} 
		\begin{tabular}{c@{\hskip \columnspace} c}
			\begin{prooftree}
				\justifies
				A \vDash \ttt 
			\end{prooftree}
			& 
			\begin{prooftree} 
				A\sub{x}{y} \vDash \phi
				\hypospace
				\isfresh{y}{\nu x.A, \phi}
				\justifies
				\nu x.A \vDash \phi
			\end{prooftree}
			\tablespace
			\begin{prooftree}
				M\theta \mathrel{=_E} N\theta
				\justifies
				\theta \cpar P \vDash M = N
			\end{prooftree}
			& 
			\begin{prooftree}
				A \dlts{\co{M}(\alpha)}{u} B 
				\hypospace
				B \mathrel{\vDash}
				\phi\sub{x}{\alpha}
				\hypospace
				\isfresh{\alpha}{\phi}
				\justifies
				A \mathrel{\vDash} \diam{\co{M}(x)}\phi
			\end{prooftree}
			\tablespace
			\begin{prooftree}
				A \vDash \phi
				\qquad
				A \vDash \psi
				\justifies
				A \vDash \phi \wedge \psi
			\end{prooftree}
			& 
			\begin{prooftree}
				A \dlts{\pi}{u} B
				\hypospace				
				B \mathrel{\vDash}
				\phi
				\justifies
				A \mathrel{\vDash} \diam{\pi}\phi
									\using
									\mbox{$\pi$ not output}
			\end{prooftree}
			\tablespace
			\hline
			&\hfill \textsc{negation}
			\\[1.4em]
			\begin{tabular}{c@{\hskip \columnspace} c}
				\begin{prooftree}
					A \nvDash \phi
					\justifies
					A \vDash \neg\phi
				\end{prooftree}
				& 
				\begin{prooftree}
					A \vDash \phi
					\justifies
					A \nvDash \neg\phi
				\end{prooftree}
			\end{tabular}
			& 
			\multirowcell{3}{
				\begin{prooftree}
					\left\{
					B \nvDash
					\phi\sub{x}{\alpha}
					\colon
					A \dlts{\co{M}(\alpha)}{u} B, 
					\text{ and }
					\isfresh{\alpha}{\phi}
					\right\}
					\justifies
					A \mathrel{\nvDash} \diam{\co{M}(x)}\phi
				\end{prooftree}
				\\[4em]
				\begin{prooftree}
					\left\{
					B \nvDash
					\phi
					\colon
					A \dlts{\pi}{u} B
					\right\}
					\justifies
					A \mathrel{\nvDash} \diam{\pi}\phi
									\using
									\mbox{$\pi$ not output}
				\end{prooftree}
			}
			\tablespace
			\begin{prooftree}
				M\theta \mathrel{\neq_E} N\theta
				\justifies
				\theta \cpar P \nvDash M = N
			\end{prooftree}	
			& 
			\tablespace
			\begin{tabular}{c@{\hskip \columnspace} c}
				\begin{prooftree}
					A \nvDash \phi
					\justifies
					A \nvDash \phi \wedge \psi
				\end{prooftree}
				&
				\begin{prooftree}
					A \nvDash \psi
					\justifies
					A \nvDash \phi \wedge \psi
				\end{prooftree}
			\end{tabular}
		\end{tabular}
		\caption{Rules for modal logic classical $\FM$, where modalities are labelled with actions.}\label{fig:modal-1}
	}
\end{figure}

In \autoref{fig:modal-1}, we present the obvious standard diamond modality
that can be used to characterise interleaving equivalences.
It is the same as in related work on classical $\FM$~\cite{Horne2021}, except that
when there is an output, we replace the output variable with 
the alias determined by the process, \ie note the $\phi\sub{x}{\alpha}$ in the rules involving the diamond.
While that alias is guaranteed by the structural operational semantics
to be fresh for the extended process $A$, we should check that $\alpha$ was fresh for $\phi$,
since $\phi$ and $A$ can be designed independently of each other.
This allows processes with actions in different locations to satisfy the same formulas, as shown in this example.

\begin{example}\label{ex:commutative}
It holds that  
\begin{align*}
\id \cpar {\cout{a}{c}} \cpar {\cout{b}{h(c)}}
& \vDash
\diam{\lout{a}{x}}\diam{\lout{b}{y}}(h(x) = y)
\shortintertext{and}
\id \cpar {\cout{b}{h(c)}} \cpar {\cout{a}{c}}
&\vDash
\diam{\lout{a}{x}}\diam{\lout{b}{y}}(h(x) = y)\text{,}
\end{align*}
even if each 
of the two $i$-bisimilar 
extended processes must produce different located aliases
when evaluating the modalities.
\end{example}

Since diamond is dual to box, the failure of diamond can be defined directly
by checking for failure in all reachable extended processes, as suggested in \autoref{fig:modal-1}.
Since our transition system is image-finite, there are finitely many branches to 
check.\footnote{
To be more precise, for any $A$ and $\pi$,
the set of $B$ such that $A \dlts{\pi}{u} B$, for any location $u$, quotiented by the congruence generated by $\alpha$-conversion, $\nu x.\nu y.A \equiv \nu y.\nu x.A$ and $\bang P \equiv P \cpar \bang P$
is finite.
When accounting for $A$, $\pi$ and $u$, the stronger \emph{event determinism} holds meaning:
if $A \dlts{\pi}{u} B$ and $A \dlts{\pi}{u} C$ then $B \equiv C$, even without $\bang P \equiv P \cpar \bang P$~\cite{Aubert2022k}.
}

Notice also that the location part of the events 
 is never used in these interleaving semantics.
The located aliases and location labels are however important for our non-interleaving equivalences in later sections,
and for concurrency diamonds required to extend techniques such as partial order reduction to the full applied $\pi$-calculus.

\subsection{Using Hennessy-Milner properties to identify attacks}
\label{sec:HM}

In \autoref{thm:i-HM} below, we exploit the established Hennessy-Milner properties: the first part ensures that bisimilarity holds whenever the same $\FM$ formulas are satisfied; while the second statement ensures that similarity holds whenever formulas in the simulation fragment of $\FM$ (\autoref{def:FM-syntax}) are preserved from left to right.
\begin{thm}\label{thm:i-HM}
The following hold for extended processes $A, B$:
\begin{itemize}
\item 
$A \bisimi{i} B$ iff, we have that, for all $\FM$ formulae $\phi$, $A \vDash \phi$ iff $B \vDash \phi$.
\item 
$A \simi{i} B$ iff, we have that, for all formulae $\phi$ in the simulation fragment of $\FM$, $A \vDash \phi$ implies $B \vDash \phi$.
\end{itemize}
\end{thm}

The proof of this theorem is classic, presented for the applied $\pi$-calculus in related work~\cite{Horne2021}, and also can be derived from the proofs of the more general Theorems in \autoref{app:HM}.

We will now use these results to check the theorems from \autoref{sec:tech_preliminry_intro}
 that establish attacks on our benchmark problems.
For this evaluation, from now until Section~\ref{sec:protocols},
 we fix a version of the BAC protocol.
 The two threads below represent the projections
 of the prover and verifier in the message sequence chart in \autoref{fig:msc}.
\begin{align}
\MainUK(c,d,ke,km) & \triangleq 
\begin{aligned}[t] 
			& \nu nt.\cout{c}{nt}.d(y). \\
			& 
                          \match{ \snd{y} = \mac{\fst{y}, km} }
			  \\
			& \hspace{0.5em} 
                                         \match{ nt = \fst{\snd{\dec{\fst{y}}{ke}}} }
                                         \\
			& \hspace{1em} \nu kt.\clet{m}{\enc{\pair{nt}{\pair{\fst{\dec{\fst{y}}{ke}}}{kt}}}{ke}} \\
			& \hspace{1.5em} \cout{c}{m, \mac{m, km}}
\end{aligned}
\raisetag{7em} 
\tag{Prover}\label{eq:def-prover}\\
\Reader(c,d,ke,km) & \triangleq
\begin{aligned}[t]
  			& d(nt).\nu nr.\nu kr. \\
			& \clet{m}{\enc{\pair{nr}{\pair{nt}{kr}}}{ke}} \\
			& \hspace{1em} \cout{c}{m, \mac{m, km}}
\end{aligned}
\tag{Verifier}\label{eq:def-verifier}
\end{align}
Process $\Reader$, modelling the control flow of a reader,
receives a challenge and sends a response to the challenge along with a MAC.
Process $\MainUK$, modelling the control flow of an ePassport, 
sends a challenge, receives the response from the verifier,
and responds to the verifier if the message received is well-formed.
In this version, if the response received from the verifier
is not well-formed the prover does not respond at all (rather than responding with an error message).
Notice how in both roles $c$ is used for all outputs and $d$ for all input (this feature of the style effectively means that channels play no significant role other than simplifying verification by ruling out $\tau$-transitions).

\begin{figure}
\[
\includegraphics[width=0.9\linewidth]
	{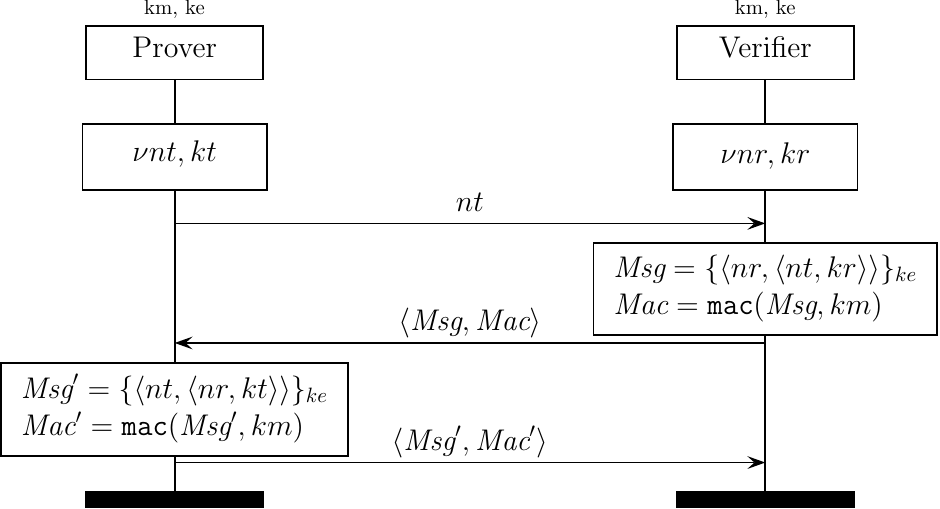}
	\]
\caption{Message sequence chart representing the BAC protocol.}\label{fig:msc}
\end{figure}

We first prove the formulation where a 
constant message $\getchallenge$ 
is sent from the reader to an ePassport
in order to initiate a session.
This corresponds to \autoref{thm:system-get}, as formulated in
\autoref{sec:tech_preliminry_intro}.\footnote{
	All specifications omit the $\id$ active substitution, for the sake of readability.%
}

\systemget*

\begin{proof}
Consider the following formula
\[
\phi_{\ref{thm:system-get}} \triangleq
\begin{array}[t]{l}
\diam{d\,\getchallenge}\diam{\co{c}(g_1)}\diam{\co{c}(g_2)}\Big(
\\
\qquad
\begin{array}[t]{l}
{g_1 = \getchallenge} \wedge {g_2 = \getchallenge}~\wedge
\\
\diam{\co{c}(nt_1)}\diam{d\,nt_1}\diam{d\,nt_1}
\\
\diam{\co{c}(u_1)}\diam{\co{c}(u_2)}\Big(
\\
\qquad
\begin{array}[t]{l}
{nt_1 \neq \getchallenge}
\wedge
{u_1 \neq \getchallenge}
\wedge
{u_2 \neq \getchallenge}
~\wedge
\\
\diam{d\,u_1}\diam{\co{c}(w)}(
w \neq \getchallenge  
  )
~\wedge
\\
\diam{d\,u_2}\diam{\co{c}(w)}(
w \neq \getchallenge  
 )~\Big)~\Big)
\end{array}
\end{array}
\end{array}
\]
Since
$\SystemCSF \vDash \phi_{\ref{thm:system-get}}$ but $\SpecCSF \nvDash \phi_{\ref{thm:system-get}}$, 
and since $\phi_{\ref{thm:system-get}}$ is in the simulation fragment of $\FM$, as negation is only applied to equality, \autoref{thm:i-HM} applies and we get $\SystemCSF \nsimi{i} \SpecCSF$.
\end{proof}

Convincing oneself of the correctness of this proof amounts to verifying that the formula is satisfied by the system but not the specification; this is an exercise in working with our semantic rules for deriving transitions and with the proof rules for our logic.
Intuitively, in this attack the same ePassport is connected to two different readers (on line 3 of the formula). The $\getchallenge$ message allows us to decide whether we are starting (as on line 2) or whether we are in the middle of the computation of the processes (as in the last three lines). The formula fails for $\SpecCSF$ in the last two lines because the two different messages are encrypted with two different keys, but the single ePassport can decrypt only one of them, thus one of the lines will break. This formula holds for $\SystemCSF$ because all this execution is carried inside the same session, with one ePassport sharing the same key with multiple verifiers.

Notice how this proof (as well as all our other proofs based on simulations) uses the modal formula to describe one behaviour that the system has but that the specification does not. Such behaviours can be understood as attacks because of the way the system is set up (compared to the specification) which allows an active\footnote{Notice that the adjective \enquote{active} comes from the power to observe branching points with similarity, as opposed to just trace equivalence.} attacker to distinguish when the same passport is used with two different readers (as in the System).

\begin{remark}
The applied $\pi$-calculus lends itself well to modelling standard Dolev-Yao attackers, who sit on the network intercepting all outputs and manufacturing inputs using those outputs and open knowledge~\cite{dolev83tit}. Thus the channels $c$ and $d$ in this model simply model connections to such a network (in this case a wireless network). Such Dolev-Yao attackers can play actions in orders that go against the expected order of messages in a benign network. The attacker knows already the constant $\getchallenge$, modelled by $\getchallenge$ being a free variable and not a bounded private name, so can inject $\getchallenge$ before it is even sent, as you can see at the beginning of the formula in the proof above. Similarly, once the attacker knows private output $nt_1$, there is nothing stopping the attacker injecting the same message in two different places, even if it was only intended for one of them, as you can see above where the same input message appears twice.
\end{remark}
 
Next, we come back to \autoref{thm:system-end}, as formulated in
\autoref{sec:tech_preliminry_intro}.
In the proof below, the formula
makes use of fresh channels (namely, $p_1$, $r_1$ and $r_2$) to identify
communications from each endpoint, namely one ePassport and two readers, respectively.
Remember, also, that in this instantiation of our problem, the first two variables of $\Reader$ and $\MainUK$ are identical.

\systemend*
\begin{proof}
	\label{proof:systemend}S
Consider the following formula. 
\[
\phi_{\ref{thm:system-end}} \triangleq 
\begin{array}[t]{l}
\diam{\co{p}(p_1)}\diam{\co{r}(r_1)}\diam{\co{r}(r_2)}
\diam{\co{p_1}(nt_1)}\diam{r_1\,nt_1}\diam{r_2\,nt_1}
\\
\diam{\co{r_1}(u_1)}\diam{\co{r_2}(u_2)}\Big(
\quad
\begin{array}[t]{l}
\diam{p_1\,u_1}\diam{\co{p_1}(w)}\ttt
\\
\wedge
\\
\diam{p_1\,u_2}\diam{\co{p_1}(w)}\ttt
~\Big)
\end{array}
\end{array}
\]
Observe that $\SystemUK \vDash \phi_{\ref{thm:system-end}}$ 
but $\SpecUK \nvDash \phi_{\ref{thm:system-end}}$.
Since $\phi_{\ref{thm:system-end}}$ is in the simulation fragment of $\FM$, \autoref{thm:i-HM} applies and we get $\SystemUK \nsimi{i} \SpecUK$.
\end{proof}

The formulas in the above two proofs represent fundamentally
the same strategy. 
The first three diamond modalities initialize the three threads involved, and generate three unique channels. 
The unique channels for each thread are used in the rest of the strategy to tie each actions to a specific thread.
The two conjuncts at the end is again the place where the formula breaks for $\SpecUK$ because different keys are used in the two messages $u_1$ and $u_2$, whereas it holds for $\SystemUK$ because the ePassport and the two readers share the same keys in one session.
The difference compared to \autoref{thm:system-get} is that $\phi_{\ref{thm:system-get}}$ indirectly
determines
 whether an action 
involved a new ePassport or reader session or an existing session,
by making use of $\getchallenge$.
For example, the innermost $w \neq \getchallenge$ indicates that an output could be performed without starting a new session,
which in combination with all other conditions, is sufficient to determine that the ePassport must have successfully authenticated. 
Recall that this is the style of the original formulation of this unlinkability problem~\cite{Arapinis2010},
while the style used in \autoref{thm:system-end} was introduced precisely to avoid such indirect explanations of attacks~\cite{Horne2021}.

Now consider our finite formulation of the same problem
(\autoref{thm:system-finite}
from \autoref{sec:tech_preliminry_intro}),
where neither fresh channels nor a constant message is employed. 

\systemfinite* 
\begin{proof}
Consider the following formula.

\[
\phi_{\ref{thm:system-finite}} \triangleq 
\begin{array}[t]{l}
\diam{\co{c}(nt_1)}\diam{\co{c}(nt_2)}\diam{d\,nt_1}\diam{d\,nt_1} \diam{\co{c}(u_1)}\diam{\co{c}(u_2)}\Big(
	\mkern-3mu
	\begin{array}[t]{l}
\diam{d\,u_1}\diam{\co{c}(w)}\ttt
		\\
		\wedge
		\\
		\diam{d\,u_2}\diam{\co{c}(w)}\ttt
		~\Big)
	\end{array}
\end{array}
\]

Since
$\SystemF \vDash \phi_{\ref{thm:system-finite}}$ but $\SpecF \nvDash \phi_{\ref{thm:system-finite}}$, and since $\phi_{\ref{thm:system-finite}}$ is in the simulation fragment of $\FM$, \autoref{thm:i-HM} applies and we get $\SystemF \nsimi{i} \SpecF$.
\end{proof}

As before, the innermost modality $\diam{\co{c}(w)}$
can only be satisfied if the ePassport successfully authenticates;
\ie both the keys and the nonces have to match.
In this case, we make certain that $\co{c}(w)$ is not an output originating 
from a new session by using (instead of a $\getchallenge$ message) a strategy where all sessions must be used up: that is the purpose of the second modality $\diam{\co{c}(nt_2)}$, which starts also the second ePassport, contrary to the previous two proofs.
The formula fails on $\SpecF$ because, when trying to authenticate, the nonces differ (seen in the modality sequence $\diam{d\,nt_1}\diam{d\,nt_1}$), even if the keys match in the two different sessions. Whereas $\SystemF$ respects this formula because the two conjuncts at the end are essentially the same when instantiating (\ie choosing the branches) with the same keys.

The above strategy assumes the presence of at most two ePassport and reader sessions,
and could be extended to larger bounds on the number of sessions.
However,
notice that the last two actions in our strategy are indistinguishable when performed by a reader compared to an authenticating ePassport.
Therefore, the above strategy does not work for an unbounded number of sessions, as there would always be a new matching output available no matter which strategy formula we use, as these formulas are finite.

\subsection{Attacks requiring the power of bisimilarity}

In the unbounded setting of the BAC protocol, perhaps surprisingly, attacks can be discovered only
if we have at our disposal the full power of bisimilarity.
The attack strategy makes heavy use of the observable difference between
the nonce output by the ePassport and the outputs at later stages,
which are structured as pairs.

\begin{thm}\label{thm:system-min-silent}
$
\SystemMin
\nbisimi{i}
\SpecMin
$, 
where
\begin{align*}
	\SystemMin & \triangleq  \bang\mathopen{\nu ke, km.\bang}\left( \Reader(c,d,ke,km) \cpar \MainUK(c,d,ke,km) \right)
	\\[1.3em]
	\SpecMin & \triangleq  \bang\mathopen{\nu ke, km.}\left( \Reader(c,d,ke,km) \cpar \MainUK(c,d,ke,km) \right)
\end{align*}

\end{thm}
\begin{proof}
Consider the following formula.
\begin{align*}
	\psi & \triangleq \diam{\co{c}(nt_1)}\diam{d\,nt_1}\diam{d\,nt_1} \diam{\co{c}(u_1)}\diam{\co{c}(u_2)}\Big( \\
	\MoveEqLeft[-3] \pair{\fst{u_1}}{\snd{u_1}} = u_1 \wedge \pair{\fst{u_2}}{\snd{u_2}} = u_2 ~\wedge \\
	\MoveEqLeft[-3] \diam{d\,u_1}\diam{\co{c}(w)}\big( \pair{\fst{w}}{\snd{w}} = w ~\wedge\\
	\MoveEqLeft[-11] \boxm{d\,z}\diam{\co{c}(v)}\left( \pair{\fst{v}}{\snd{v}} = v \right)~\big)\\
	\MoveEqLeft[-3] \wedge \\
	\MoveEqLeft[-3] \diam{d\,u_2}\diam{\co{c}(w)}\big( \pair{\fst{w}}{\snd{w}} = w ~\wedge \\
	\MoveEqLeft[-11] \boxm{d\,z}\diam{\co{c}(v)}\left( \pair{\fst{v}}{\snd{v}} = v \right)~\Big)
\end{align*}
Since
$\SystemMin \vDash \psi$ but $\SpecMin \nvDash \psi$,
by \autoref{thm:i-HM},
$\SystemMin \nbisimi{i} \SpecMin$.
\end{proof}

We now explain the formula \(\psi\) used in the proof above.
The message equality tests
$\pair{\fst{u_1}}{\snd{u_1}} = u_1$
and
$\pair{\fst{u_2}}{\snd{u_2}} = u_2$
ensure that $u_1$ and $u_2$ are pairs, hence \eg a nonce will fail this test.
Furthermore, these two messages $u_1$ and $u_2$ must be responses to the two inputs
induced by $\diam{d\,nt_1}\diam{d\,nt_1}$,
and must have involved readers: indeed, 
the only other possibility is that one of the input to $\diam{d\,nt_1}$ sent the $nt_1$ received along \(\co{c}\) back to the ePassport, which would block the ePassport,
and make it impossible for both $u_1$ and $u_2$ to be pairs.
So we know that $u_1$ and $u_2$ are responses from readers to the nonce $nt_1$, originating from an ePassport.

Now consider the following subformula.
\[
\diam{d\,u_1}\diam{\co{c}(w)}\big(
\begin{array}[t]{l}
  \pair{\fst{w}}{\snd{w}} = w ~\wedge \\
  \boxm{d\,z}\diam{\co{c}(v)}\left( \pair{\fst{v}}{\snd{v}} = v \right)~\big)
\end{array}
\] 
The diamond modalities
$\diam{d\,u_1}\diam{\co{c}(w)}$
and
$\pair{\fst{w}}{\snd{w}} = w$
indicate that $u_1$, the first response received from a reader, 
is fed as an input into something after which we see an output \(w\) structured as a pair (hence not involving a new ePassport).
Moreover, the test
$\boxm{d\,z}\diam{\co{c}(v)}\left( \pair{\fst{v}}{\snd{v}} = v \right)$
ensures that this sequence of input and output must have involved the 
ePassport that was already running and not a new reader session: this is precisely where \(\SpecMin\) fails.
Notice that both a successfully authenticating ePassport and
reader can produce this sequence.
The difference is that, if the ePassport was used,
then there would only be readers left that can respond to inputs.
But readers always produce a pair as a response to any input,
such as the dummy message $z$ indicated in the formula.
In contrast, if the ePassport was still alive, it would block on
receiving a meaningless message $z$: hence, we can conclude that the ePasseport was not used previously, since that would have prevented the box modality from holding.

The above strategy is surprising and non-trivial.
However, in reality, it expresses the same attack as all other formulas, just in a roundabout way.
By checking whether certain messages are pairs the strategy can determine whether a new session starts and also whether an old session has completed.
Simulation is not enough, since bisimulation is required to check that certain steps are impossible,
upon reaching a state where the ePassport has been used in a successful authentication attempt.
This is a novel observation, since all violations of unlinkability previously reported in literature required at most $i$-similarity 
in order to reveal them,
and never the additional capabilities of $i$-bisimilarity.

\subsection{Interleaving semantics and Feldhofer}\label{sec:meaningless}

The protocol we now consider is a simpler predecessor to the BAC protocol, proposed by Feldhofer et al.~\cite{Feldhofer2004}.
The flow of the Feldhofer protocol is similar to BAC except that no MAC appears.
We also, for now, assume authentication fails silently, as for the variant of the BAC protocol discussed previously.
We focus on the minimal unbounded formulation of unlinkability.
With respect to the other formulations discussed in \autoref{sec:tech_preliminry_intro} (involving a bounded number of sessions, a constant $\getchallenge$ or an endpoint), 
the same attack strategies as expressed in formulas $\phi_{\ref{thm:system-get}}$, $\phi_{\ref{thm:system-end}}$ and 
$\phi_{\ref{thm:system-finite}}$ 
of Theorems~\ref{thm:system-get}, \ref{thm:system-end} and 
\ref{thm:system-finite}, respectively,
also apply to the Feldhofer protocol.
Thus there is no new insight offered by considering Feldhofer with respect to most styles considered in \autoref{sec:tech_preliminry_intro}.
This, as we will explain next, is surprising, considering differences between the unlinkability analysis of 
Feldhofer and BAC for the minimal unbounded style of formulating unlinkability explained in what follows.
We also explain here how the formal evidence we present further reinforces our case for adopting HP semantics in the subsequent section.

\begin{figure}
\[
\includegraphics[width=0.8\linewidth]
	{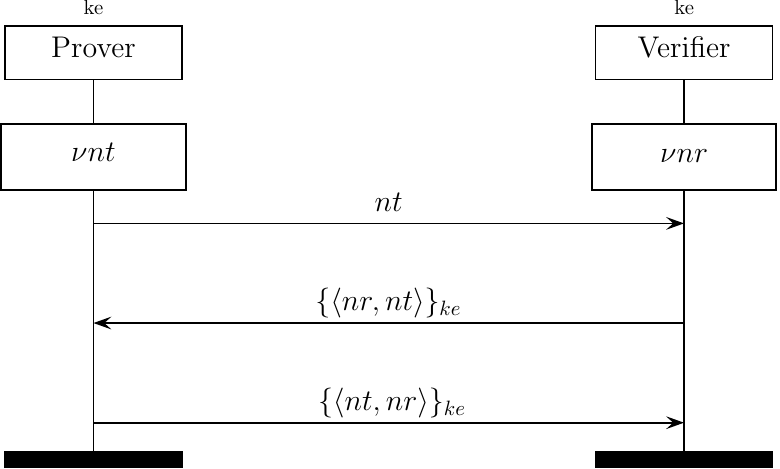}
\]
\caption{Message sequence chart representing the Feldhofer protocol.}\label{fig:msc-feld}
\end{figure}

The prover and verifier of the Feldhofer protocol are defined by the following processes:
\label{feldhoferf}

\begin{align*}
	\MainUK^{F}(c,d,k) & \triangleq
	\begin{array}[t]{l}
		\nu nt.\cout{c}{nt}.d(y). \\
		\match{nt = \snd{\dec{y}{k}}}
		\\
		\hspace{1em} \cout{c}{\enc{\pair{nt}{\fst{\dec{y}{k}}}}{k}}
	\end{array}
&& 
	\Reader^{F}(c,d,k) & \triangleq
	\begin{array}[t]{l}
		d(nt).\nu nr. \\
		\cout{c}{\enc{\pair{nr}{nt}}{k}}
	\end{array}
\end{align*}

Now, since $i$-bisimilarity was sufficient to find attacks on all variants of unlinkability studied
for the BAC protocol,
and since the same attacks apply to Feldhofer,  
we na{\i}vely expect that $i$-bisimilarity is sufficient to find attacks in all styles.
This expectation is in fact incorrect.
We address this concern formally by proving
that, with the power of $i$-bisimilarity, it is impossible to find an attack using the formulation of unlinkability and the Feldhofer protocol under consideration.
\begin{thm}\label{thm:bisim-proof}
$\SystemMin^{F} \bisimi{i} \SpecMin^{F}$, where
\[
\begin{array}{rl}
	\SystemMin^F \triangleq & \bang\mathopen{\nu k.\bang}\left( \Reader^F(c,d,k) \cpar  \MainUK^F(c,d,k) \right)
	\\[1.6em]
	\SpecMin^F \triangleq & \bang\mathopen{\nu k.}\left( \Reader^F(c,d,k) \cpar  \MainUK^F(c,d,k) \right)
\end{array}
\]
\end{thm}
\begin{proof}
We establish that
\begin{align*}
\SystemMin^{F} \bisimi{i} \bang\nu n.\cout{c}{n} \cpar \bang d(x)
&& \text{ and } &&
\bang\nu n.\cout{c}{n} \cpar \bang d(x)
 \bisimi{i} \SpecMin^{F}
\end{align*}
both hold.
In both cases, when 
$\bang\nu n.\cout{c}{n} \cpar \bang d(x)$
leads, with an output or input action, we start any new RFID tag or reader thread respectively
in the other process.
When $\SystemMin^{F}$ or $\SpecMin^{F}$ leads, either an input or output action is played
that can be matched by starting a new process on the RHS that performs only that input or output.

To be more precise about 
$\bang\nu n.\cout{c}{n} \cpar \bang d(x) \bisimi{i} \SpecMin^{F}$,
we construct a bisimulation relation which on one side covers the execution space of $\SpecMin^{F}$
and the other matches the number of outputs performed with nonces.
Thus a bisimulation relation on which the game described above can be played can be defined as the symmetric closure of the relation below, where the process on the left below is quotiented by congruence $0 \cpar P \equiv P$ and $P \cpar \bang P \equiv \bang P$.
\[
\mathopen{\nu \vec{y}.}\left(\sigma \cpar \bang\nu n.\cout{c}{n} \cpar \bang d(x)\right)
\mathrel{\mathcal{R}^{\rho}}
\mathopen{\nu \vec{x}.}\left( \theta \cpar \left( \left(T_1 \cpar R_1\right) \cpar \ldots  
 \left(T_n \cpar R_n\right) \cpar \SpecMin^F \right)\right) 
\]
The above is defined with respect to two arbitrary partitionings of $\left\{1,\ldots n\right\}$, 
$\phi^0, \phi^1, \phi^2, \phi^3$ and $\psi^0, \psi^1, \psi^2$ representing the various stages of execution of each thread. In the above, $T_i$ and $R_i$ are defined as follows, for any messages $Y_i$, $NT_i$.
\[
T_i = \left\{
\begin{array}{lr}
\MainUK^{F}(c,d,k_i)
 & \mbox{if $i \in \phi^0$}
\\
d(y).\match{nt_i = \snd{\dec{y}{k_i}}}
\cout{c}{\enc{\pair{nt_i}{\fst{\dec{y}{k_i}}}}{k_i}}
 & \mbox{if $i \in \phi^1$}
\\
\match{nt_i = \snd{\dec{Y_i}{k_i}}}
\cout{c}{\enc{\pair{nt_i}{\fst{\dec{Y_i}{k_i}}}}{k_i}}
 & \mbox{if $i \in \phi^2$}
\\
0
 & \mbox{if $i \in \phi^3$}
\end{array}
\right.
\]
\[
R_i = \left\{
\begin{array}{lr}
\Reader^{F}(c,d,k_i)
 & \mbox{if $i \in \psi^0$}
\\
\nu nr.\cout{c}{\enc{\pair{nr_i}{NT_i}}{k_i}}
 & \mbox{if $i \in \psi^1$}
\\
0
 & \mbox{if $i \in \psi^2$}
\end{array}
\right.
\]
Names $\vec{x}$ are such that if $i \in \phi^1 \cup \phi^2 \cup \phi^3$ then $k_i, nt_i \in \vec{x}$,
if $i \in \psi_1 \cup \psi_2$ then $k_i \in \vec{x}$, and if $i \in \psi_2$ then $nr_i \in \vec{x}$
and each name is unique.
Substituition $\theta$ is defined such that
$\alia^0_i\theta =
\enc{\pair{nr_i}{NT_i}}{k_i}$
whenever
$i \in \psi^2$,
$\alia^1_i\theta =
nt_i$
whenever
$i \in \phi^1 \cup \phi^2 \cup \phi^3$,
$\alia^2_i\theta =
\enc{\pair{nt_i}{\fst{\dec{Y_i}{k_i}}}}{k_i}$
whenever
$i \in \phi^3$,
and 
$\alia \theta = \alia$ otherwise.
Substitution $\rho$ is a bijection between aliases.
Substitution $\sigma$ is injective and if $\alia \in \dom{\theta}$ then $\left(\alia\rho\right)\sigma$ is a variable.
Furthermore, $\vec{y} = \ran{\sigma}$, that is, every ciphertext or nonce in $\theta$ corresponds to a unique nonce in $\sigma$.

A similar construction covering the execution space of $\SystemMin^{F}$ defines a bisimulation establishing
$\SystemMin^{F} \bisimi{i} \bang\nu n.\cout{c}{n} \cpar \bang d(x)$.

From
$\SystemMin^{F} \bisimi{i} \bang\nu n.\cout{c}{n} \cpar \bang d(x)$
and
$\bang\nu n.\cout{c}{n} \cpar \bang d(x) \bisimi{i} \SpecMin^{F}$,
the result then follows
 by transitivity of $i$-bisimulation.
\end{proof}
The bisimulation in the above proof overfits the problem by allowing
inputs of information from the future and also covering some states
that cannot be reached, e.g., when there are not enough $0$'s or if a guard fails. However, since the initial state is covered and even the extra processes do not break the game described previously, the relation is a bisimulation witnessing the bisimilarity instance.

The proof suggests that interleaving semantics do not always provide a reasonable model of unlinkability, 
since the model effectively collapses and no information about the protocol is retained in the model.
This is a compelling reason to adopt a tighter semantics
such as HP-bisimilarity, which we introduce in the next section, rather than $i$-similarity.

\section{History-preserving semantics uncover attacks, regardless of style}
\label{sec:HP}

From the previous section, one might conclude that $i$-bisimilarity is sufficient to find attacks on some but not all protocols.
In particular, for a quite reasonable style of modelling unlinkability
protocols, such as Feldhofer, where 
all outputs are of the same form,
where $i$-bisimilarity misses attacks that
are evident for other models of unlinkability considered.
We elaborate on and address this observation in this section
by showing that a still finer, non-interleaving notion of bisimilarity
does not miss attacks regardless of the stylistic choices made
when modelling unlinkability.
This has the advantage that the semantics can work harder for the programmer,
giving them more liberty to choose the style in which they model unlinkability without missing attacks.

Towards this aim,
we introduce history-preserving bisimilarity (abbreviated HP-bisimilarity)
for the applied $\pi$-calculus,
and substantiate formally the above claim that it handles appropriately unlinkability problems.
To formalise attacks, we develop a logical characterisation of
HP-bisimilarity using a refinement of the modal logic classical $\FM$ from \autoref{sec:modal}.

\subsection{History-preserving bisimilarity: tighter semantics for tighter protocols}

History-preserving semantics take into account the partial order of causal dependencies between events.
We make use of a device to pair events that fire 
from the two extended processes compared.
In this work, this is done using a relation over events.
We define some simple auxiliary functions to work with relations and sets of events. 
\begin{defi}[Auxiliary functions]
	Given a relation $\ESS$ over events,
	we write $\dom{\ESS}$ and $\ran{\ESS}$ the sets of events forming the domain and range of $\ESS$, respectively.
	Given an event $e$ and set of events $E$ we write $e \Indy E$ whenever for all $e' \in E$, $e \Indy e'$ holds, and \(e \notIndy E\) if for all \(e' \in E\), \(e \Indy e'\) does not hold\footnote{The careful reader will have noted that the negation of $e \Indy E$ is \emph{not} \(e \notIndy E\): in the former case, it suffices to find a single event in \(E\) that is not independent with \(e\).}. 
\end{defi}

HP semantics preserve independence and also dependence at every step made by both processes.
Technically, this is achieved by partitioning the relation representing concurrently started events $\ESS$ according to the firing event $\event{\pi}{u}$ into $\ESS_1$ containing the events that are independent of the current event (\ie $\event{\pi}{u} \Indy \dom{\ESS_1} $) and $\ESS_2$ containing the events that are not independent (\ie $\event{\pi}{u} \notIndy \dom{\ESS_2}$).
Thus, $\dom{\ESS_2}$ is the minimal set of events that must have terminated before the new event can proceed.
This partitioning must be reflected by the matching transition on the right, thereby preserving both independence and dependence.
Since only the independent events and the new event are retained at the next step, the relation over events always consists of independent events.
The differences between $i$-similarity (\autoref{def:i-sim}) and HP-similarity are
\diff{highlighted} below.

\begin{defi}[HP-similarity]\label{def:hp-sim}
	A relation $\mathcal{R}$ is an \emph{HP-simulation}
	whenever if $A \mathrel{\mathcal{R}^{\rho, \diff{\ESS_1 \cup {\ESS_2}}}} B$,
	such that $\rho \colon \dom{A} \rightarrow \dom{B}$ is a bijection and \diff{$\ESS_1 \cup {\ESS_2}$ is a relation over events}, then:
	\begin{itemize}
		\item If $A \dlts{\pi}{u} A'$ and $\pi$ is not an output, \diff{$\event{ \pi }{ u } \Indy \dom{\ESS_1}$} and 
		\diff{$\event{ \pi }{ u } \notIndy \dom{\ESS_2}$}, then 
		there exist $B'$, $u'$, \st 
		$B \dlts{\pi\rho}{u'} B'$, 
		\diff{$\event{ \pi\rho }{ u' } \Indy \ran{\ESS_1}$},
		$\diff{\event{ \pi\rho }{ u' } \notIndy \ran{\ESS_2}}$,
		and	
		$A' \mathrel{\mathcal{R}^{\rho, \diff{\ESS_1 \cup {\left\{\left( \event{\pi }{ u } , \event{ \pi\rho }{ u' 		}\right)\right\}}}}} B'$.
 		\item
		If $A \dlts{\co{M}(\alpha)}{u} A'$, \diff{$\event{ \co{M}(\alpha) }{ u } \Indy \dom{\ESS_1}$} and 
		\diff{$\event{ \co{M}(\alpha) }{ u } \notIndy \dom{\ESS_2}$}, then there exist $B'$, $u'$, and $\alpha'$ \st 
		$B \dlts{\co{M\rho}(\alpha')}{u'} B'$, \diff{$\event{ \co{M\rho}(\alpha') }{ u' } \Indy \ran{\ESS_1}$},
		$\diff{\event{ \co{M\rho}(\alpha') }{ u' } \notIndy \ran{\ESS_2}}$, and $A' \mathrel{\mathcal{R}^{\rho\left[\alpha \mapsto \alpha'\right], \diff{\ESS_1 \cup {\left\{\left( \event{\co{M}(\alpha) }{ u } , \event{ \co{M\rho}(\alpha') }{ u' }\right)\right\}}}}} B'$
		\item $A \vDash M = N$ iff $B \vDash M\rho = N\rho$.
	\end{itemize}
	We say that an extended process $A$ \emph{is HP-simulated by} $B$, and write $A \simi{HP} B$, whenever
	there exists a HP-simulation $\mathcal{R}$ such that $A \mathrel{\mathcal{R}^{\id, \diff{\emptyset}}} B$.
	If, in addition, $\mathcal{R}$ is symmetric,
	(defined such that $A \mathrel{\mathcal{R}^{\rho, \ESS}} B$ iff $B \mathrel{\mathcal{R}^{\rho^{-1}, \ESS^{-1}}} A$)
	then $A$ and $B$ are HP-bisimilar, written $A \bisimi{HP} B$. 
\end{defi}

Notice how in the continuation of the simulation game the $\ESS_2$ is dropped, and the new event pair $( \event{\pi }{ u } , \event{ \pi\rho }{ u' })$ is added to $\ESS_1$. We are thus always continuing only with the concurrent events (whereas the $\ESS_2$ become just dependencies).
This is different from previous definitions of HP on event structures~\cite{Glabeek1989,Bednarczyk1991} or process algebras~\cite{Aubert2020b}, which generally requires extending mappings on \emph{all} events, including the ones that have been terminated.
This style was first introduced in previous work~\cite{Aubert2022k} specifically for the applied \(\pi\)-calculus and it is conjectured to match existing definitions on \eg CCS, event structures or behaviour structures when mapped to those systems.

\begin{example}\label{eg:small}
We consider some examples.
The extended processes
\begin{align*}
\id \cpar  \nu x, y, z. (\cout{a}{x} . (\cout{b}{y} \cpar \cout{c}{z})) &&& \text{and} &&\id \cpar  \nu x, y, z. (\cout{a}{x} . \cout{b}{y} \cpar \cout{c}{z})
\shortintertext{and infinite extended processes}
\id \cpar  \bang(\mathopen{\nu x.}\cout{a}{x} . \mathopen{\nu x.}\cout{a}{x}) &&& \text{and} && \id \cpar  \bang(\mathopen{\nu x.}\cout{a}{x})
\end{align*}
are related by $i$-similarity but not by HP-similarity.\footnote{The above examples also separate HP-similarity from other notions of similarity, notably ST-similarity as explained in 
detail in a companion paper~\cite[sec.4.3]{Aubert2022k}.
}
Even if the first example is enough to separate HP-similarity from coarser equivalences, 
we find it instructive to also look at an infinite, albeit simple looking, example.

For a still simpler example,
proving that HP-similarity is strictly finer than
$i$-similarity, observe that,
\[
\id \cpar {\cout{a}{a}.\cout{a}{a}} \nsimi{HP} \id \cpar {\cout{a}{a}} \cpar \cout{a}{a}
\]
\noindent holds (nor does the converse), despite 
the two processes being i-bisimilar.
\end{example}

By the above examples, HP semantics are strictly finer than interleaving semantics.
The following lemma also clarifies that HP-bisimilarity is sound with respect to $i$-bisimilarity
 (and hence observational equivalence also).
\begin{lem}\label{lem:HP-i}
	If $A \bisimi{HP} B$ then $A \bisimi{i} B$.
	Similarly, 
	if $A \simi{HP} B$ then $A \simi{i} B$.
	The converse of these implications do not hold in general.
\end{lem}
\begin{proof}
  Any two processes related by an HP-(bi)similation are related by an
  $i$-(bi)simulation since the definition of i-(bi)similarity relaxes
  the conditions of HP-(bi)similarity. By this observation, it follows
  immediately by definition, that any HP-(bi)simulation defines an
  $i$-(bi)simulation, by dropping the relation over events, and thus
  HP-(bi)similarity is contained in i-(bi)similarity.
	Strictness follows from any of the distinguishing examples above.
\end{proof}

\subsection{A modal logic characterising HP-bisimilarity}%
\label{ssec:modal-logic}

A novel contribution of this paper is the logical characterization of HP-bisimilarity by a simple adaptation of the classical $\FM$ logic~\cite{Horne2021}.
While this is the first such modal logic tailored to the applied $\pi$-calculus,
several modal logics have been proposed in the literature for capturing HP semantics, \eg~\cite{Nielsen1994,Baldan2014a,Phillips2014}.
These logics are defined on concurrency models different than ours,
yet strongly related. For instance, \cite{Phillips2014} interprets their Event Identifier Logic on (stable) configuration structures, which are
 more expressive than (prime) event structures on which the logic of \cite{Baldan2014a,Baldan2020ACMmodelChecking} is interpreted. The logics of \cite{Nielsen1994}, which are defined in the setting of bisimulations from open maps \cite{Joyal1996b}, are interpreted over transition systems with independence, which are the most close to our labelled asynchronous transition systems \cite{hildebrandt1996comparing}.

Studying the precise relationships of the logic from this section with other logics such as the ones above is not justified by our more practical goal here, \ie 
we work on a specific syntax (that of the applied $\pi$-calculus), with a specific semantics (namely, directly on the SOS rules from \autoref{sec:SOS}, which eventually generate a LATS~\cite{Aubert2022e}), with our specific way of defining (bi)similarities meant to be implemented in security tools.
For us, the logic is meant as a tool for describing attacks, whereas the above works investigate more generally logics for concurrency models.
Our approach here is novel in the sense that the syntax of classical $\FM$ barely changes (we simply move from actions to events) and instead we enrich the semantics of the diamond modality.
This is in keeping with the philosophy of the paper: that the languages
and tools should stay the same for the programmer and the semantics should work harder for us.

The novelty, however, is not in our annotation of the modality with events, as all the above logics do this in one form or another. For example, both logics of \cite{Phillips2014} and \cite{Baldan2020ACMmodelChecking} have modalities that execute events, but also bind events so they can later be referred back to. The novelty here is in the simplicity of our extension of a logic that already exists for the applied $\pi$-calculus. And moreover, this simplicity is in the syntax, \ie for the user (the modeling and verification expert), whereas the semantics becomes more complex to match the definition of the HP-bisimilarity.

\begin{defi}
The syntax of history-preserving $\FM$ ($\HPFM$) is given by the grammar on the left below 
and common abbreviations to the right,
where $\pi$ is either of the form $\tau$, $M\,N$ or $\co{M}(x)$, where $x$ is a variable rather than an alias.
\begin{multicols}{2}
	\noindent
	\begin{align*}
		\phi, \psi \Coloneqq &&& \ttt \tag*{top} \\
		\mid &&& M = N &&& \tag*{equality} \\
		\mid &&& \phi \wedge \psi \tag*{conjunction} \\
		\mid &&& \neg \phi \tag*{negation} \\
		\mid &&& \ediam{\pi}{u}\phi & \tag*{diamond} 
	\end{align*}
	\begin{align*}
		\fff &\triangleq \neg \ttt\\
		M \neq N &\triangleq \neg M = N\\
		\phi \vee \psi &\triangleq \neg\left(\neg \phi \wedge \neg \psi \right) \\
		\phi \yields \psi &\triangleq \neg \phi \vee \psi\\
		\eboxm{\pi}{u}\phi &\triangleq \neg \ediam{\pi}{u}\neg \phi
	\end{align*}
\end{multicols}
\noindent Formulas of the type $\ediam{\co{M}(x)}{u}\phi$ and
$\eboxm{\co{M}(x)}{u}\phi$
bind occurrence of $x$ in $\phi$.
The \emph{simulation} fragment of $\HPFM$ is that where negation is applied only to equality.
\end{defi}

We draw attention to the novel features of this logic.
The semantics of our modal logic, defined in \autoref{fig:modal-2}, has the satisfaction relation \(\vDash^{\ESS}\) parametrized by a relation over events \(\ESS\).
This relation allows us to track causal relations between the event in the modalities \(\ediam{\pi}{u}\) and \(\eboxm{\pi}{u}\) and other events.
In particular, we are interested in the causal dependencies between the event \(\event{\pi}{u}\) inside the modality and a subset of \(\ran{\ESS}\) of events that are concurrent with it.

\begin{remark}
In \autoref{fig:modal-2} you will see that a location $u$ appears in the events labelling transitions while a different location $u'$ appears in the modality. This is intentional since permitting a different location in the formula ensures that we are not simply syntactically pairing events in the formula with events in the transition system. Instead, we are pairing up events that preserve the same causal relationships with concurrent events.

This is essential for $\HPFM$ to characterise HP-bisimilarity. For example,
we expect that still a simple example like \autoref{ex:commutative} should satisfy the same formulas when we move to HP semantics.
In $\HPFM$, modalities range over events so we must give locations in the formula, which may happen to match the locations of events of
one process, as follows.
\begin{align*}
\id \cpar {\cout{a}{c}} \cpar {\cout{b}{h(c)}}
& \vDash^{\emptyset}
\ediam{\lout{a}{x}}{\prefix{0}\lambda}\ediam{\lout{b}{y}}{\prefix{1}\lambda}(h(x) = y)
\end{align*}
However, the same formula must also be satisfied by the process where the locations differ,
as in the following.
\begin{align*}
\id \cpar {\cout{b}{h(c)}} \cpar {\cout{a}{c}}
&\vDash^{\emptyset}
\ediam{\lout{a}{x}}{\prefix{0}\lambda}\ediam{\lout{b}{y}}{\prefix{1}\lambda}(h(x) = y)
\end{align*}
Notice, in the above, that the events in the formula and the events of the process are labelled with different locations.
However, both events in the process are independent as are both events in the formula.
Thus the labelling is correct with respect to the definition of satisfaction in \autoref{fig:modal-2}.

This observation is essential to understand the location labels in events also for all interesting examples of unlinkability problems, such as those that appear later in this section. This is because the structure of locations differs between the specification  and system, due to the differences in how replication is nested. Hence locations in a formula cannot possibly refer directly to events relevant to both processes.
\end{remark}

	\begin{figure}
	{
		\newcommand{\tablespace}{\\[3em]} 
		\newcommand{\hypospace}{\hspace{1.5em}} 
		\newcommand{\columnspace}{1.5em} 
		\begin{tabular}{c@{\hskip \columnspace} c}
			\begin{prooftree}
				\justifies
				A \vDash^{\ESS} \ttt 
			\end{prooftree}
			& 
			\begin{prooftree} 
				A\sub{x}{y} \vDash^{\ESS} \phi
				\hypospace
				\isfresh{y}{\nu x.A, \phi}
				\justifies
				\nu x.A \vDash^{\ESS} \phi
			\end{prooftree}
			\tablespace
			\begin{prooftree}
				M\theta \mathrel{=_E} N\theta
				\justifies
				\theta \cpar P \vDash^{\ESS} M = N
			\end{prooftree}
			& 
			\begin{prooftree}
				A \vDash^{\ESS} \phi
				\qquad
				A \vDash^{\ESS} \psi
				\justifies
				A \vDash^{\ESS} \phi \wedge \psi
			\end{prooftree}
			\tablespace
			\multicolumn{2}{c}{
				\begin{prooftree}
					B \mathrel{\vDash^{\ESS_1 \cup {\left\{( \event{ \co{M}(\alpha) }{ u } , \event{ \co{M}(\alpha) }{ u' })\right\}}}} \phi\sub{x}{\alpha}
					\hypospace
					A \dlts{\co{M}(\alpha)}{u} B 
					\hypospace
					\isfresh{\alpha}{\phi}
					\hypospace
					\begin{array}[b]{l}
						\event{ \co{M}(\alpha) }{ u } \Indy \dom{\ESS_1}
						\\
						\event{ \co{M}(\alpha) }{ u' } \Indy \ran{\ESS_1}
						\\
						\event{ \co{M}(\alpha) }{ u } \notIndy \dom{\ESS_2}
						\\
						\event{ \co{M}(\alpha) }{ u' } \notIndy \ran{\ESS_2}
					\end{array}
					\justifies
					A \mathrel{\vDash}^{\ESS_1 \cup {\ESS_2}} \ediam{\co{M}(x)}{u'}\phi
				\end{prooftree}
			}
			\tablespace
			\multicolumn{2}{c}{
				\begin{prooftree}
					B \mathrel{\vDash^{\ESS_1 \cup {\left\{( \event{ \pi }{ u } , \event{ \pi }{ u' })\right\}}}}
					\phi
					\hypospace
					A \dlts{\pi}{u} B
					\text{ with }\pi\text{ not output}
					\hypospace
					\begin{array}[b]{l}
						\event{ \pi }{ u } \Indy \dom{\ESS_1}
						\\
						\event{ \pi }{ u' } \Indy \ran{\ESS_1}
						\\
						\event{ \pi }{ u } \notIndy \dom{\ESS_2}
						\\
						\event{ \pi }{ u' } \notIndy \ran{\ESS_2}
					\end{array}
					\justifies
					A \mathrel{\vDash}^{\ESS_1 \cup {\ESS_2}} \ediam{\pi}{u'}\phi
				\end{prooftree}
			}
			\\[4.4em]
			\hline
			&\hfill \textsc{negation}
			\\[1.4em]
			\multicolumn{2}{c}{
				\begin{tabular}{c@{\hskip \columnspace} c@{\hskip \columnspace} c@{\hskip \columnspace} c@{\hskip \columnspace} c@{\hskip \columnspace}}
					\begin{prooftree}
						A \nvDash^{\ESS} \phi
						\justifies
						A \vDash^{\ESS} \neg\phi
					\end{prooftree}
					& 
					\begin{prooftree}
						A \vDash^{\ESS}  \phi
						\justifies
						A \nvDash^{\ESS}  \neg\phi
					\end{prooftree}
					& 
					\begin{prooftree}
						M\theta \mathrel{\neq_E} N\theta
						\justifies
						\theta \cpar P \nvDash^{\ESS} M = N
					\end{prooftree}	
					& 
					\begin{prooftree}
						A \nvDash^{\ESS} \phi
						\justifies
						A \nvDash^{\ESS} \phi \wedge \psi
					\end{prooftree}
					&
					\begin{prooftree}
						A \nvDash^{\ESS} \psi
						\justifies
						A \nvDash^{\ESS} \phi \wedge \psi
					\end{prooftree}
				\end{tabular}
			}
			\tablespace
			\multicolumn{2}{c}{
				\begin{prooftree}
					\left\{
					B \mathrel{\nvDash^{\ESS_1 \cup {\left\{( \event{ \co{M}(\alpha) }{ u } , \event{ \co{M}(\alpha) }{ u' })\right\}}}} \phi\sub{x}{\alpha}
					\colon
					\begin{array}{l}
						A \dlts{\co{M}(\alpha)}{u} B
						\hypospace 
						\isfresh{\alpha}{\phi}
						\\
						\event{ \co{M}(\alpha) }{ u } \Indy \dom{\ESS_1}
						\\
						\event{ \co{M}(\alpha) }{ u } \notIndy \dom{\ESS_2}
					\end{array}
					\right\}
					\hypospace
					{
						\begin{array}[b]{l}
							\event{ \co{M}(x) }{ u' } \Indy \ran{\ESS_1}
							\\
							\event{ \co{M}(x) }{ u' } \notIndy \ran{\ESS_2}
					\end{array}}
					\justifies
					A \mathrel{\nvDash}^{\ESS_1 \cup {\ESS_2}} \ediam{\co{M}(x)}{u'}\phi
				\end{prooftree}
			}	
			\tablespace\tablespace
			\multicolumn{2}{c}{
				\begin{prooftree}
					\left\{
					B \mathrel{\nvDash^{\ESS_1 \cup {\left\{( \event{ \pi }{ u } , \event{ \pi }{ u' })\right\}}}}
					\phi
					\colon
					\begin{array}{l}
						A \dlts{\pi}{u} B,
						\\
						\event{ \pi }{ u } \Indy \dom{\ESS_1}
						\\
						\event{ \pi }{ u } \notIndy \dom{\ESS_2}
					\end{array}
					\right\}
					\hypospace
					{\begin{array}{l}
							\pi\text{ not output}
							\\
							\event{ \pi }{ u' } \Indy \ran{\ESS_1}
							\\
							\event{ \pi }{ u' } \notIndy \ran{\ESS_2}
					\end{array}}
					\justifies
					A \mathrel{\nvDash}^{{\ESS_1} \cup {\ESS_2}} \ediam{\pi}{u'}\phi
				\end{prooftree}
			}
		\end{tabular}
}
		\caption{Rules for history-preserving modal logic classical $\HPFM$, where modalities are labelled with events.
	}
\label{fig:modal-2}
\end{figure}

The semantics is designed such that we obtain a 
Hennessy-Milner property, much like \autoref{thm:i-HM},
that proves that our new modal logic can be used to characterise
HP-bisimilarity.
\begin{thm}\label{thm:HP-HM}
	The following holds for extended processes $A$, $B$:
		$A \bisimi{HP} B$ iff, we have that, for all $\HPFM$ formulae $\phi$, $A \vDash_{HP}^{\emptyset} \phi$ iff $B \vDash_{HP}^{\emptyset} \phi$.
\end{thm}

A proof of a theorem from which the above follows immediately appears in \autoref{app:proofs}, where the proof involves composing the relations between events in bisimulations and satisfaction relations when moving between the perspective of different processes with respect to the same formula.
In the definition of the rules of $\HPFM$, the relation between events in the process and formula are essential, since two bisimilar processes may label events differently yet be HP-bisimilar; and hence must satisfy the same formulas. Thus, in such cases the relationship between the events in a formula and the process may differ from process to process and so should be recorded. What is important is the preservation of the independence of events between the formula and process rather than their locations.

We revisit some example processes that are distinguished by HP-bisimilarity,
and provide a modal logic formula distinguishing each property.

\begin{example}\label{eg:hm-hp}
We can prove the inequalities
of \autoref{eg:small} 
via the observation that, for the finite example,
\begin{align*}
\id \cpar  (\nu x, y, z. (\cout{a}{x} . (\cout{b}{y} \cpar \cout{c}{z})))
& \vDash^{\emptyset}
\ediam{\lout{a}{x} }{ \prefix{0}[]}
\ediam{\lout{c}{z} }{ \prefix{01}[]}\ttt\\
\shortintertext{but}
\id \cpar (\nu x, y, z. (\cout{a}{x} . \cout{b}{y} \cpar \cout{c}{z}))
& \nvDash^{\emptyset}
\ediam{\lout{a}{x} }{ \prefix{0}[]}
\ediam{\lout{c}{z} }{ \prefix{01}[]}\ttt\text{,}
\shortintertext{and, for the infinite example,}
\id \cpar (\bang(\mathopen{\nu x.}\cout{a}{x} . \mathopen{\nu x.}\cout{a}{x}))
& \vDash^{\emptyset}
\ediam{\lout{a}{x} }{ \prefix{0}[]}\ediam{\lout{a}{x} }{ \prefix{0}[]}\ttt
\shortintertext{but}
\id \cpar (\bang(\mathopen{\nu x.}\cout{a}{x}))
& \nvDash^{\emptyset}
\ediam{\lout{a}{x} }{ \prefix{0}[]}\ediam{\lout{a}{x} }{ \prefix{0}[]}\ttt\text{.}
\end{align*}
We then conclude by appealing to \autoref{thm:HP-HM}.
\end{example}

\begin{example}[Preserving independence]\label{eg:hm-hp-2}
Further to the previous examples 
consider the following, adapted here from classic work~\cite{Glabbeek1997,aceto1994adding,GORRIERI1995272,Vogler1996} to the $\pi$-calculus.
There is an 
interleaving simulation relating the following processes.
\begin{multline*}
\id \cpar (\mathopen{\nu c, d.}\left( \left( {\cout{d}{d}} \cpar {\nu n.\cout{a}{n}.\cin{d}{z}.\cin{n}{x}} \right) \cpar \left( {\cout{c}{c}} \cpar {\cin{c}{y}}\right)\right)) \\
\simi{i}
\id \cpar (\mathopen{\nu e, f.}\left( \left( {\cout{f}{f}} \cpar {\nu n.\cout{a}{n}.\cin{f}{z}} \right) \cpar \left( {\cout{e}{e}} \cpar  {\cin{e}{y}.\cin{n}{x}} \right)\right))
\end{multline*}
In contrast, the above inequality does not hold with respect to HP-similarity.
Observe that formula
$
\phi \triangleq
\ediam{\tau}{\left(\prefix{10}[], \prefix{11}[]\right)}
\ediam{\lout{a}{x}}{\prefix{01}[]}
\ediam{\tau}{\left(\prefix{00}[], \prefix{01}[]\right)}
\ediam{x\,M}{\prefix{01}[]}
\ttt
$
is such that
\begin{align*}
\id \cpar ( \mathopen{\nu c, d.}\left( 
  \left( {\cout{d}{d}} \cpar {\nu n.\cout{a}{n}.\cin{d}{z}.\cin{n}{x}} \right) \cpar \left( {\cout{c}{c}} \cpar {\cin{c}{y}} \right)
\right))
& \vDash^{\emptyset}
 \phi
\shortintertext{holds, but, in contrast,}
\id \cpar (\mathopen{\nu e, f.}\left(
  \left( {\cout{f}{f}} \cpar {\nu n.\cout{a}{n}.\cin{f}{z}} \right) \cpar \left( {\cout{e}{e}} \cpar  {\cin{e}{y}.\cin{n}{x}} \right)
\right))
&\nvDash^{\emptyset}
\phi\text{.}
\end{align*}
To help see why this strategy is distinguishing,
consider the path suggested by the formula through the graphical representation
of the state space below.

\noindent\includegraphics[width=\linewidth]
	{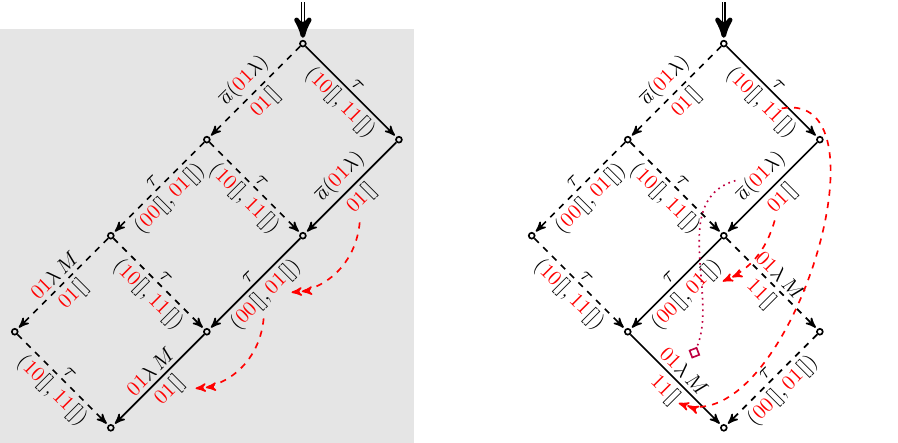}

The above represents the state spaces of the two processes, with the event transitions captured by the formula $\phi$ represented by solid lines with arrow tips (\tikz[baseline=-2pt]{\draw [transition, thick, >=stealth', shorten >=1pt](0, 0) --++(.5, 0);}), while the other possible transitions are dashed (\tikz[baseline=-2pt]{\draw [transition, dashed, thick, >=stealth', shorten >=1pt](0, 0) --++(.5, 0);}).
Relevant dependencies are also represented: dotted lines with square ends for causal dependencies (\tikz[baseline=-2pt]{\draw [causal](0, 0) --++(.5, 0);}), and dashed lines with double arrowheads (\tikz[baseline=-2pt]{\draw [structural](0, 0) --++(.5, 0);}) for the structural dependencies.
Dependencies that can be obtained by transitivity are also omitted, for readability, and the active player (in this case, the left-hand formula) is shaded.

On the left of the diagram above, observe that the rightmost $\tau$-transition (synchronising on channel $c$)
is independent from all other transitions,
while all other events in that diagram are causally dependent on each other.
In contrast, on the right above, both $\tau$-transitions are dependent on 
one other event,
and independent of the others. 
Since the formula reflects the independence structure on the left but not the right,
the formula is distinguishing.
Remember it is the independence structure of the locations not names of the locations that are important,
and the fact that the locations in the formula match the locations of the process on the left is just to aid readability,
and, furthermore, the difference in the location label of the input transitions is not in itself the reason that the strategy is distinguishing.
\end{example}

\begin{example}[Preserving dependencies]\label{eg:hm-hp-3}
For another classic example that we adapt for $\pi$-calculi consider the following.
\begin{equation*}
	\id \cpar (\mathopen{\nu a,b.}\left(
	\left(
	{\cout{a}{a}}
	\cpar 
	\left(
	{\cin{a}{x}}
	+
	{\cin{b}{x}}
	\right)
	\right)
	\cpar
	{\cout{c}{c}.\cout{b}{b}}
	\right))
	\nsimi{HP}
	\id \cpar (
	\mathopen{\nu a.}\left(
	\left(
	{\cout{a}{a}}
	\cpar 
	{\cin{a}{x}}
	\right)
	\cpar
	{\cout{c}{c}}
	\right)
	)
\end{equation*}
To see that they are unrelated by HP-similarity,
consider the following formula.
\[
\ediam{\co{c}(\prefix{1}\lambda)}{\prefix{1}[]}
\ediam{\tau}{(\prefix{01}[1], \prefix{1}[])}\ttt
\]
This is satisfied by the process on the left,
since it can 
perform the following transition.
\begin{align*}
   & \id
	\cpar (
	\mathopen{\nu a,b.}\left(
	\left(
	{\cout{a}{a}}
	\cpar \left(
	{\cin{a}{x}}
	+
	{\cin{b}{x}}
	\right)
	\right)
	\cpar
	{\cout{c}{c}.\cout{b}{b}}
	\right))
	\\
	& \hspace{3em} \dlts{\co{c}(\prefix{1}\lambda)}{\prefix{1}[]}
	\mathopen{\nu a,b.}\left(
	{\sub{\prefix{1}\lambda}{c}}
	\cpar
	\left(
	{\cout{a}{a}}
	\cpar \left(
	{\cin{a}{x}}
	+
	{\cin{b}{x}}
	\right)
	\right)
	\cpar
	{{\cout{b}{b}}}
	\right)
\end{align*}
Notice now that the process 
$\mathopen{\nu a,b.}\left(
{\sub{\prefix{1}\lambda}{c}}
\cpar
\left(
{\cout{a}{a}}
\cpar \left(
\cin{a}{x}
+
\cin{b}{x}
\right)
\right)
\cpar
{\cout{b}{b}}
\right)$
can perform a transition labelled with
$\event{\tau}{(\prefix{01}[1], \prefix{1}[])}$ (synchronising on channel $b$),
which is not independent from
$\event{\co{c}(\prefix{1}\lambda) }{ \prefix{1}[] }$.
Yet, although the other process can perform a $\tau$-transition
at the respective point,
it cannot match that dependency constraint:

\noindent
\includegraphics[width=\linewidth]{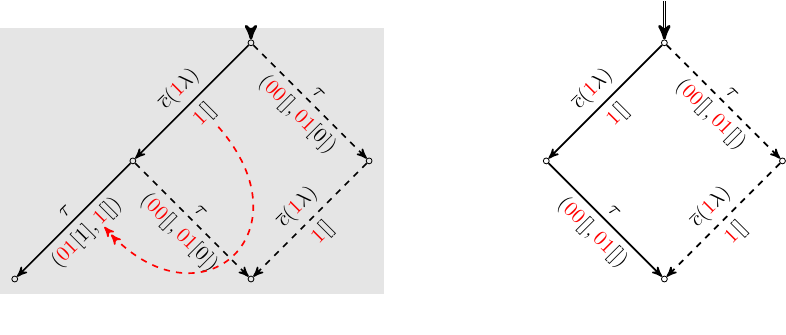}

In fact, the matching of dependency constraints is necessary for this
distinguishing strategy, and hence preserving independence only is insufficient.
\end{example}

Notice that the above formulas employ only diamond modalities and
hence these examples are all interpreted as 
describing a pomset trace.
There are, of course, examples that employ other features of the logic,
when branching time is relevant for example, as we illustrate using
the Feldhofer protocol from \autoref{fig:msc-feld}.

\subsection{Analysis of Feldhofer with silent authentication failure}

We are now able to illustrate formally 
how
HP-bisimilarity 
reliably detects attacks on
 unlinkability,
as mentioned in the technical preliminaries (\autoref{sec:tech_preliminry_intro}).
Furthermore,
we illustrate that our methodology
allows us to use $\HPFM$
to provide
 descriptions of attacks when
unlinkability does not hold,
for any formulation of unlinkability we have considered.
Thus we have a semantics that 
 is robust with respect
to both the formulation of unlinkability and the choice of protocol. 
\begin{thm}\label{thm:Feldhofer-attack}
$
\SystemMin^{F}
\nbisimi{HP}
\SpecMin^{F}
$, 
where, as previously,
\[
\begin{array}{rl}
\SystemMin^F \triangleq& \bang\mathopen{\nu k.\bang}\left( \Reader^F(c,d,k) \cpar \MainUK^F(c,d,k) \right)
\\\\
\SpecMin^F \triangleq& \bang\mathopen{\nu k.}\left( \Reader^F(c,d,k) \cpar \MainUK^F(c,d,k) \right)
\end{array}
\]
with \(\Reader^F\) and \(\MainUK^F\) as defined on page~\pageref{feldhoferf}.
\end{thm}
\begin{proof}
Consider the following event-labelled
$\FM$ formula involving independent locations $\locp$, $\locra$, $\locrb$, and $\locq$.\footnote{
Notice $\locp$, $\locra$, $\locrb$, and $\locq$
are structurally independent locations
that need not match the names of locations given by the structural operational semantics
to the events.
The use of the relation $\ESS$ in Def.~\ref{def:hp-sim}
is such that the choice of the name of the location is not important, but the structure, independence and causal dependencies induced are.
}
\[
\chi \triangleq
\begin{array}[t]{l}
\ediam{\co{c}(nt_1)}{\locp} \ediam{d\,nt_1}{\locra}\ediam{d\,nt_1}{\locrb}\ediam{\co{c}(u_1)}{\locra}\ediam{\co{c}(u_2)}{\locrb}\Big(
	\begin{array}[t]{l}
	\ediam{d\,u_1}{\locp}\ediam{\co{c}(w)}{\locp}\eboxm{d\,z}{\locq}\ediam{\co{c}(v)}{\locq}\ttt
		\\
		\wedge
		\\
	\ediam{d\,u_2}{\locp}\ediam{\co{c}(w)}{\locp}\eboxm{d\,z}{\locq}\ediam{\co{c}(v)}{\locq}\ttt
		~\Big)
	\end{array}
\end{array}
\]

Since
$\SystemMin^F \vDash^{\emptyset} \chi$
but
$\SpecMin^F \nvDash^{\emptyset} \chi$,
by \autoref{thm:HP-HM},
$\SystemMin^F
\nbisimi{HP}
\SpecMin^F$.
\end{proof}
Observe that the formula is shorter
than $\psi$ in the proof of \autoref{thm:system-min-silent},
 and, furthermore, this strategy would also serve as a distinguishing strategy for the BAC protocol
 with respect to HP-bisimilarity.
This is because we do not need to make use of the fact that some messages are pairs,
and instead we rely on causal dependencies between inputs that cause outputs in the same location.
In common with \autoref{thm:system-min-silent}, 
the presence of the box modality shows a non-trivial usage of alternation of player
coming from the fact that HP-bisimilarity is a symmetric relation.
It is not immediately obvious why the box modality is necessary here, and hence we illustrate the strategy in detail.

\begin{figure}
\includegraphics[width=\linewidth-.5em]{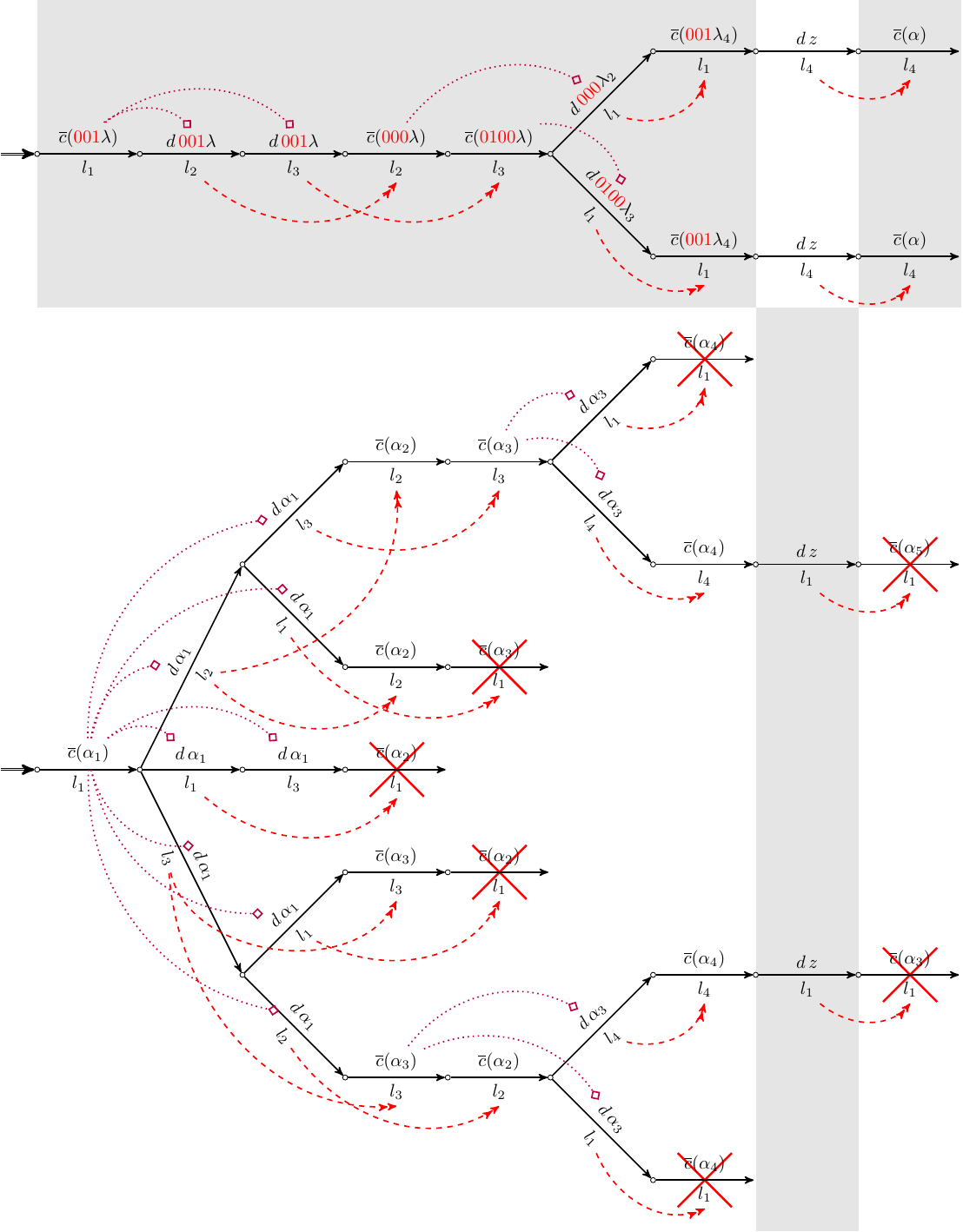}

\begin{tabular}{l c l l}
Role &	Location & Upper-diagram & Lower-diagram\\
\hline 
Prover &	\locA & $\prefix{001}[]$ & $\prefix{1^{i}01}[]$ \\
Verifier &	\locB  & $\prefix{000}[]$  (same keys) & $\prefix{1^{j}00}[]$ \\
Verifier &	\locC  & $\prefix{0100}[]$ (same keys)   & $\prefix{1^{k}00}[]$ (different keys) \\
Verifier &	\locD & $\prefix{1^{m}01^{n}00}[]$  & $\prefix{1^{h}00}[]$ \\
& $m \neq 0$ or $n > 1$,
&
$i,j \neq k$
and
$h \neq j,k$
\end{tabular}

\caption{
Graphical representation of the distinguishing strategy
described by 
$\chi$ from proof of (\autoref{thm:Feldhofer-attack}).
The top represents transitions of $\SystemMin^F$,
while the bottom represents transitions of $\SpecMin^F$. The shaded player is leading in the game.
}\label{fig:nice}
	\end{figure}

A graphical representation of the distinguishing strategy described by $\chi$ is presented in 
\autoref{fig:nice}.
This diagram uses the same graphical convention as the ones from Examples~\ref{eg:hm-hp-2} and \ref{eg:hm-hp-3}, but does not represent the (infinitely many) possible transitions that are not explored by the formula, for clarity.
The strategy begins with the system leading, in the top of the diagram, while the specification, at the bottom, has to follow events of the system.
The initial event-labelled transitions
performed by the system can be matched by the
tree of possible event transitions of the idealised specification, in the bottom of the diagram.
Notice that, in three of the branches, already after four or five events the specification cannot keep up
with the system, which corresponds to failed branches where the specification erroneously chooses to pair 
an event stemming from a verifier in the system with an event stemming from the prover in the specification.
Transitions that fail to be matched are emphasised in the diagram by explicitly crossing them out in red.
Importantly, since HP-bisimilarity ensures that the structural causality is respected,
the outputs that are crossed out are expected to be in the same location as the preceding input,
and this causal dependency is what cannot be matched by the specification.

There are two branches of the specification that can match the system for more than five events.
In each of these branches, either the first or second verifier, in location $\prefix{000}[]$ or $\prefix{0100}[]$ respectively, would fail to authenticate the prover,
since that verifier would possess the keys of a different prover.
At that moment, the system can choose either of the verifiers
in order to authenticate successfully,
as represented by the two branches of the state space in the top diagram.
Therefore, depending on which branch the specification has taken, in the bottom diagram, the system can play the branch
in the top diagram that the specification cannot keep up with indefinitely.

The strategy is non-trivial at the point where the system branches.
There are two ways (symmetric in both branches) for the specification to respond to the strategy of the system.
\begin{itemize}
\item 
In the shorter of the branches indicated at that point in \autoref{fig:nice},
the specification plays along literally, by feeding the input back into the original prover,
in $\ell_1$ in the diagram, at which point the prover blocks since it cannot authenticate, as indicated by the struck-out output transition, also in $\ell_1$, at that point.
\item
Since this minimal model of unlinkability does not distinguish between actions with provers and verifiers, and the causal dependencies rather than locations are recorded,
there is a longer strategy that specification can play, demanding features of bisimilarity.
The specification can mimic the communication pattern of a successfully authenticating prover by feeding the input into any read that will always respond.
This leads us to a state where the prover is still waiting for an input in the specification, but there is no active prover session on the side of the system.
A trick is then applied to distinguish these two states, described next.
\end{itemize}
The trick for distinguishing a state in which a prover awaits an input from one where there is no active prover is as follows.
\textit{We use the additional power of bisimilarity over similarity} -- that is, that the underlying bisimulation relation is symmetric allowing us to change who is leading in the strategy. The change of leading player in the strategy is indicated in \autoref{fig:nice} by the change of background shading.
At the point, a message, say $z$, is fed into the prover still active on the side of the system, which can be matched by the specification only by feeding $z$ into a verifier.
Any message in place of $z$ is appropriate, as long as it is not a prior output from a verifier loaded with the keys of the prover.
Notice in this penultimate state the prover receiving the input on the side of the specification fails to authenticate upon receiving $z$ and blocks silently, 
while the verifier on side of the system can still perform an output event in the given location.
Therefore, we can distinguish these states by changing the player again and allowing the system to lead with the final event that the system cannot match.
Notice that again structural causality, as preserved by HP-bisimilarity, is essential to link the final output with the penultimate input.

\paragraph{Reflection}
One may argue that some of the branches of this strategy are due to information missing in the model.
\Eg the location or the fact that a tag or verifier is involved would allow the strategy to be simplified significantly.
However, recall that a key objective of this work is to ensure that attacks are not missed even if
the programmer
misses some information when modelling unlinkability.
Thus the programmer need not have the insight to include such information, and yet HP-bisimilarity will indicate 
that the protocol is suspect; whether or not the programmer knows immediately how to interpret the attack,
which typically takes some additional effort.

The strategy presented is not unique of course. There are typically infinitely many strategies when one exists.
An alternative strategy that would prove the above theorem is as follows.
\begin{multline*}
	\chi' \triangleq
	\ediam{\co{c}(nt_1)}{\locp} \ediam{d\,nt_1}{\locra}\ediam{d\,nt_1}{\locrb}\ediam{\co{c}(u_1)}{\locra}\ediam{\co{c}(u_2)}{\locrb} 
	\Big( 
	\begin{array}[t]{l}
	\eboxm{d\,u_1}{\locq}\ediam{\co{c}(w)}{\locq}\ttt
	\\
	\wedge
	\\
	\eboxm{d\,u_2}{\locq}\ediam{\co{c}(w)}{\locq}\ttt
~\Big)
	\end{array}
\end{multline*}

The strategy above, that we illustrate in \autoref{fig:nice-2}, still features box modalities and hence requires the power of bisimilarity in order to be realised.
It is slightly shorter, since the strategy switches players earlier than the previous strategy.
This early switch forces the specification to take the branch where the prover fails to 
authenticate the verifier earlier.\footnote{For readers familiar with the linear-time/branching-time spectrum, who ask whether something coarser than bisimilarity would suffice, notice that in both strategies two changes of play are required, and hence failure similarity would not suffice for either of these strategies.}

\begin{figure}
\includegraphics[width=\linewidth]{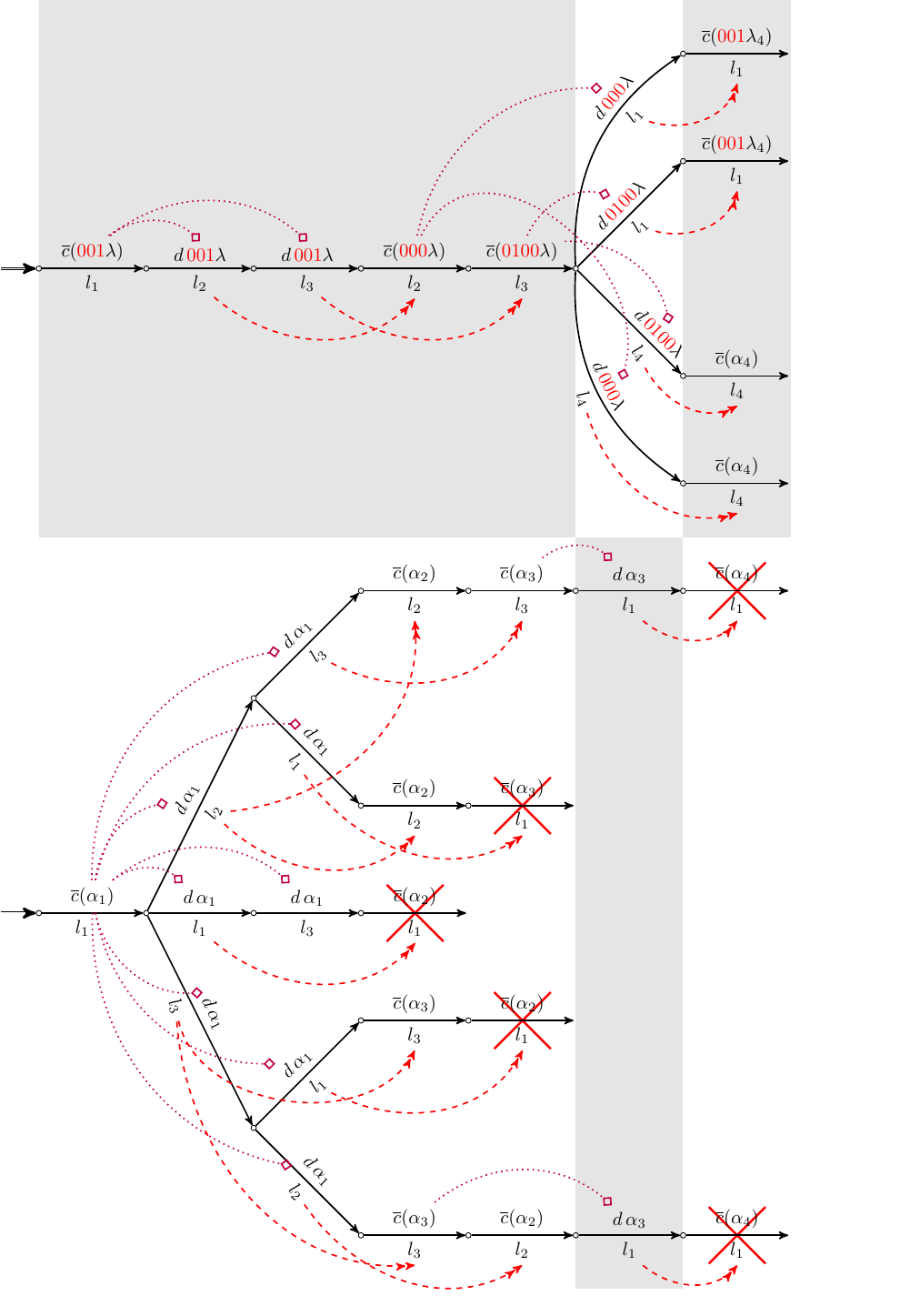}

\caption{
Graphical representation of an
alternative distinguishing strategy,
$\chi'$,
with an earlier change of player.
}\label{fig:nice-2}
\end{figure}

\subsection{HP-similarity misses attacks on unlinkability}

Having demonstrated that
HP-bisimilarity can discover attacks on Feldhofer where $i$-bisimilarity fails,
and having explained that a non-trivial feature of bisimilarity is employed,
a natural question is whether we can prove that HP-similarity really would not suffice.
This can be proven formally, as we do next. 
Recall the style we consider is where we have
 no constant message at the beginning of the 
protocol, we consider the unbounded setting, and we do not have
fresh channels for each new session of the protocol.

Our formal reference for Feldhofer is \autoref{thm:Feldhofer-attack}, showing that there is an attack 
on unlinkability with respect to HP-bisimilarity.
In addition, for BAC,
recall that in \autoref{thm:system-min-silent} we establish 
that 
$\SystemMin
\nbisimi{i}
\SpecMin$
from which it follows immediately, by \autoref{lem:HP-i} that
$\SystemMin
\nbisimi{HP}
\SpecMin$.
Thus for either protocol HP-bisimilarity can discover an attack.
In the theorem below, we show that, in contrast,
 HP-similarity cannot discover an attack on the same formulation of unlinkability
 for either protocol.
We consider BAC, since BAC gives more observables for the attacker to play with,
and hence the same proof lifts immediately to the Feldhofer protocol.
\begin{restatable}{thm}{systemHPtwo}\label{thm:sim-proof}
$
\SystemMin
\simi{HP}
\SpecMin
$
\end{restatable}
The proof appears in \autoref{app:proofs},
since it involves the construction of a large relation 
covering the state space of all behaviours of $\SystemMin$
and how each state may be related a state of $\SpecMin$.
The idea of the proof however is straightforward.
Match actions of almost any session in 
$\SystemMin$
with any session 
in $\SpecMin$, regardless of the keys,
with one exception that we explain.
The exception is at the point when a prover, say $\MainUK_i$
, in $\SystemMin$ receives an input
that will result in $\MainUK_i$ 
 being successfully authenticated.
At that point in an execution, the duplicator matches that input and the resulting
output from $\MainUK_i$ 
by using  
instead
an input and output from a
fresh verifier
session in $\SpecMin$, say $\Reader_{2i + 1}$. 
Since the input and output by a new session of $\Reader_{2i + 1}$ 
 and the final input and output of a successfully authenticating session of $\MainUK_i$ 
  are indistinguishable to an attacker,
the duplicator wins the game.
Furthermore, the constraints on independence and dependence imposed for HP-simularity
 are preserved at that moment when the execution of $\SpecMin$ deviates
from simply following the session in $\SystemMin$.
This is because the first input event of $\Reader_{2i + 1}$ 
 simulating
successfully authenticating $\MainUK_i$ 
is related by link causality to output events 
in the same way as for the corresponding events in  $\MainUK_i$. 

Appealing to \autoref{lem:HP-i}, we obtain as a corollary that 
$
\SystemMin
\simi{i}
\SpecMin
$, which is exactly \autoref{thm:system-min-interleaving} from \autoref{sec:tech_preliminry_intro}.
Thus, an immediate consequence is that we prove that 
there is a formulation of unlinkability of the BAC protocol 
for which simulation is not enough to find attacks, but bisimulation is,
and that moving to HP-semantics does not help us with respect to similarity.
Furthermore, the same applied to Feldhofer.

\begin{figure}
\noindent\includegraphics[width=\linewidth]{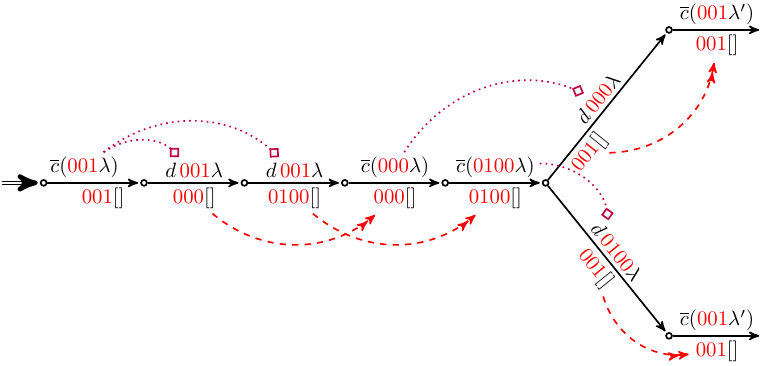}

Legend:
$\prefix{001}$ ePassport,
$\prefix{000}$ reader (same keys),
$\prefix{0100}$ reader (same keys).

\caption{
Graphical representation of a strategy
that is not distinguishing. 
}\label{fig:niceold}
	\end{figure}

To help understand why HP-similarity is insufficient,
consider the $\HPFM$ formula below,
which is graphically represented in \autoref{fig:niceold}.
\begin{multline*}
  \ediam{\co{c}(nt_1)}{\locp} \ediam{d\,nt_1}{\locra}\ediam{d\,nt_1}{\locrb}\ediam{\co{c}(u_1)}{\locra}\ediam{\co{c}(u_2)}{\locrb}\Big( 
	\begin{array}[t]{l}
	\ediam{d\,u_1}{\locp}\ediam{\co{c}(w)}{\locp}
	\ttt
	\\
	\wedge
	\\
	\ediam{d\,u_2}{\locp}\ediam{\co{c}(w)}{\locp}
	\ttt
	~\Big)
	\end{array}
\end{multline*}
The above formula can be satisfied by both the system and specification.
For the specification, a dummy reader session can be used to mimic the last two messages of a successfully authenticating ePassport
in both branches.

\section{Further results for protocol variants with error messages}
\label{sec:protocols}

We summarise here further results obtained during the investigation
of this work that further tighten the analysis by considering variants of protocols.
Thus far we considered a version of the BAC and Feldhofer protocols where the prover fails to authenticate silently,
without sending an error message.
We now address the obvious question of how robust these observations are with respect to varying the protocol.
We thus address here the more specific question of how the presence of an error message gives rise to different strategies.

To express this problem, we extend the non-interleaving structural operational semantics for the applied $\pi$-calculus from Sec.~\ref{sec:SOS} with \texttt{if}-\texttt{then}-\texttt{else}. Until now, we have used only the traditional match $\match{M=N}$ borrowed from the $\pi$-calculus, that behaves as $\texttt{if}\,M=N\,\texttt{then}$ without an \texttt{else}-branch. 
Rules presented in \autoref{fig:else}, in contrast to non-deterministic choice, do not affect the location. This is because, whether one branch is taken or not is determined by whether a guard is enabled or not, and hence we need not remember which branch was taken when permuting transitions.
Therefore, the branches of \texttt{if}-\texttt{then}-\texttt{else} need not be guarded in the same way they are for non-deterministic choices.
\begin{figure}
\[
	\begin{prooftree}
			M \mathrel{=_E} N
			\qquad
				P
				\dlts{\pi}{
					u
				}
				{A}
			\justifies
				\texttt{if}\,
				M = N
				\,\texttt{then}\,
				P
				\,\texttt{else}\,
				Q
				\dlts{\pi}{
					u
				}
				{A}
			\using
			\mbox{\textsc{Then}}
		\end{prooftree}
\qquad
	\begin{prooftree}
			M \mathrel{\neq_E} N
			\qquad
			Q
				\dlts{\pi}{
					u
				}
				{A}
			\justifies
				\texttt{if}\,
				M = N
				\,\texttt{then}\,
				P
				\,\texttt{else}\,
				Q
				\dlts{\pi}{
					u
				}
				{A}
			\using
			\mbox{\textsc{Else}}
		\end{prooftree}
\]
\caption{Rules for \texttt{if-then-else}.}\label{fig:else}		
\end{figure}

We consider Feldhofer only since, other than the error message,
the messages sent are uniform, forcing the attacker to exploit the error message.
Since, as explained in \autoref{sec:meaningless}, unlinkability became meaningless with respect to $i$-bisimilarity
for the silent version of Feldhofer, without an error message, one might believe it remains meaningless even if an explicit error message is introduced.
This turns out not to be the case, for rather surprising reasons.
The following theorem establishes that,
 with respect to $i$-bisimilarity,
unlinkability is violated for this formulation of Feldhofer
in the unbounded setting, and without either fresh channels or a constant message at the beginning.
\begin{thm}\label{thm:system-min-error}
$
\SystemMin^{\error}
\nbisimi{i}
\SpecMin^{\error}
$, 
where
\[
\begin{array}{rl}
\SystemMin^{\error} \triangleq& \bang\mathopen{\nu k.\bang}\left( \Reader^F(c,d,k) \cpar \MainUK^{\error}(c,d,k) \right)
\\\\
\SpecMin^{\error} \triangleq& \bang\mathopen{\nu k.}\left( \Reader^F(c,d,k) \cpar \MainUK^{\error}(c,d,k) \right)
\\\\
\MainUK^{\error}(c,d,k) \triangleq&
\begin{array}[t]{l}
 \nu nt.\cout{c}{nt}.d(y). \\
 \ifo nt = \snd{\dec{y}{k}}\,\theno  \\
    \hspace{1em} \cout{c}{\enc{\pair{nt}{\fst{\dec{y}{k}}}}{k}}
 \\
 \elseo \cout{c}{\error}
        \end{array}
\end{array}
\]
\end{thm}
\begin{proof}
Consider the following $\FM$ formula.
\[
\psi_{\ref{thm:system-min-error}}
\triangleq
\begin{array}[t]{l}
\diam{\co{c}(nt_1)}\diam{d\,nt_1}\diam{d\,nt_1}
\\
\diam{\co{c}(u_1)}\diam{\co{c}(u_2)}\Big(
\quad
\begin{array}[t]{l}
u_1 \neq \error
\wedge
u_2 \neq \error~\wedge
\\
\diam{d\,z}\diam{\co{c}(e)}\big(
   \begin{array}[t]{l}
   e = \error~\wedge 
   \\
   \boxm{d\,z}\boxm{\co{c}(v)}\left( v \neq \error \right)~\big)
   \end{array}
\\
\wedge
\\
\diam{d\,u_1}\diam{\co{c}(w)}\big(
\begin{array}[t]{l}
  w \neq \error~\wedge \\
  \boxm{d\,z}\boxm{\co{c}(v)}\left( v \neq \error \right)~\big)
\end{array}
\\
\wedge
\\
\diam{d\,u_2}\diam{\co{c}(w)}\big(
\begin{array}[t]{l}
  w \neq \error~\wedge \\
  \boxm{d\,z}\boxm{\co{c}(v)}\left( v \neq \error \right)~\big)~\Big)
\end{array}
\end{array}
\end{array}
\]
Since
$\SystemMin^{\error} \vDash \psi_{\ref{thm:system-min-error}}$ but $\SpecMin^{\error} \nvDash \psi_{\ref{thm:system-min-error}}$,
by \autoref{thm:i-HM},
$\SystemMin^{\error} \nbisimi{i} \SpecMin^{\error}$.
\end{proof}
We explain some features of the formula $\psi_{\ref{thm:system-min-error}}$ above.
The tests $u_1 \neq \error$ and $u_2 \neq \error$ are sufficient
to ensure that neither of the outputs $u_1$ and $u_2$ are coming from the ePassport
 that previously sent $nt_1$,
and hence both $d\,n_1$ and $d\,n_1$ must have been inputs starting two new reader sessions.
This does not, by itself, ensure that $u_1$ and $u_2$ were both set by the 
two readers started by the previous inputs,
since either may instead indicate the initiation of a new ePassports.
However, the first branch of the conjunction, specifically the subformula
\[
\diam{d\,z}\diam{\co{c}(e)}\big(
   e = \error~\wedge 
   \boxm{d\,z}\boxm{\co{c}(v)}\left( v \neq \error \right)~\big)
\]
 rules this out
by checking that at that moment there is exactly one ePassport currently active, \ie this one that sent $nt_1$.
The trick to achieve this is to show that that there is an ePassport that can be induced to send an error 
(as represented by $\diam{d\,z}\diam{\co{c}(e)}$ and $e = \error$)
and, furthermore, having done so no other ePassport is at a stage where it can do the same
(as represented by $\boxm{d\,z}\boxm{\co{c}(v)}\left( v \neq \error \right)$).
This is an elaborate way of achieving the effect of the strategies $\phi_{\ref{thm:system-get}}$, $\phi_{\ref{thm:system-end}}$ and $\phi_{\ref{thm:system-finite}}$, 
that ensure by that point that exactly
two readers have responded to a challenge from the same ePassport.
The other two branches of the conjunction correspond to the branches in the other strategies.
Again, however, this is done in an indirect manner, 
where we ensure that the last message $w$ really was the final message of the ePassport,
indicating successful authentication, and not the start of a new session,
by checking that at that point there is not ePassport that is already started
and therefore ready to receive an input that will induce an error. 

The above distinguishing strategy
also works when analysing the BAC protocol when constant error messages are present~\cite{Arapinis2010,Hirschi2019,Horne2021},
 since the strategy does not exploit any feature specific to Feldhofer.
Hence no formulation of the BAC protocol is unlinkable 
with respect to any formulation
of unlinkability derived from the literature if $i$-bisimilarity is used as the notion of equivalence.

One might ask whether the above, rather complex formula may be simplified with respect to HP-bisimilarity.
There is the following event-labelled $\FM$ formula, where some parts of the formula are trimmed away thanks to knowing when an output is caused by a particular input.
\[
\chi_{\ref{thm:system-min-error}}
\triangleq
\begin{array}[t]{l}
\ediam{\co{c}(nt_1)}{\prefix{001}[]}
\ediam{d\,nt_1}{\prefix{000}[]}
\ediam{d\,nt_1}{\prefix{0100}[]}
\\
\ediam{\co{c}(u_1)}{\prefix{000}[]}
\ediam{\co{c}(u_2)}{\prefix{0100}[]}
\Big(
\quad
\begin{array}[t]{l}
{u_1 \neq \error}
\wedge
{u_2 \neq \error}~\wedge
\\
\wedge
\\
\ediam{d\,u_1}{\prefix{001}[]}
\ediam{\co{c}(w)}{\prefix{001}[]}
\big(
\begin{array}[t]{l}
  w \neq \error~\wedge \\
  \eboxm{d\,z}{\prefix{1000}[]}\eboxm{\co{c}(v)}{\prefix{1000}[]}\left( v \neq \error \right)~\big)
\end{array}
\\
\wedge
\\
\ediam{d\,u_2}{\prefix{001}[]}
\ediam{\co{c}(w)}{\prefix{001}[]}
\big(
\begin{array}[t]{l}
  w \neq \error~\wedge \\
  \eboxm{d\,z}{\prefix{1000}[]}\eboxm{\co{c}(v)}{\prefix{1000}[]}\left( v \neq \error \right)~\big)~\Big)
\end{array}
\end{array}
\end{array}
\]

In contrast to the above observations,
a proof strategy similar to that which was used to establish Theorem~\ref{thm:sim-proof}
can be used to prove that HP-similarity cannot discover attacks on this minimal
formulation of unlinkability applied to BAC and Feldhofer with a cleartext error message.

\section{Discussion}
\label{sec:discussion}
In this section, we present the insights gained from our research on
how the choice of semantics influences the ability to detect privacy
and security attacks. We begin with a summary of the main findings and
then turn to future work making a comparison of the verification strength of ST semantics
versus located semantics, as well as a discussion of tooling aspects.

\subsection{Main findings}
The mantra of this work is: the language stays the same;
the semantics work harder for the programmer, thereby making privacy problems more robust against small stylistic differences.
By keeping the language the same, the same usual applied $\pi$-calculus syntax, used in many protocol analysis tools,
can be employed.
The modal logics introduced for describing attacks differ little from those which have already been established 
for privacy analysis in the literature.
We have, however, proposed an
enriched semantics for the applied $\pi$-calculus to ensure that histories are preserved
(\autoref{def:hp-sim}, \autoref{fig:modal-2}, \autoref{thm:HP-HM}).

The existence of security and privacy problems sensitive to branching-time
has been hypothesised for some time~\cite{Ryan2001,doi:10.3233/JCS-1994/1995-3103}.
Only recently has unlinkability been shown to be sensitive to branching-time semantics~\cite{Filimonov2019},
and the current paper is the first to prove that there exist
unlinkability problems 
for which even similarity does not suffice (\autoref{thm:sim-proof}),
and the distinguishing power of bisimilarity is required.

The surprising insight regarding unlinkability this work draws attention to is how
seemingly irrelevant variations in models of unlinkability impact the effectiveness of a model of unlinkability in terms of discovering attacks.
We demarcate boundaries between problems as follows.
\begin{description}
\item[{\color{DarkGreen} Feldhofer / BAC (get / ch / $2$)}] 
Using $i$-simulation is enough to discover attacks
on formulations of unlinkability for
styles with implicit location information, \eg for both the BAC and Feldhofer protocols,
that: (a) start with a constant cleartext (\autoref{thm:system-get}),
(b) create fresh endpoint channels for each session (\autoref{thm:system-end})
or (c) involve a bounded number of sessions (\autoref{thm:system-finite}).

\item[{\color{blue} BAC (min)}]
We consider formulations of unlinkability of the BAC protocol with an unbounded number of sessions,
with neither fresh endpoint channels nor a constant message signalling the beginning of each session.
For that formulation,
HP-similarity cannot discover an attack (\autoref{thm:sim-proof}),
yet $i$-bisimilarity can discover an attack (\autoref{thm:system-min-silent}).
This observation holds whether a prover (ePassport) fails authentication silently or with an error message (\autoref{thm:system-min-error}).

\item[{\color{red} Feldhofer (min)}]
We consider the formulation of the Feldhofer protocol in the same style as above,
and specifically where the prover fails authentication silently.
For that formulation 
again, HP-similarity cannot discover an attack (\autoref{thm:sim-proof}),
but, in addition, neither can $i$-bisimilarity (\autoref{thm:bisim-proof}).
Of the equivalences considered, only HP-bisimilarity can discover attacks (\autoref{thm:min-HP}).
\end{description} 
The three problems described in the three items above correspond to the three dotted lines in the diagram from \autoref{fig:schema_summary}.
Each line graphically separates semantics 
that can be used to find attacks, above the line, from those
semantics that cannot find attacks, below the line: the last line, corresponding to the trace semantics, cannot find any attack.

\begin{figure}[t]
	{
		\centering
		\includegraphics[width=\linewidth]{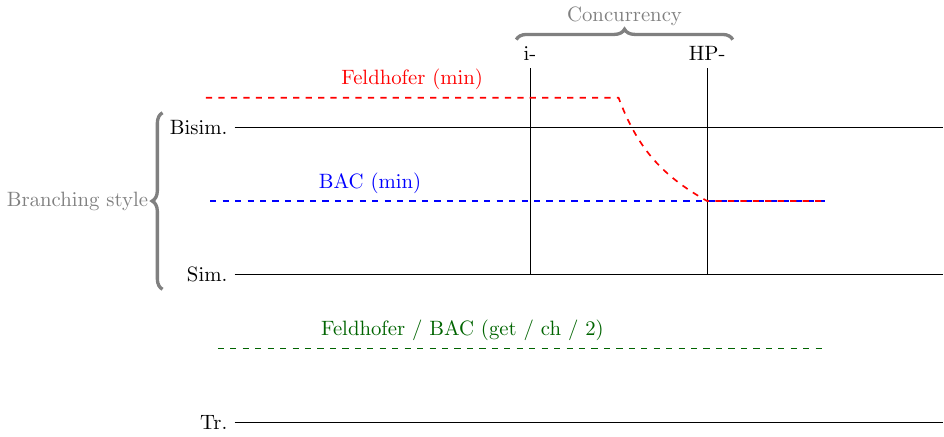}

	}
	\caption{The three main unlinkability problems vs.\ semantics useful to find attacks (if above the dotted lines).}
	\label{fig:schema_summary}
\end{figure}

Notice that, of the semantics explored,
only HP-bisimilarity, in the top right of the diagram,
finds attacks for all formulations of unlinkability and for all protocols considered.
Also notice, above, that selecting 
similarity misses attacks for all protocol styles from \autoref{table:stylistic-diff} (the ones with get message, with endpoints, or bounded)
unless a model of unlinkability with implicit location information is provided.
Both of these observations are new, and were not obvious to the authors
before conducting this investigation.

\subsection{ST semantics miss attacks; located semantics find attacks}

Between interleaving and history-preserving semantics, there is a spectrum of 
non-interleaving semantics.
It is natural to ask whether a less discriminating non-interleaving semantics than HP-bisimilarity can also
discover attacks on the Feldhofer protocol.
To this end, we developed, in a companion paper~\cite{Aubert2022k}, ST-bisimilarity (start-terminate bisimilarity)
for the applied $\pi$-calculus, which preserves the timing of events.
We were surprised to discover that ST-bisimilarity
 was not able to prove an analogous result to \autoref{thm:Feldhofer-attack}.
Indeed, using ST-bisimilarity, we can strengthen \autoref{thm:bisim-proof},
such that $\SystemMin^{F} \bisimi{ST} \SpecMin^{F}$,
thereby ruling out ST-bisimilarity (and hence timing observations alone)
as a method for detecting attacks.

While ST semantics are insufficient for finding all attacks,
another branch of non-interleaving semantics, where locations are directly observable on labels is however effective for discovering attacks~\cite{Castellani2001,Sangiorgi1996}.
Located semantics are related to \enquote{session} equivalences under development in the DeepSec tool~\cite{Cheval2019}.
There is related work on \enquote{causal} 
bisimilarity for the $\pi$-calculus~\cite{Boreale1998},
which is strictly 
finer than HP-bisimilarity.
This is because causal bisimilarity observes the difference between structural and link causality.
For example,
although
$
\id \cpar (\mathopen{\nu n.}\left(
{\cout{a}{n}}
\cpar {\cin{n}{x}}
\right))
\bisimi{HP}
\id \cpar (
\mathopen{\nu n.}\left(
\cout{a}{n}.
{\cin{n}{x}}
\right))
$
holds, these processes are distinguished by causal bisimilarity, because \textcquote[p.~387]{Boreale1998}{there is both a subject and an object dependency between the actions \textins{in the former}, whereas in \textins{the latter} there is only an object dependency}.
Note, a suitable located or causal similarity can discover attacks on all formulations of unlinkability considered.
That is, unlike HP-bisimilarity, the full power of bisimilarity would not be required --
an observation that would be interesting to elaborate on in future work.

A formal analysis of the claim that pomset traces cannot discover attacks on any of formulations of unlinkability in this paper is also postponed to future work.
While it is relatively clear what the definition of pomset traces 
should be, it would need to be lifted to the applied $\pi$-calculus.
This is related to old work on \enquote{execution DAGs} (or Mazurkiewicz traces) recording all input-output dependencies~\cite{Clarke2000,Cremers2004,Fokkink2010,Basin2010}.

\subsection{Applicability of HP-bisimilarity to Tamarin and other tooling problems}

Although we fix the applied $\pi$-calculus for this study,
we argue that HP-bisimilarity would also be a natural choice of equivalence 
for protocol verification tools such as Tamarin~\cite{Meier2013}.
Currently, Tamarin does not handle any notion of bisimilarity, implementing a notion of equivalence checking 
where only differences in messages are accounted for.
Our argument that HP-bisimilarity could provide a suitable notion of behavioural equivalence for Tamarin is based on the following observations,
where we map concepts in this work to Tamarin.

Firstly, we explain how structural dependencies arise in Tamarin.
State transitions in Tamarin are achieved by rewriting a set of multisets. Each element of a multiset forms a natural notion of location that can be used in an adaption of \autoref{def:location-and-events}.
Partially ordered location names~\cite{Castellani1995,Sangiorgi1996} can be assigned to elements of multisets,
where the partial order captures structural dependencies.
Two locations are structurally dependent when one appears on the left and another appears on the right of 
an instance of a multiset rewrite.
Indeed, such structural dependencies are already explicit in Tamarin, as \enquote{origin} arrows.

Secondly, we explain that link dependencies can be defined between inputs and outputs similarly to \autoref{def:independence}.
An input is dependent on an output if the input message is constructed using exactly that output instance.
Indeed, input messages for this purpose may be represented similarly to the applied $\pi$-calculus in \autoref{figure:active},
by providing a message representing a recipe for constructing the input from outputs and public information.
Again, this information is available already in Tamarin, since it can be extracted via a correspondence between messages and constructs such as $KU$ and $KD$ that manipulate attacker knowledge.

The definition of HP-bisimilarity itself would then require only minor adjustments
for it to be reused for the equivalence checking of models defined using Tamarin.
Based on such a definition,
Tamarin could be extended in future work
to allow for two protocol models defined in terms of multi-set rewrite rules and restrictions
to be compared using HP-bisimilarity.

Observe, furthermore, that Tamarin does not feature channels as primitive.
This means that Tamarin lends itself naturally to the \enquote{minimal} style and interpretation in \autoref{table:stylistic-diff}
where channels play no active role.
This \enquote{minimal} style is exactly where HP-bisimilarity has impact,
as demonstrated by our leading examples Theorems~\ref{thm:system-min-interleaving} and \ref{thm:min-HP}. 
Indeed, the attack strategies on the Feldhofer protocol in \autoref{thm:min-HP}
would be valid also in a faithful translation of the unlinkability problem to two Tamarin models
representing each process.

Regarding tooling for the applied $\pi$-calculus, 
$i$-bisimilarity checking has recently been added to an experimental branch of ProVerif~\cite{Cheval23}
and DeepSec~\cite{Cheval24}.
Future work could explore whether those implementations can be adapted to HP-bisimilarity.
An implementation would enable experimental evaluation of the impact of HP-bisimilarity on verification in comparison to $i$-bisimilarity.
While bisimilarity is undecidable in general~\cite{DBLP:journals/tcs/Jancar95,DBLP:conf/ac/Esparza96}, a complexity analysis is possible for 
bounded processes with a particular class of message theory. 
For the message theory running throughout this paper, related work found $i$-bisimilarity to be coNEXP-complete~\cite{Cheval24}.
Future work could leverage results on the high complexity of deciding HP-bisimilarity for bounded Petri nets~\cite{DBLP:conf/stacs/Vogler91,DBLP:conf/stacs/MontanariP97,DBLP:journals/lmcs/CescoG23} for determining the complexity class of HP-bisimilarity for the applied $\pi$-calculus.

The development of tools would assist with larger unlinkability problems.
Further to those mentioned already, examples of unlinkability problems include the 
client unlinkability of TLS 1.3 with Encrypted Client Hello~\cite{Bhargavan22}.
In that problem, multiple connections of a client using a valid certificate for an honest server 
cannot be distinguished from connections by different clients also with such valid certificates.
Related work investigates the unlinkability of payment protocols~\cite{Bursuc2023}, making use of what we call the \enquote{endpoint} style in Table~\ref{table:stylistic-diff}.
Better tool support would enable us to investigate further the impact of stylistic and semantic choices on such rich examples.
Many privacy properties of protocols are amenable to the methods in this paper,
not limited to the anonymity of Noise~\cite{Girol2020}, peer deniability of Signal~\cite{Signal2017}, and issuer unlinkability of W3C Verifiable Credentials~\cite{Braun2024}.

\section{Conclusion}
\label{sec:conclusion}

This study has covered key stylistic differences between models of what should be the same ePassport unlinkability problem.
The study finds that different formulations of unlinkability are sensitive to different semantics.
In particular, for attacks on unlinkability problems discovered using branching-time equivalences,
some utilise the power of switching players in the bisimulation game.
Furthermore, some attacks also required the power of non-interleaving semantics as we have investigated by developing and applying a theory of HP-bisimilarity for the applied $\pi$-calculus.
This shows that the power of the attacker in a threat model for privacy problems is connected in a subtle way to the semantics of processes used to model a privacy problem in a symbolic Dolev-Yao style.
The class of models that are most impacted are in fact a natural class of models where nonessential information about channels and non-cryptographic material is stripped away, suggesting that insight gained connecting non-interleaving semantics with privacy problems is reusable across other models of cryptographic protocols that do not feature channels as primitive.

\subsection*{Acknowledgements}
We are grateful to the anonymous reviewers for their detailed reading of the semantics and suggestions connecting this work with classic work on noninterference.
This research did not receive any specific grant from funding agencies in the public, commercial, or not-for-profit sectors, except C.\ Aubert, from the \href{https://www.nsf.gov}{National Science Foundation} under Grant No.: \href{https://www.nsf.gov/awardsearch/showAward?AWD_ID=2242786}{2242786}.

\bibliographystyle{elsarticle-num.bst}
\bibliography{bib}

@inproceedings{Cheval23,
	title        = {Indistinguishability Beyond Diff-Equivalence in ProVerif},
	author       = {Vincent Cheval and Itsaka Rakotonirina},
	year         = 2023,
	booktitle    = {36th {IEEE} Computer Security Foundations Symposium, {CSF} 2023, Dubrovnik, Croatia, July 10-14, 2023},
	publisher    = {{IEEE}},
	pages        = {184--199},
	doi          = {10.1109/CSF57540.2023.00036},
	bibsource    = {dblp computer science bibliography, https://dblp.org}
}

@article{Cheval24,
	title        = {DeepSec: Deciding Equivalence Properties for Security Protocols -- Improved theory and practice},
	author       = {Vincent Cheval and Steve Kremer and Itsaka Rakotonirina},
	year         = 2024,
	month        = {Mar},
	journal      = {TheoretiCS},
	volume       = {Volume 3},
	doi          = {10.46298/theoretics.24.4},
	issn         = {2751-4838},
	eid          = 4
}

@inproceedings{Signal2017,
	title        = {A Formal Security Analysis of the Signal Messaging Protocol},
	author       = {Cohn-Gordon, Katriel and Cremers, Cas and Dowling, Benjamin and Garratt, Luke and Stebila, Douglas},
	year         = 2017,
	booktitle    = {2017 IEEE European Symposium on Security and Privacy (EuroS\&P)},
	volume       = {},
	number       = {},
	pages        = {451--466},
	doi          = {10.1109/EuroSP.2017.27}
}

@inproceedings{Braun2024,
	title        = {{SSI}, from Specifications to Protocol? {F}ormally Verify Security!},
	author       = {Christoph H.{-}J. Braun and Ross Horne and Tobias K{\"{a}}fer and Sjouke Mauw},
	year         = 2024,
	booktitle    = {Proceedings of the {ACM} on Web Conference 2024, {WWW} 2024, Singapore, May 13-17, 2024},
	publisher    = {{ACM}},
	pages        = {1620--1631},
	doi          = {10.1145/3589334.3645426},
	editor       = {Tat{-}Seng Chua and Chong{-}Wah Ngo and Ravi Kumar and Hady W. Lauw and Roy Ka{-}Wei Lee},
	bibsource    = {dblp computer science bibliography, https://dblp.org}
}

@inproceedings{Girol2020,
	title        = {A spectral analysis of {Noise}: A comprehensive, automated, formal analysis of {Diffie-Hellman} protocols},
	author       = {Girol, Guillaume and Hirschi, Lucca and Sasse, Ralf and Jackson, Dennis and Cremers, Cas and Basin, David},
	year         = 2020,
	booktitle    = {29th USENIX Security Symposium (USENIX Security 20)},
	pages        = {1857--1874}
}

@inproceedings{Bursuc2023,
	title        = {Provably Unlinkable Smart Card-based Payments},
	author       = {Bursuc, Sergiu and Horne, Ross and Mauw, Sjouke and Yurkov, Semen},
	year         = 2023,
	booktitle    = {Proceedings of the 2023 ACM SIGSAC Conference on Computer and Communications Security},
	location     = {Copenhagen, Denmark},
	publisher    = {Association for Computing Machinery},
	address      = {New York, NY, USA},
	series       = {CCS '23},
	pages        = {1392–1406},
	doi          = {10.1145/3576915.3623109},
	isbn         = 9798400700507,
	numpages     = 15
}

@inproceedings{Bhargavan22,
	title        = {A Symbolic Analysis of Privacy for TLS 1.3 with Encrypted Client Hello},
	author       = {Bhargavan, Karthikeyan and Cheval, Vincent and Wood, Christopher},
	year         = 2022,
	booktitle    = {Proceedings of the 2022 ACM SIGSAC Conference on Computer and Communications Security},
	location     = {Los Angeles, CA, USA},
	publisher    = {Association for Computing Machinery},
	address      = {New York, NY, USA},
	series       = {CCS '22},
	pages        = {365–379},
	doi          = {10.1145/3548606.3559360},
	isbn         = 9781450394505,
	numpages     = 15
}

@incollection{Castellani2001,
	title        = {Process Algebras with Localities},
	author       = {Ilaria Castellani},
	year         = 2001,
	booktitle    = {Handbook of Process Algebra},
	publisher    = {Elsevier Science},
	address      = {Amsterdam},
	pages        = {945--1045},
	doi          = {10.1016/B978-044482830-9/50033-3},
	isbn         = {978-0-444-82830-9},
	editor       = {J.A. Bergstra and A. Ponse and S.A. Smolka}
}

@inproceedings{Feldhofer2004,
	title        = {Strong Authentication for {RFID} Systems Using the {AES} Algorithm},
	author       = {Martin Feldhofer and Sandra Dominikus and Johannes Wolkerstorfer},
	year         = 2004,
	booktitle    = {Cryptographic Hardware and Embedded Systems - CHES 2004},
	publisher    = {Springer},
	pages        = {357--370},
	doi          = {10.1007/978-3-540-28632-5_26},
	editor       = {Marc Joye and Jean-Jacques Quisquater}
}

@misc{CC,
	title        = {Common Criteria for Information Technology Security Evaluation, Part 2: Security functional components. Version 3.1, Revision 5},
	year         = 2017,
	number       = {CCMB-2017-04-002},
	url          = {https://www.commoncriteriaportal.org/files/ccfiles/CCPART2V3.1R5.pdf},
	urldate      = {2023-03-25}
}

@inproceedings{Joyal1993,
	title        = {Bisimulation and open maps},
	author       = {Andr{\'{e}} Joyal and Mogens Nielsen and Glynn Winskel},
	year         = 1993,
	booktitle    = {Proceedings of the Eighth Annual Symposium on Logic in Computer Science {(LICS} '93), Montreal, Canada, June 19-23, 1993},
	publisher    = {{IEEE} Computer Society},
	pages        = {418--427},
	doi          = {10.1109/LICS.1993.287566},
	bibsource    = {dblp computer science bibliography, https://dblp.org}
}

@inproceedings{Hirschi2016,
	title        = {A method for verifying privacy-type properties: the unbounded case},
	author       = {Hirschi, Lucca and Baelde, David and Delaune, St{\'e}phanie},
	year         = 2016,
	booktitle    = {Security and Privacy (S\&P), 2016 IEEE Symposium on},
	pages        = {564--581},
	doi          = {10.1109/SP.2016.40},
	organization = {IEEE}
}

@article{Hirschi2019,
	title        = {A method for unbounded verification of privacy-type properties},
	author       = {Lucca Hirschi and David Baelde and St{\'e}phanie Delaune},
	year         = 2019,
	journal      = {Journal of Computer Security},
	publisher    = {IOS Press},
	volume       = 27,
	number       = 3,
	pages        = {277--342},
	doi          = {10.3233/JCS-171070}
}

@inproceedings{Filimonov2019,
	title        = {Breaking Unlinkability of the {ICAO} 9303 Standard for e-{P}assports Using Bisimilarity},
	author       = {Ihor Filimonov and Ross Horne and Sjouke Mauw and Zach Smith},
	year         = 2019,
	booktitle    = {Computer Security - {ESORICS} 2019 - 24th European Symposium on Research in Computer Security, Luxembourg, September 23-27, 2019, Proceedings, Part {I}},
	publisher    = {Springer},
	series       = {LNCS},
	volume       = 11735,
	pages        = {577--594},
	doi          = {10.1007/978-3-030-29959-0_28},
	isbn         = {978-3-030-29958-3},
	editor       = {Kazue Sako and Steve A. Schneider and Peter Y. A. Ryan}
}

@techreport{ICAO,
	title        = {Machine Readable Travel Documents. Part 11: Security mechanisms for {MRTDs}},
	year         = 2015,
	number       = {Doc 9303. Seventh Edition},
	organization = {International Civil Aviation Organization (ICAO)},
	institution  = {International Civil Aviation Organization (ICAO)}
}

@article{Castellani1995,
	title        = {Observing distribution in processes: static and dynamic localities},
	author       = {Castellani, Ilaria},
	year         = 1995,
	journal      = {Int. J. Found. Comput. Sci.},
	publisher    = {World Scientific},
	volume       = 6,
	number       = {04},
	pages        = {353--393},
	doi          = {10.1142/S0129054195000196}
}

@article{Sangiorgi1996,
	title        = {Locality and interleaving semantics in calculi for mobile processes},
	author       = {Davide Sangiorgi},
	year         = 1996,
	journal      = {Theor.\ Comput.\ Sci.},
	volume       = 155,
	number       = 1,
	pages        = {39--83},
	doi          = {10.1016/0304-3975(95)00020-8},
	issn         = {0304-3975}
}

@article{Glabbeek1997,
	title        = {The Difference between Splitting in $n$ and $n+1$},
	author       = {van Glabbeek, Rob J. and Frits W. Vaandrager},
	year         = 1997,
	journal      = {Inf. Comput.},
	volume       = 136,
	number       = 2,
	pages        = {109--142},
	doi          = {10.1006/inco.1997.2634}
}

@article{Vogler1996,
	title        = {The Limit of Split$_n$-Language Equivalence},
	author       = {Walter Vogler},
	year         = 1996,
	journal      = {Inf. Comput.},
	volume       = 127,
	number       = 1,
	pages        = {41--61},
	doi          = {10.1006/inco.1996.0048},
	issn         = {0890-5401}
}

@article{Basin2010,
	title        = {Constraint differentiation: Search-space reduction for the constraint-based analysis of security protocols},
	author       = {M{\"o}dersheim, Sebastian and Vigan{\`o}, Luca and Basin, David},
	year         = 2010,
	journal      = {J. Comput. Secur.},
	publisher    = {IOS Press},
	volume       = 18,
	number       = 4,
	pages        = {575--618},
	doi          = {10.3233/JCS-2009-0351}
}

@inproceedings{Aubert2022k,
	title        = {Bisimulations Respecting Duration and Causality for the Non-interleaving Applied $\pi$-Calculus},
	author       = {Aubert, Clément and Horne, Ross and Johansen, Christian},
	year         = 2022,
	month        = sep,
	booktitle    = {Proceedings Combined 29th International Workshop on Expressiveness in Concurrency and 19th Workshop on Structural Operational Semantics, Warsaw, Poland, 12th September 2022},
	publisher    = {Open Publishing Association},
	series       = {Electronic Proceedings in Theoretical Computer Science},
	volume       = 368,
	pages        = {3--22},
	doi          = {10.4204/EPTCS.368.1},
	editor       = {Castiglioni, Valentina and Mezzina, Claudio A.}
}

@article{Ryan2001,
	title        = {Process algebra and non-interference},
	author       = {Ryan, Peter Y. A. and Schneider, Steve A.},
	year         = 2001,
	journal      = {Journal of Computer Security},
	volume       = 9,
	number       = {1-2},
	pages        = {75--103},
	doi          = {10.3233/JCS-2001-91-204}
}

@inproceedings{Aubert2022e,
	title        = {Diamonds for Security: A Non-Interleaving Operational Semantics for the Applied Pi-Calculus},
	author       = {Aubert, Cl{\'e}ment and Horne, Ross and Johansen, Christian},
	year         = 2022,
	booktitle    = {33rd International Conference on Concurrency Theory},
	publisher    = {Schloss Dagstuhl--Leibniz-Zentrum f{\"u}r Informatik},
	series       = {Leibniz International Proceedings in Informatics},
	volume       = 243,
	pages        = {30:1--30:26},
	doi          = {10.4230/LIPIcs.CONCUR.2022.30},
	editor       = {Bartek Klin and S\l{}awomir Lasota and Anca Muscholl}
}

@inproceedings{Aubert2020b,
	title        = {How Reversibility Can Solve Traditional Questions: The Example of Hereditary History-Preserving Bisimulation},
	author       = {Aubert, Clément and Cristescu, Ioana},
	year         = 2020,
	booktitle    = {31st International Conference on Concurrency Theory, {CONCUR} 2020, September 1--4, 2020, Vienna, Austria},
	publisher    = {Schloss Dagstuhl},
	series       = {LIPIcs},
	volume       = 2017,
	pages        = {13:1--13:24},
	doi          = {10.4230/LIPIcs.CONCUR.2020.13},
	editor       = {Konnov, Igor and Kov\'{a}cs, Laura}
}

@article{Abadi2018,
	title        = {The Applied Pi Calculus: Mobile Values, New Names, and Secure Communication},
	author       = {Abadi, Mart{\'\i}n and Blanchet, Bruno and Fournet, C{\'e}dric},
	year         = 2018,
	journal      = {J.\ ACM},
	volume       = 65,
	number       = 1,
	pages        = {1:1--1:41},
	doi          = {10.1145/3127586}
}

@inproceedings{hildebrandt1996comparing,
	title        = {Comparing Transition Systems with Independence and Asynchronous Transition Systems},
	author       = {Thomas T. Hildebrandt and Vladimiro Sassone},
	year         = 1996,
	booktitle    = {{CONCUR} '96, Concurrency Theory, 7th International Conference, Pisa, Italy, August 26-29, 1996, Proceedings},
	publisher    = {Springer},
	series       = {LNCS},
	volume       = 1119,
	pages        = {84--97},
	doi          = {10.1007/3-540-61604-7_49},
	editor       = {Ugo Montanari and Vladimiro Sassone}
}

@article{Johansen2016,
	title        = {{ST}-structures},
	author       = {Christian Johansen},
	year         = 2016,
	journal      = {J. Log. Algebraic Methods Program.},
	volume       = 85,
	number       = 6,
	pages        = {1201--1233},
	doi          = {10.1016/j.jlamp.2015.10.009}
}

@inproceedings{DBLP:conf/fosad/FocardiG00,
	title        = {Classification of Security Properties (Part {I:} Information Flow)},
	author       = {Riccardo Focardi and Roberto Gorrieri},
	year         = 2000,
	booktitle    = {Foundations of Security Analysis and Design, Tutorial Lectures [revised versions of lectures given during the {IFIP} {WG} 1.7 International School on Foundations of Security Analysis and Design, {FOSAD} 2000, Bertinoro, Italy, September 2000]},
	publisher    = {Springer},
	series       = {LNCS},
	volume       = 2171,
	pages        = {331--396},
	doi          = {10.1007/3-540-45608-2_6},
	editor       = {Riccardo Focardi and Roberto Gorrieri},
	bibsource    = {dblp computer science bibliography, https://dblp.org}
}

@article{doi:10.3233/JCS-1994/1995-3103,
	title        = {A Classification of Security Properties for Process Algebras1},
	author       = {Riccardo Focardi and Roberto Gorrieri},
	year         = 1995,
	journal      = {Journal of Computer Security},
	volume       = 3,
	number       = 1,
	pages        = {5--33},
	doi          = {10.3233/JCS-1994/1995-3103}
}

@inproceedings{DBLP:conf/fosad/FocardiGM02,
	title        = {Classification of Security Properties - Part {II:} Network Security},
	author       = {Riccardo Focardi and Roberto Gorrieri and Fabio Martinelli},
	year         = 2002,
	booktitle    = {Foundations of Security Analysis and Design II, {FOSAD} 2001/2002 Tutorial Lectures},
	publisher    = {Springer},
	series       = {LNCS},
	volume       = 2946,
	pages        = {139--185},
	doi          = {10.1007/978-3-540-24631-2_4},
	editor       = {Riccardo Focardi and Roberto Gorrieri}
}

@article{DBLP:journals/tcs/Jancar95,
	title        = {Undecidability of Bisimilarity for Petri Nets and Some Related Problems},
	author       = {Petr Jancar},
	year         = 1995,
	journal      = {Theor. Comput. Sci.},
	volume       = 148,
	number       = 2,
	pages        = {281--301},
	doi          = {10.1016/0304-3975(95)00037-W}
}

@inproceedings{DBLP:conf/ac/Esparza96,
	title        = {Decidability and Complexity of Petri Net Problems - An Introduction},
	author       = {Javier Esparza},
	year         = 1996,
	booktitle    = {Lectures on Petri Nets {I:} Basic Models, Advances in Petri Nets, the volumes are based on the Advanced Course on Petri Nets, held in Dagstuhl, September 1996},
	publisher    = {Springer},
	series       = {LNCS},
	volume       = 1491,
	pages        = {374--428},
	doi          = {10.1007/3-540-65306-6_20},
	editor       = {Wolfgang Reisig and Grzegorz Rozenberg}
}

@inproceedings{DBLP:conf/stacs/Vogler91,
	title        = {Bisimulation and Action Refinement},
	author       = {Walter Vogler},
	year         = 1991,
	booktitle    = {{STACS} 91, 8th Annual Symposium on Theoretical Aspects of Computer Science, Hamburg, Germany, February 14-16, 1991, Proceedings},
	publisher    = {Springer},
	series       = {LNCS},
	volume       = 480,
	pages        = {309--321},
	doi          = {10.1007/BFB0020808},
	editor       = {Christian Choffrut and Matthias Jantzen}
}

@inproceedings{DBLP:conf/stacs/MontanariP97,
	title        = {Minimal Transition Systems for History-Preserving Bisimulation},
	author       = {Ugo Montanari and Marco Pistore},
	year         = 1997,
	booktitle    = {{STACS} 97, 14th Annual Symposium on Theoretical Aspects of Computer Science, L{\"{u}}beck, Germany, February 27 - March 1, 1997, Proceedings},
	publisher    = {Springer},
	series       = {LNCS},
	volume       = 1200,
	pages        = {413--425},
	doi          = {10.1007/BFB0023477},
	editor       = {R{\"{u}}diger Reischuk and Michel Morvan}
}

@article{DBLP:journals/lmcs/CescoG23,
	title        = {Decidability of Two Truly Concurrent Equivalences for Finite Bounded Petri Nets},
	author       = {Arnaldo Cesco and Roberto Gorrieri},
	year         = 2023,
	journal      = {Log. Methods Comput. Sci.},
	volume       = 19,
	number       = 4,
	doi          = {10.46298/LMCS-19(4:37)2023}
}

@inproceedings{Glabeek1989,
	title        = {Equivalence Notions for Concurrent Systems and Refinement of Actions (Extended Abstract)},
	author       = {van Glabbeek, Rob J. and Goltz, Ursula},
	year         = 1989,
	booktitle    = {MFCS},
	publisher    = {Springer},
	series       = {LNCS},
	volume       = 379,
	pages        = {237--248},
	doi          = {10.1007/3-540-51486-4_71},
	isbn         = {3-540-51486-4},
	editor       = {Kreczmar, Antoni and Mirkowska, Grazyna}
}

@article{Horne2021,
	title        = {Discovering e{P}assport Vulnerabilities using Bisimilarity},
	author       = {Ross Horne and Sjouke Mauw},
	year         = 2021,
	journal      = {Log.\ Meth.\ Comput.\ Sci.},
	volume       = 17,
	number       = 2,
	pages        = 24,
	doi          = {10.23638/LMCS-17(2:24)2021}
}

@article{Blanchet2016,
	title        = {Modeling and Verifying Security Protocols with the Applied Pi Calculus and ProVerif},
	author       = {Bruno Blanchet},
	year         = 2016,
	journal      = {Foundations and Trends in Privacy and Security},
	volume       = 1,
	number       = {1-2},
	pages        = {1--135},
	doi          = {10.1561/3300000004},
	issn         = {2474-1558}
}

@article{Phillips2014,
	title        = {Event Identifier Logic},
	author       = {Phillips, Iain and Ulidowski, Irek},
	year         = 2014,
	journal      = {Math. Struct. Comput. Sci.},
	volume       = 24,
	number       = 2,
	doi          = {10.1017/S0960129513000510}
}

@article{Glabbeek2001,
	title        = {Refinement of actions and equivalence notions for concurrent systems},
	author       = {van Glabbeek, Rob J. and Goltz, Ursula},
	year         = 2001,
	journal      = {Acta Inform.},
	volume       = 37,
	number       = {4/5},
	pages        = {229--327},
	doi          = {10.1007/s002360000041}
}

@inproceedings{Glabbeek1993,
	title        = {The Linear Time - Branching Time Spectrum {II}},
	author       = {van Glabbeek, Rob J.},
	year         = 1993,
	booktitle    = {{CONCUR} '93},
	publisher    = {Springer},
	series       = {LNCS},
	volume       = 715,
	pages        = {66--81},
	doi          = {10.1007/3-540-57208-2_6},
	isbn         = {3-540-57208-2},
	editor       = {Best, Eike}
}

@inproceedings{Mukund1992,
	title        = {{CCS, Location and Asynchronous Transition Systems}},
	author       = {Madhavan Mukund and Nielsen, Mogens},
	year         = 1992,
	booktitle    = {Foundations of Software Technology and Theoretical Computer Science, 12th Conference, New Delhi, India, December 18-20, 1992, Proceedings},
	publisher    = {Springer},
	series       = {LNCS},
	volume       = 652,
	pages        = {328--341},
	doi          = {10.1007/3-540-56287-7_116},
	isbn         = {3-540-56287-7},
	editor       = {R. K. Shyamasundar}
}

@article{Boreale1998,
	title        = {A fully abstract semantics for causality in the $\pi$-calculus},
	author       = {Boreale, Michele and Sangiorgi, Davide},
	year         = 1998,
	journal      = {Acta Inform.},
	publisher    = {Springer},
	volume       = 35,
	number       = 5,
	pages        = {353--400},
	doi          = {10.1007/s002360050124}
}

@inproceedings{Cheval2019,
	title        = {Exploiting Symmetries When Proving Equivalence Properties for Security Protocols},
	author       = {Vincent Cheval and Steve Kremer and Itsaka Rakotonirina},
	year         = 2019,
	booktitle    = {Proceedings of the 2019 {ACM} {SIGSAC} Conference on Computer and Communications Security, {CCS} 2019, London, UK, November 11-15, 2019},
	publisher    = {{ACM}},
	pages        = {905--922},
	doi          = {10.1145/3319535.3354260},
	isbn         = {978-1-4503-6747-9},
	editor       = {Lorenzo Cavallaro and Johannes Kinder and XiaoFeng Wang and Jonathan Katz}
}

@inproceedings{Cremers2004,
	title        = {Checking Secrecy by Means of Partial Order Reduction},
	author       = {Cas J. F. Cremers and Sjouke Mauw},
	year         = 2004,
	booktitle    = {System Analysis and Modeling, 4th International {SDL} and {MSC} Workshop, {SAM} 2004, Ottawa, Canada, June 1-4, 2004, Revised Selected Papers},
	publisher    = {Springer},
	series       = {LNCS},
	volume       = 3319,
	pages        = {171--188},
	doi          = {10.1007/978-3-540-31810-1_12},
	isbn         = {3-540-24561-8},
	editor       = {Daniel Amyot and Alan W. Williams}
}

@article{Clarke2000,
	title        = {Efficient verification of security protocols using partial-order reductions},
	author       = {Edmund M. Clarke and Somesh Jha and Wilfredo R. Marrero},
	year         = 2003,
	journal      = {Int. J. Softw. Tools Technol. Transf.},
	volume       = 4,
	number       = 2,
	pages        = {173--188},
	doi          = {10.1007/s10009-002-0103-4},
	url         = {https://doi.org/10.1007/s10009-002-0103-4},
	bibsource    = {dblp computer science bibliography, https://dblp.org}
}

@inproceedings{Fokkink2010,
	title        = {Partial order reduction for branching security protocols},
	author       = {Fokkink, Wan and Dashti, Mohammad Torabi and Wijs, Anton},
	year         = 2010,
	booktitle    = {2010 10th International Conference on Application of Concurrency to System Design},
	pages        = {191--200},
	doi          = {10.1109/ACSD.2010.19},
	organization = {IEEE}
}

@article{Rabinovich1988,
	title        = {Behavior Structures and Nets},
	author       = {Rabinovich, Alexander and Trakhtenbrot, Boris Avraamovich},
	year         = 1988,
	journal      = {Fund.\ Inform.},
	volume       = 11,
	number       = 4,
	pages        = {357--404},
	doi          = {10.3233/FI-1988-11404}
}

@article{Joyal1996b,
	title        = {Bisimulation from Open Maps},
	author       = {Joyal, André and Nielsen, Mogens and Winskel, Glynn},
	year         = 1996,
	journal      = {Inf.\ Comput.},
	volume       = 127,
	number       = 2,
	pages        = {164--185},
	doi          = {10.1006/inco.1996.0057}
}

@article{GORRIERI1995272,
	title        = {Split and {ST} Bisimulation Semantics},
	author       = {Gorrieri, Roberto and Laneve, Cosimo},
	year         = 1995,
	journal      = {Inf. Comput.},
	volume       = 118,
	number       = 2,
	pages        = {272--288},
	doi          = {10.1006/inco.1995.1066},
	issn         = {0890-5401}
}

@techreport{Bednarczyk1991,
	title        = {Hereditary History Preserving Bisimulations or What is the Power of the Future Perfect in Program Logics},
	author       = {Bednarczyk, Marek A.},
	year         = 1991,
	institution  = {Instytut Podstaw Informatyki PAN filia w Gdańsku},
	url          = {http://www.ipipan.gda.pl/~marek/papers/historie.ps.gz}
}

@inproceedings{Nielsen1994,
	title        = {Bisimulation for Models in Concurrency},
	author       = {Nielsen, Mogens and Clausen, Christian},
	year         = 1994,
	booktitle    = {{CONCUR} '94},
	publisher    = {Springer},
	series       = {LNCS},
	volume       = 836,
	pages        = {385--400},
	doi          = {10.1007/BFb0015021},
	isbn         = {3-540-58329-7},
	editor       = {Jonsson, Bengt and Parrow, Joachim}
}

@article{Baldan2014a,
	title        = {A Logic for True Concurrency},
	author       = {Baldan, Paolo and Crafa, Silvia},
	year         = 2014,
	journal      = {J.\ ACM},
	volume       = 61,
	number       = 4,
	pages        = 24,
	doi          = {10.1145/2629638}
}

@article{Baldan2020ACMmodelChecking,
	title        = {Model Checking a Logic for True Concurrency},
	author       = {Baldan, Paolo and Padoan, Tommaso},
	year         = 2020,
	month        = {oct},
	journal      = {ACM Trans. Comput. Logic},
	publisher    = {Association for Computing Machinery},
	address      = {New York, NY, USA},
	volume       = 21,
	number       = 4,
	doi          = {10.1145/3412853},
	issn         = {1529-3785},
	issue_date   = {October 2020},
	articleno    = 34,
	numpages     = 49
}

@article{Kremer2016,
	title        = {Automated analysis of security protocols with global state},
	author       = {Steve Kremer and Robert K{\"{u}}nnemann},
	year         = 2016,
	journal      = {JCS},
	volume       = 24,
	number       = 5,
	pages        = {583--616},
	doi          = {10.3233/JCS-160556}
}

@incollection{glabbeek90concur,
	title        = {The linear time-branching time spectrum {I}. {T}he semantics of concrete, sequential processes},
	author       = {van Glabbeek, Rob J.},
	year         = 2001,
	booktitle    = {Handbook of process algebra},
	publisher    = {Elsevier},
	pages        = {3--99},
	doi          = {10.1016/b978-044482830-9/50019-9},
	editor       = {J. A. Bergstra and A. Ponse and S. A. Smolka}
}

@article{Chadha2016,
	title        = {Automated Verification of Equivalence Properties of Cryptographic Protocols},
	author       = {Chadha, Rohit and Cheval, Vincent and Ciob\^{a}c\u{a}, {\c{S}}tefan and Kremer, Steve},
	year         = 2016,
	month        = {sep},
	journal      = {ACM Trans.\ Comput.\ Log.},
	publisher    = {Association for Computing Machinery},
	address      = {New York, NY, USA},
	volume       = 17,
	number       = 4,
	pages        = {23:1--23:32},
	doi          = {10.1145/2926715},
	issn         = {1529-3785}
}

@inproceedings{Cheval2018,
	title        = {{DEEPSEC:} Deciding Equivalence Properties in Security Protocols Theory and Practice},
	author       = {Vincent Cheval and Steve Kremer and Itsaka Rakotonirina},
	year         = 2018,
	booktitle    = {2018 {IEEE} Symposium on Security and Privacy},
	publisher    = {{IEEE} Computer Society},
	pages        = {529--546},
	doi          = {10.1109/SP.2018.00033}
}

@inproceedings{Meier2013,
	title        = {The {TAMARIN} prover for the symbolic analysis of security protocols},
	author       = {Meier, Simon and Schmidt, Benedikt and Cremers, Cas and Basin, David},
	year         = 2013,
	booktitle    = {International Conference on Computer Aided Verification},
	publisher    = {Springer},
	series       = {LNCS},
	volume       = 8044,
	pages        = {696--701},
	doi          = {10.1007/978-3-642-39799-8_48},
	editor       = {Sharygina, N. and Veith, H.}
}

@inproceedings{Tiu2016,
	title        = {SPEC: An Equivalence Checker for Security Protocols},
	author       = {Tiu, Alwen and Nguyen, Nam and Horne, Ross},
	year         = 2016,
	booktitle    = {14th Asian Symposium on Programming Languages and Systems (APLAS'16)},
	publisher    = {Springer},
	series       = {LNCS},
	volume       = 10017,
	pages        = {87--95},
	doi          = {10.1007/978-3-319-47958-3_5},
	isbn         = {978-3-319-47958-3},
	editor       = {Igarashi, Atsushi}
}

@inproceedings{Cortier2017,
	title        = {{SAT-Equiv}: An Efficient Tool for Equivalence Properties},
	author       = {V{\'e}ronique Cortier and Antoine Dallon and St{\'e}phanie Delaune},
	year         = 2017,
	month        = {Aug},
	booktitle    = {2017 IEEE 30th Computer Security Foundations Symposium (CSF)},
	volume       = {},
	number       = {},
	pages        = {481--494},
	doi          = {10.1109/CSF.2017.15},
	issn         = {}
}

@article{Ross,
	title        = {A Bisimilarity Congruence for the Applied pi-Calculus Sufficiently Coarse to Verify Privacy Properties},
	author       = {Ross Horne},
	year         = 2018,
	journal      = {Arxiv},
	volume       = {arXiv:1811.02536},
	pages        = {1--31},
	url          = {https://arxiv.org/abs/1811.02536}
}

@inproceedings{Arapinis2010,
	title        = {Analysing Unlinkability and Anonymity Using the Applied Pi Calculus},
	author       = {Arapinis, Myrto and Chothia, Tom and Ritter, Eike and Ryan, Mark},
	year         = 2010,
	booktitle    = {23rd IEEE Computer Security Foundations Symposium (CSF'10)},
	pages        = {107--121},
	doi          = {10.1109/CSF.2010.15},
	issn         = {1063-6900}
}

@article{dolev83tit,
	title        = {On the security of public-key protocols},
	author       = {Danny Dolev and Andrew Yao},
	year         = 1983,
	journal      = {IEEE Transactions on Information Theory},
	volume       = 2,
	number       = 29,
	doi          = {10.1109/TIT.1983.1056650}
}

@article{aceto1994adding,
	title        = {Adding action refinement to a finite process algebra},
	author       = {Aceto, Luca and Hennessy, Matthew},
	year         = 1994,
	journal      = {Inform. and Comput.},
	publisher    = {Elsevier},
	volume       = 115,
	number       = 2,
	pages        = {179--247},
	doi          = {10.1006/inco.1994.1096}
}

@book{international1996machine,
	title        = {Machine Readable Travel Documents (Doc. 9303)},
	author       = {International Civil Aviation Organization},
	year         = 2021,
	publisher    = {International Civil Aviation Organization},
	series       = {Doc},
	url          = {https://www.icao.int/publications/doc-series/doc-9303}
}

@inproceedings{10.1145/2810103.2813662,
	title        = {Automated Symbolic Proofs of Observational Equivalence},
	author       = {Basin, David and Dreier, Jannik and Sasse, Ralf},
	year         = 2015,
	booktitle    = {Proceedings of the 22nd ACM SIGSAC Conference on Computer and Communications Security},
	location     = {Denver, Colorado, USA},
	publisher    = {Association for Computing Machinery},
	address      = {New York, NY, USA},
	series       = {CCS '15},
	pages        = {1144–1155},
	doi          = {10.1145/2810103.2813662},
	isbn         = 9781450338325,
	numpages     = 12
}

@inproceedings{10.1007/978-3-319-11851-2_11,
	title        = {A Formal Definition of Protocol Indistinguishability and Its Verification Using Maude-NPA},
	author       = {Santiago, Sonia and Escobar, Santiago and Meadows, Catherine and Meseguer, Jos{\'e}},
	year         = 2014,
	booktitle    = {Security and Trust Management},
	publisher    = {Springer International Publishing},
	address      = {Cham},
	pages        = {162--177},
	isbn         = {978-3-319-11851-2},
	editor       = {Mauw, Sjouke and Jensen, Christian Damsgaard}
}

@inproceedings{DBLP:conf/apn/Gorrieri20,
	title        = {Interleaving vs True Concurrency: Some Instructive Security Examples},
	author       = {Roberto Gorrieri},
	year         = 2020,
	booktitle    = {Application and Theory of Petri Nets and Concurrency - 41st International Conference, {PETRI} {NETS} 2020, Paris, France, June 24-25, 2020, Proceedings},
	publisher    = {Springer},
	series       = {LNCS},
	volume       = 12152,
	pages        = {131--152},
	doi          = {10.1007/978-3-030-51831-8_7},
	editor       = {Ryszard Janicki and Natalia Sidorova and Thomas Chatain}
}

@article{10.1145/2825026,
	title        = {A Survey of Security and Privacy Issues in ePassport Protocols},
	author       = {Avoine, Gildas and Beaujeant, Antonin and Hernandez-Castro, Julio and Demay, Louis and Teuwen, Philippe},
	year         = 2016,
	month        = feb,
	journal      = {ACM Comput. Surv.},
	publisher    = {Association for Computing Machinery},
	address      = {New York, NY, USA},
	volume       = 48,
	number       = 3,
	doi          = {10.1145/2825026},
	issn         = {0360-0300},
	issue_date   = {February 2016},
	articleno    = 47,
	numpages     = 37
}

\appendix

\section{Proofs}%
\label{app:proofs}

This appendix provides:
(1), proofs of the Hennessy-Milner duality between HP-bisimilarity and its characteristic modal logic;
and (2), a proof that HP-similarity cannot find attacks on our minimal formulation of unlinkability.

\subsection{Hennessy-Milner property for classical HP \texorpdfstring{$\FM$}{FM}}\label{app:HM}

We provide here proofs of the Hennessy-Milner results --
Theorems~\ref{thm:i-HM} and \ref{thm:HP-HM}.
The proof of the former is an immediate consequence of the theorems presented in this appendix.
The proof of the latter follows from the same proofs, where information regarding the relations on events is forgotten.

Since we consider the applied $\pi$-calculus, 
compared to a standard proof for a classical Milner-Parrow-Walker logic
for the $\pi$-calculus,
the use of static equivalence simplifies the analysis, since there are no special cases for bound actions.
The additional complexity in the proof
for $\HPFM$
comes from ensuring that
 conditions relating $\ESS$, $\ESS_L$, $\ESS_R$ and $\rho$ are preserved.
In particular, we should ensure throughout that the following diagram commutes:
\[
\xymatrix{
 A
 \ar^{\rho, \ESS}[d]
 & \vDash^{\ESS_L} &
 \phi
 \ar^{\rho}[d]
\\
B & \vDash^{\ESS_R} & \psi
}
\]
In what follows, 
composition of relations is as standard.
The composition of a relation with a bijection over aliases
 ${\ESS} \circ {\rho}$ is such that
$(e, e') \in \ESS$ iff $(e, e'\rho) \in {({\ESS \circ \rho})}$.

\begin{thm}[Soundness of HP-$\FM$]
If
$A \mathrel{\mathcal{R}^{\rho,\ESS}} B$
such that $\mathcal{R}$ is an HP-bisimulation,
then,
for all $\FM$ formulae $\phi$
and for all  $\ESS_L$ and $\ESS_R$ such that 
$\dom{\ESS_L} = \dom{\ESS}$ 
 and
$\dom{\ESS_R} = \ran{\ESS}$
and
 $\ESS \circ \ESS_R = {\ESS_L} \circ \rho$,
we have $A \vDash^{\ESS_L} \phi$ if and only if $B \vDash^{\ESS_R} \phi\rho$.
\end{thm}
\begin{proof}
Assume 
$A \mathrel{\mathcal{R}^{\rho,\ESS}} B$
such that $\mathcal{R}$ is an HP-bisimulation.
We proceed by induction on the structure of $\FM$ formulae $\phi$.
In what follows, assume that
$\ESS_L$ and $\ESS_R$ are such that $\dom{\ESS_L} = \dom{\ESS}$,
$\dom{\ESS_R} = \ran{\ESS}$ and 
$\ESS \circ \ESS_R = {\ESS_L} \circ \rho$.

\textit{Case for diamond modalities.}
Consider when
$A \vDash^{\ESS_L} \ediam{\pi}{u} \phi$,
where $\pi$ is not an output action
and
$\ESS_L = \ESS^1_L \cup \ESS^2_L$
such that
 $\event{\pi}{u} \Indy \dom{\ESS^1_L}$
 and
 $\event{\pi}{u} \notIndy \dom{\ESS^2_L}$.
Hence, by definition of the diamond modality,
for some $A'$ and $u'$ we have,
 $A \dlts{\pi}{u'} A'$ 
such that
 $\event{\pi}{u'} \Indy \dom{\ESS^1_L}$
 and
 $\event{\pi}{u'} \notIndy \dom{\ESS^2_L}$
 and, furthermore,
$A' \vDash^{\ESS^1_L \cup \left\{ \left( \event{\pi}{u'}  ,  \event{\pi}{u} \right) \right\}} \phi$.

Observe that since  
$\dom{\ESS_L} = \dom{\ESS}$,
we have that
$\ESS = \ESS^1 \cup \ESS^2$
such that
$\dom{\ESS^1} = \dom{\ESS_L^1}$
and
$\dom{\ESS^2} = \dom{\ESS_L^2}$.
Hence, since
 $\event{\pi}{u'} \Indy \dom{\ESS^1_L}$
 and
 $\event{\pi}{u'} \notIndy \dom{\ESS^2_L}$,
 we have
 that
 $\event{\pi}{u'} \Indy \dom{\ESS^1}$
 and
 $\event{\pi}{u'} \notIndy \dom{\ESS^2}$.
Now, since $\mathrel{R}$ is an HP-bisimulation and
$A \mathrel{\mathcal{R}^{\rho,\ESS}} B$
there exists $B'$ and location $v$
such that
$B \dlts{\pi\rho}{v} B'$
and
 $\event{\pi\rho}{v} \Indy \ran{\ESS^1}$
 and
 $\event{\pi\rho}{v} \notIndy \ran{\ESS^2}$
 and, furthermore,
$A' \mathrel{\mathcal{R}^{\rho,\ESS^1 \cup \left\{ \left(\event{\pi}{u'}, \event{\pi\rho}{v}\right) \right\}}} B'$.

Now, since $\dom{\ESS_R} = \ran{\ESS}$
and
$\ESS = \ESS_1 \cup \ESS_2$
such that
 $\event{\pi\rho}{v} \Indy \ran{\ESS^1}$
 and
 $\event{\pi\rho}{v} \notIndy \ran{\ESS^2}$
 and also
 $\ESS \circ \ESS_R = {\ESS_L} \circ \rho$,
we have that
$\ESS_R = \ESS^1_R \cup \ESS^2_R$
such that
 $\event{\pi\rho}{v} \Indy \dom{\ESS_R^1}$
 and
 $\event{\pi\rho}{v} \notIndy \dom{\ESS_R^2}$
 and 
 $\dom{\ESS^1_L} = \dom{\ESS^1}$ 
 and
$\dom{\ESS^1_R} = \ran{\ESS^1}$
and
 also
$\ESS^1 \circ \ESS^1_R = {\ESS^1_L} \circ \rho$.
Now, by the induction hypothesis,
since 
\begin{itemize}
	\item $A' \mathrel{\mathcal{R}^{\rho,\ESS^1 \cup \left\{ \left(\event{\pi}{u'}, \event{\pi\rho}{v}\right) \right\}}} B'$, 
	\item $\dom{\ESS^1_L  \cup \left\{ \left(\event{\pi\rho}{v}, \event{\pi\rho}{u}\right) \right\}}  = \dom{\ESS^1  \cup \left\{ \left(\event{\pi}{u'}, \event{\pi\rho}{v}\right) \right\}}$,
	\item $\dom{
		\ESS^1_R  \cup 
		\left\{
		\left(\event{\pi\rho}{v}, \event{\pi\rho}{u}\right)
		\right\}}
	=
	\ran{
		\ESS^1  \cup
		\left\{
		\left(\event{\pi}{u'}, \event{\pi\rho}{v}\right)
		\right\}}$,
	\item $\ESS^1 \circ \ESS^1_R = {\ESS^1_L} \circ \rho$ and 
	\item $A' \vDash^{\ESS^1_L \cup \left\{ \left( \event{\pi}{u'}  ,  \event{\pi}{u} \right) \right\}} \phi$,
\end{itemize}
we have that
$B' \vDash^{\ESS^1_R \cup \left\{ \left( \event{\pi\rho}{v}, \event{\pi\rho}{u} \right) \right\}} \phi\rho$ holds.
Therefore, by the definition of diamond modalities,
since 
$B \dlts{\pi\rho}{v} B'$
and
$\event{\pi\rho}{v} \Indy \dom{\ESS_R^1}$
and
$\event{\pi\rho}{v} \notIndy \dom{\ESS_R^2}$
and
$\event{\pi\rho}{u} \Indy \dom{\ESS_R^1}$
and
$\event{\pi\rho}{u} \notIndy \dom{\ESS_R^2}$,
we can conclude
that 
$B' \vDash^{\ESS_R} \ediam{\pi\rho}{u}\phi\rho$ holds.
Thus $B' \vDash^{\ESS_R} \left(\ediam{\pi}{u}\phi\right)\mathclose{\rho}$ holds, as required.

The case for output diamond modalities is similar to the above but, in addition $\rho$ is extended.
The remaining cases are standard.

\textit{Base case for true.}
$B \vDash^{\ESS_R} \ttt$ holds for any $B$ and $\ESS_R$.
Hence this case is immediate.
 
\textit{Base case for equality.}
Consider when 
$A \vDash^{\ESS_L} M = N$.
Observe that $\ESS_L$ does not impact equality,
that is, for any $\ESS'$, 
$A \vDash^{\ESS'} M = N$ iff $A \vDash^{\empty} M = N$.
Now, 
since $A \mathrel{\mathcal{R}^{\rho, \ESS}} B$ and 
$\mathcal{R}$ is an HP-bisimulation,
and 
$A \vDash^{\empty} M = N$,
by static equivalence, 
$B \vDash^{\empty} M\rho = N\rho$,
and hence by the above observation,
$B \vDash^{\ESS_R} M\rho = N\rho$, as required.

\textit{Case of conjunction.}
Consider when $A \vDash^{\ESS_L} \phi \wedge \psi$
hence $A \vDash^{\ESS_L} \phi$ and $A \vDash^{\ESS_L} \psi$.
So, since $A \mathrel{\mathcal{R}^{\rho, \ESS}} B$,
  by the induction hypothesis $B \vDash^{\ESS_R} \phi\rho$ and $B \vDash^{\ESS_R} \psi\rho$
and hence $B \vDash^{\ESS_R} \phi\rho \wedge \psi\rho$, as required.

\textit{Case of negation.}
Consider when $A \vDash^{\ESS_L} \neg \phi$,
hence $A \nvDash^{\ESS_L} \phi$.
Now, assume for contradiction that
$B \vDash^{\ESS_R} \phi\rho$ holds.
By symmetry, 
$B \mathrel{\mathcal{R}^{\rho^{-1}, \ESS^{-1}}} A$ and 
hence, by the induction hypothesis, 
$A \vDash^{\ESS_L} \phi\rho\rho^{-1}$ would hold,
contradicting the assumption that $A \vDash^{\ESS_L} \phi$ holds.
Thus, $B \nvDash^{\ESS_R} \phi\rho$ and hence $B \vDash^{\ESS_R} \neg\phi\rho$, as required.

Therefore, by induction on the structure of $\phi$, for all formulae $\phi$, and for all $A$, $B$ such that $A \mathrel{\rho, \mathcal{S}} B$, we have $A \vDash^{\ESS_L} \phi$ iff $B \vDash^{\ESS_R} \phi\rho$,
such that 
$\dom{\ESS_L} = \dom{\ESS}$ 
 and
$\dom{\ESS_R} = \ran{\ESS}$
and
 $\ESS \circ \ESS_R = {\ESS_L} \circ \rho$.

The converse implication then follows by symmetry of a bisimulation.
\end{proof}

\begin{thm}[Expressive completeness of HP-$\FM$]
If
for all $\FM$ formulae $\phi$
we have $A \vDash^{\emptyset} \phi$ if and only if $B \vDash^{\emptyset} \phi$,
then
$A \bisimi{HP} B$.
\end{thm}
\begin{proof}
The strategy is to show that
the following relation is an HP-bisimulation.
\begin{align*}
	\mathcal{R} &= 
	\left\{
	\left(A, \rho, \ESS, B\right)
	\colon
	\forall \phi, \ESS_L, \ESS_R,
	\text{\st }
	\begin{aligned}
		& \dom{\ESS_L} = \dom{\ESS}, \dom{\ESS_R} = \ran{\ESS},\\
		& \ESS \circ \ESS_R = {\ESS_L} \circ \rho \mbox{ and }A \vDash^{\ESS_L} \phi \text{ iff } B \vDash^{\ESS_R} \phi\rho  
	\end{aligned}
	\right\}
\end{align*}
   We aim to prove that
   $\mathcal{R}$ is an HP-bisimulation.
Symmetry is immediate.
In the following cases assume $A \mathrel{\mathcal{R}^{\rho, \ESS}} B$.

\textit{Case for input and tau transitions.}
Suppose $A \dlts{\pi}{u} A'$. 
Hence $A \vDash^{\ESS_L} \ediam{\pi}{u}\ttt$,
so, by definition of $\mathcal{R}$,
we have 
$\ESS_L, \ESS_R,
   ~\mbox{s.t.}~
   \dom{\ESS_L} = \dom{\ESS},
   \dom{\ESS_R} = \ran{\ESS},
   \ESS \circ \ESS_R = {\ESS_L} \circ \rho$,
and 
$B \vDash^{\ESS_R} \ediam{\pi\rho}{u}\ttt$.
Therefore, for some $B'$ we have $B \dlts{\pi\rho}{v} B'$.
By event determinisation (see~\cite{Aubert2022e}),
$B'$ is unique up to equivariance.
Indeed, it is easy to check $\equiv$ preserves satisfiability, by induction over the definition of satisfaction.

Observe that 
$\ESS = \ESS^0 \cup \ESS^1$
such that 
$\event{\pi}{u} \Indy \dom{\ESS^0}$
and
$\event{\pi}{u} \notIndy \dom{\ESS^1}$.
Since
$\dom{\ESS} = \dom{\ESS_L}$
and
$\ran{\ESS} = \dom{\ESS_R}$,
there exists
$\ESS^0_L$ 
and
$\ESS^0_R$ 
and
$\ESS^1_L$ 
and
$\ESS^1_R$ 
such
that
$\ESS_L = \ESS^0_L \cup \ESS^1_L$ 
and
$\ESS_R = \ESS^0_R \cup \ESS^1_R$
and
$\dom{\ESS^0} = \dom{\ESS^0_L}$
and
$\ran{\ESS^0} = \dom{\ESS^0_R}$.
Hence also we have
$\ESS^0 \circ \ESS^0_R = {\ESS^0_L} \circ \rho$.

Now, suppose for contradiction
that
$A' \mathrel{\mathcal{R}^{\rho,\ESS^0 \cup \left\{\left( \event{\pi}{u}, \event{\pi\rho}{v} \right) \right\}}} B'$
does not hold.
If that were true, then there exists $\phi$ such that
$A' \vDash^{\ESS^0_L \cup \left\{\left( \event{\pi}{u}, \event{\pi}{u} \right) \right\} } \phi$
and
$B' \nvDash^{\ESS^0_R \cup \left\{\left( \event{\pi\rho}{v}, \event{\pi\rho}{u} \right) \right\}} \phi\rho$.
Hence $A \vDash^{\ESS_L} \ediam{\pi}{u}\phi$
but
$B \nvDash^{\ESS_R} \ediam{\pi\rho}{u}\phi\rho$,
contradicting the assumption that
$A \mathrel{\mathcal{R}^{\rho,\ESS}} B$.
Hence, 
$A' \mathrel{\mathcal{R}^{\rho, \ESS^0 \cup\left\{\left( \event{\pi}{u}, \event{\pi\rho}{v} \right) \right\}  }} B'$,
as required.

The case outputs is similar to the above but, in addition $\rho$ is extended.

\textit{Case of static equivalence.}
This folows immediately from the
definition of $\mathcal{R}$,
since for any $\ESS_L$
$A \vDash^{\ESS_L} M = N$
for some $\ESS_R$, we have
$B \vDash^{\ESS_R} M\rho = N\rho$.
Hence if
$A \vDash^{\emptyset} M = N$
then
$B \vDash^{\emptyset} M\rho = N\rho$.
Furthermore, by symmetry of $\mathcal{R}$ and since $\rho$ is a bijection,
if $B \vDash^{\emptyset} M\rho = N\rho$
then
$A \vDash^{\emptyset} M\rho\rho^{-1} = B\rho\rho^{-1}$,
and hence 
$A \vDash^{\emptyset} M = N$.
Therefore $A$ and $B$ are statically equivalent.

Thus, $\mathcal{R}$ is an HP-bisimulation, as required.
\end{proof}

\subsection{Example HP-simulation}

By including the proof of the following theorem, 
we also illustrate the proof method
of constructing an HP-simulation.

\systemHPtwo*

\begin{proof}
In what follows, we will make use of the following representation of intermediate states of ePassport (prover) and reader (verifier) sessions.

\begin{align*}
	\MainUK'(c,d,ke,km, nt) & \triangleq d(y). \MainUK''(c,d,ke,km, nt, y)
	\\[2em]
	\MainUK''(c,d,ke,km, nt, y) &\triangleq
		\begin{aligned}[t]
			&
			\match{ \snd{y} = \mac{\fst{y}, km} }
			\\
			\MoveEqLeft[-.5]
			\match{ nt = \fst{\snd{\dec{\fst{y}}{ke}}} }
			\\
			\MoveEqLeft[-1]\nu kt.\clet{m}{\enc{\pair{nt}{\pair{\fst{\dec{\fst{y}}{ke}}}{kt}}}{ke}} \\
			\MoveEqLeft[-1.5] \cout{c}{m, \mac{m, km}}
			\end{aligned}
		\\[2em]
		\Reader'(c,d,ke,km, nt) & \triangleq
		\begin{aligned}[t]
			& \nu nr.\nu kr. \\
			& \clet{m}{\enc{\pair{nr}{\pair{nt}{kr}}}{ke}} \\
			\MoveEqLeft[-.5] \cout{c}{m, \mac{\pair{m}{km}}}
	\end{aligned}
\end{align*}

To construct a proof certificate, 
we build a relation that is an HP-simulation upto equivariance ($\alpha$-conversion and swapping of bound names).
Let $\mathcal{R}$ be the least symmetric relation below,
where the various parameters we will define progressively in what follows.
\[
\mathopen{\nu \vec{z}.}\left(
\theta \cpar 
\left(
Q_1
\cpar 
\ldots
\left(
Q_n
\cpar
\SystemMin \right) \ldots \right)\right)
\mathrel{\mathcal{R}^{\rho, \ESS}}
\mathopen{\nu \vec{w}.}\left(
\sigma
\cpar
\left(
T_{1} \cpar \ldots \left(T_{2m}
\cpar 
\SpecMin \right)
\ldots
\right)\right)
\]
The relation above is defined for all natual numbers $n$ and sequences of natural numbers $m_1, m_2, \ldots m_n$,
where $m_i$ will be used to enumerate the number of sessions of prover $i$ that are represented in the system.
In the above, let $m = \sum_{i = 0}^{n} m_i$, \ie the total number of sessions represented across all provers.
Next, below, we define first the processes $Q_1, \ldots Q_n$ on the LHS above,
where each $Q_i$ is a process comprising of $m_i$ sessions involving the same passport and reader,
and $\vec{z}$ and $\theta$ representing knowledge that have been output thus far.
After that, later, we define the processes $T_1, \ldots T_{2m}$ representing individual sessions
on the RHS above, and the corresponding knowledge $\vec{w}$ and $\sigma$ output thus far by the idealised specification.

In what follows,
the sets $\varphi_0$, $\varphi_1$ and $\varphi_2$ are disjoint such that
$\left\{ 1 \ldots m \right\} = \varphi_0 \cup \varphi_1 \cup \varphi_2$, used to track states of $m$ readers
(not started, received first message only, sent response).
The sets $\vartheta_0$, $\vartheta_1$, $\vartheta^{KO}_2$, $\vartheta^{OK}_2$, and $\vartheta^{OK}_3$ are disjoint such that
$\left\{ 1 \ldots m \right\} = \vartheta_0 \cup \vartheta_1 \cup \vartheta^{KO}_2 \cup \vartheta^{OK}_2 \cup \vartheta^{OK}_3$, used to track states of $m$ ePassports
(respectively, not started, sent first message only, received bad message, received good message and authenticated, sent message after successful authentication).

\textit{LHS processes, active substitution, and extruded names.} 
In the following, also let $k_{i}^{j} = \sum_{\ell = 0}^{i}{m_{\ell}} + j$, which we use to enumerate the $j^{th}$ session involving the $i^{th}$ ePassport, according to the unfolding of system in the given state.

We set up the set of names $\vec{z}$, which may be in any order due to equivariance.
Let $\vec{z}$ be the least set of names such that, for all $i \in \left\{1 \ldots n\right\}$ such that $m_i > 0$, we have $ke_i \in \vec{z}$ and $km_i \in \vec{z}$ (that is, the key material of the ePassport must have been extruded),
and,
for all $j \in \left\{1 \ldots m_i \right\}$, we have the following:
\begin{itemize}
\item
if $k^i_j \notin \vartheta_0$
then $nt_{k^i_j} \in \vec{z}$ (\ie the $i^{th}$ ePassport has extruded a nonce in session $j$);
\item
if $k^i_j \in \varphi_2$ then $kt_{k^i_j} \in \vec{z}$ and $nt_{k^i_j} \in \vec{z}$ (\ie the $i^{th}$ reader has extruded a nonce and key in session $j$);
\item
if $k^i_j \in \vartheta^{OK}_3$ then $kr_{k^i_j} \in \vec{z}$ (\ie the $i^{th}$ ePassport has extruded a key in session $j$).
\end{itemize}

We set up the active substitution $\theta$ as follows, making use of messages $NT_{k^i_j}$ and $Y_{k^i_j}$
such that
$\isfresh{\vec{z}}{Y_{k^i_j}, NT_{k^i_j}}$.
Message $NT_{k^i_j}$ is defined whenever $k^i_j \in \varphi_1 \cup \varphi_2$, \ie the $j^{th}$ session of the $i^{th}$ reader has received an input message.
Message $Y_{k^i_j}$ is defined whenever 
$k^i_j \in \vartheta^{OK}_2 \cup \vartheta^{OK}_3 \cup \vartheta^{KO}_2$,
\ie the $j^{th}$ session of the $i^{th}$ ePassport has received an input message.
Let $\theta$ and the above input messages be related as follows.
\begin{align*}
\theta(\prefix{1^i0^j1}\lambda_1) & = nt_{k^i_j}
&&
\text{only if }k^i_j \in \vartheta_1 \cup \vartheta^{KO}_2 \cup \vartheta^{OK}_2  \cup \vartheta^{OK}_3
\\
\theta(\prefix{1^i0^{\ell}0}\lambda_2) & = \pair{
 M^i_\ell
}{
 \mac{ M^i_\ell, km_i}
}
&& \text{only if } k^i_\ell \in \varphi_2
\\
\theta(\prefix{1^i0^j1}\lambda_3) & = \pair{
 N^i_j
}{
 \mac{ N^i_j, km_i}
} &&
 \text{only if }
k^i_j \in \vartheta^{OK}_3
\\
\exists \ell,\text{ s.t. }
&
k^i_{\ell} \in \varphi_2
 \text{ and }
NT_{k^i_{\ell}} \mathrel{=_{E}} \prefix{1^i0^j1}\lambda_1
\\
&
\text{ and } Y_{k^i_j} \mathrel{=_{E}} \prefix{1^i0^j0}\lambda_2
&&
 \text{only if }
k^i_j \in \vartheta^{OK}_2 \cup \vartheta^{OK}_3
\end{align*}
where, in the above, we make use of the messages $M^i_\ell$ and $N^i_j$ below.
\begin{align*}
	M^i_\ell = \enc{\pair{ nr_{k^i_\ell} }{ \pair{ NT_{k^i_{\ell}}\theta }{ kr_{k^i_\ell} }}}{ke_i}
\qquad
	N^i_j =  \enc{\pair{ nt_{k^i_j} }{ \pair{ nr_{k^i_\ell} }{ kt_{k^i_j} }}}{ke_i}
\end{align*}
To avoid cycles in the definition above we ensure, in addition,
that $\isfresh{\prefix{1^i0^{\ell}0}\lambda_2}{NT_{k^i_{\ell}}}$.
That is, the output with alias
$\prefix{1^i0^{\ell}0}\lambda_2$ 
cannot be used to produce the input $NT_{k^i_{\ell}}$ on which it depends.

Notice that there are several constraints above that apply
when $k^i_j \in \vartheta^{OK}_3$.
These, together, ensure that
$\vartheta^{OK}_3$ concerns ePassport sessions that has successfully authenticated
and output the final message $\prefix{1^i0^j1}\lambda_3$ has been sent.
The constraints also ensure that,
for all $i$ and $j$ such that
$k^i_j \in \vartheta^{OK}_2 \cup \vartheta^{OK}_3$ (\ie authentication is successful for the $j^{th}$ session of ePassport $i$),
there exists some $\ell$
such that, 
 $k^i_{\ell} \in \varphi_2$ (i.e, there is some $\ell^{th}$ reader session with ePassport $i$ that has output a reponse)
such that
$NT_{k^i_{\ell}} \mathrel{=_{E}} \prefix{1^i0^j1}\lambda_1$
and
$Y_{k^i_j} \mathrel{=_{E}} \prefix{1^i0^j0}\lambda_2$.
That is, we have
$NT_{k^i_{\ell}}\theta \mathrel{=_{E}} nt_{k^i_j}$
and also.
\[
Y_{k^i_j}\theta \mathrel{=_{E}} 
\pair{
 \enc{\pair{ nr_{k^i_\ell} }{ \pair{ nt_{k^i_j} }{ kr_{k^i_\ell} }}}{ke_i}
}{\mac{ 
 \enc{\pair{ nr_{k^i_\ell} }{ \pair{ nt_{k^i_j} }{ kr_{k^i_\ell} }}}{ke_i}
 , km_i}}.
 \]
These conditions are sufficient and necessary to identify ePassport sessions that can reach a successful termination state.

The process $Q_i$ defined below,
is defined via the processes $R$ and $S$ representing possible states of a verifier and prover respectively.
Notice that if $m_i = 0$ then no action has happened involving the keys of the passport represented and hence they are not yet extruded.
\begin{align*}
	Q_i & = 
	\begin{dcases}
		\mathopen{\nu ke, km.\bang}\left( \Reader(c,d,ke,km) \cpar \MainUK(c,d,ke,km) \right) & \text{if } m_i = 0 \\[1em]
		\begin{aligned}
	& \left( R(ke_i, km_i)_{k_{i}^{1}}\theta \cpar S(ke_i, km_i)_{k_{i}^{1}}\theta \right) \cpar \big( \hdots\\
	& \Big(
	\left( R(ke_i, km_i)_{k_{i}^{m_i}}\theta \cpar S(ke_i, km_i)_{k_{i}^{m_i}}\theta \right)
	\cpar\\
	& \bang\left( \Reader(c,d,ke_i,km_i) \cpar \MainUK(c,d,ke_i,km_i) \right)
	\Big) \ldots \big)
	\end{aligned}
		& \text{if } m_i > 0
\end{dcases}
\\[2em]
R(ke_i, km_i)_k & = 
\begin{dcases}
	\Reader(c,d,ke_i,km_i) & \text{if }  k \in \varphi_0 \\
	\Reader'(c,d,ke_i,km_i, NT_k) & \text{if } k \in \varphi_1 \\
	0 & \text{if } k \in \varphi_2
\end{dcases}
\\[2em]
S(ke_i, km_i)_k & =
\begin{dcases}
	\MainUK(c,d,ke_i,km_i) & \text{if } k \in \vartheta_0 \\
	\MainUK(c,d,ke_i,km_i, nt_k)' & \text{if } k \in \vartheta_1 \\
	\MainUK(c,d,ke_i,km_i, nt_k, Y_k)'' & \text{if } k \in \vartheta^{KO}_2 \cup \vartheta^{OK}_2 \\
	0 & \text{if } k \in \vartheta^{OK}_3 
\end{dcases}
\end{align*}
Notice that processes $R$ and $S$ may contain input messages,
to which the active substitution has not yet been applied.
Hence the active substitution is applied to them in process $Q$.

\textit{RHS processes, active substitution, and extruded names.} 
Now, to set up the other side of the relation,
let $f \colon \left\{1 \ldots m \right\} \rightarrow \left\{1 \ldots m\right\}$ be some bijection.
The idea is that for every ePassport-reader pair, $k$, in the system, 
the ePassport-Reader pair $2f(k)$ in the specification is in the same state,
unless the ePassport has received an input that it can authenticate (that is $k \in \vartheta^{OK}_2 \cup \vartheta^{OK}_3$),
in which case the thread in the specification does not progress
and instead reader $2f(k) + 1$ simulates the final input and output of a successful authentication 
attempt involving the ePassport,
which are indistinguishable from the steps an ePassport would take if they were to be successfully authenticated.

We set up the set of names $\vec{w}$, which may be in any order due to equivariance.
Let $\vec{w}$ be the least set of names such that, for all
$k \in \left\{1 \ldots m\right\}$, we have the following:
\begin{itemize}
\item
If $k \notin \vartheta_0 \cap \varphi_0$, then
$ke_{2f(k)} \in \vec{w}$ and $km_{2f(k)} \in \vec{w}$ (\ie either the ePassport or reader in session $k$ has performed an action).
\item
If $k \notin \vartheta_0$ then $nt_{2f(k)} \in \vec{w}$ (\ie the ePassport in session $k$ has extruded its first nonce).
\item
If $k \in \varphi_2$
then $nr_{2f(k)} \in \vec{w}$ and $kr_{2f(k)} \in \vec{w}$ (\ie the reader in session $k$ has extruded a nonce and key in response to a challenge).
\item
If $k \in \vartheta^{OK}_2 \cup \vartheta^{OK}_3$, then
$ke_{2f(k)+1} \in \vec{w}$ and $km_{2f(k)+1} \in \vec{w}$ (\ie the ePassport in session $k$ authencates and hence it is simulated by a reader in $2f(k)+1$ that must have extruded its keys.
\item
If $k \in \vartheta^{OK}_3$
then $nr_{2f(k)+1} \in \vec{w}$ and $kr_{2f(k)+1} \in \vec{w}$ (\ie as the above, but, in addition, the ePassport has responded in session $k$. This results in key material being extruded by the reader $2f(k)+1$ that simulates it).
\end{itemize}

We now define the active substitution $\sigma$ as follows.
We make use of 
a bijection between $\rho \colon \dom{\theta} \rightarrow \dom{\sigma}$ such that
$\rho(\prefix{1^i0^j1}\lambda_1) =  \prefix{0^{2f(k^i_j)}1}\lambda_1$
and
$\rho(\prefix{1^i0^j0}\lambda_2) = \prefix{0^{2f(k^i_j)}0}\lambda_2$
and
$\rho(\prefix{1^i0^j1}\lambda_3) = \prefix{0^{2f(k^i_j)+1}0}\lambda_3$.
\begin{align*}
	\sigma(\prefix{0^{2f(k^i_j)}1}\lambda_1) & = nt_{k^i_j}
	&&
	\text{only if }
	k^i_j \notin \vartheta_0 
	\\
	\sigma(\prefix{0^{2f(k^i_j)}0}\lambda_2) & =
	\pair{ \hat{M}^i_j }{\mac{ \hat{M}^i_j, km_{2f(k^i_j)}}}
	&&
	\text{only if } k^i_j \in \varphi_2
	\\
	\sigma(\prefix{0^{2f(k^i_j)+1}0}\lambda_3) & =
	\pair{ \hat{N}^i_j }{
 \mac{ \hat{N}^i_j, km_i}
} && \text{iff } k^i_j \in \vartheta^{OK}_3
\end{align*}
where, in the above, we make use of the following in the output of a reader used to simulate the successful authentication of an ePassport.
\[
\hat{M}^i_j = \enc{\pair{ nr_{f(k^i_j)} }{ \pair{ NT_{k^i_j}\rho\sigma }{ kr_{f(k^i_j)} }}}{ke_{2f(k^i_j)} }
\]
\[
\hat{N}^i_j = \enc{\pair{ nr_{f(k^i_j)+1} }{ \pair{ Y_{k^i_j}\rho\sigma }{ kr_{f(k^i_j)+1} }}}{ke_{2f(k^i_j)+ 1} }
\]
Notice that the constraints on $NT_{k^i_j}$ and $Y_{k^i_j}$, that appear  in $M^i_j$ and in $K^i_j$, respectively, are 
inherited from the constraints on the LHS.

Following the above proof idea we define for each $k \in \left\{1 \ldots m \right\}$ we define
$T_{2f(k)}$ as follows, where it behaves as the $k^{th}$ process on the LHS, except that the ePassport does not progress if
it were to receive a message that would result in it being authenticated.
\[
T_{2f(k)} = 
\begin{dcases}
	\mathopen{\nu ke, km.}\left( \Reader(c,d,ke,km) \cpar \MainUK(c,d,ke,km) \right) & \text{if } k \in \varphi_0 \wedge k \in \vartheta_0\\
	R(ke_{2k}, km_{2k})_{2k} \cpar S(ke_{2k}, km_{2k})_{2k} \\
	\MoveEqLeft[8.5] \text{if $\left( k \notin \varphi_0 \vee k \notin \vartheta_0 \right) \wedge  k \notin \vartheta^{OK}_2 \cup \vartheta^{OK}_3$}\\
	R(ke_{2k}, km_{2k})_{2k} \cpar \MainUK(c,d,ke_{2k},km_{2k}, nt_{2k})' & \text{if } k \in \vartheta^{OK}_2 \cup \vartheta^{OK}_3 
\end{dcases}
\]
If the ePassport in session $k$ of the system 
has received a message that will result in successful authentication,
 we simulate it by starting a new ePassport in session $2f(k) + 1$ of the specification.
When the same ePassport outputs a message upon successful authentication, the state of the ePassport is also advanced.
 Nothing else can happen in the session $2f(k) + 1$ of the specification.
\[
T_{2f(k) + 1} = 
\begin{dcases}
	\mathopen{\nu ke, km.}\left( \Reader(c,d,ke,km) \cpar \MainUK(c,d,ke,km) \right) \\
	\MoveEqLeft[5.7] \text{if } k \not\in \vartheta^{OK}_{2} \cup \vartheta^{OK}_{3}\\
	\Reader'(ke_{2f(k)+1}, km_{2f(k)+1}, Y^i_j) \cpar \MainUK(c, d, ke_{2f(k)+1}, km_{2f(k)+1}) \\
	\MoveEqLeft[5.7] \text{if } k \in \vartheta^{OK}_2 \\
	0 \cpar \MainUK(c, d, ke_{2f(k)+1}, km_{2f(k)+1}) \\
	\MoveEqLeft[5.7] \text{if } k \in \vartheta^{OK}_3 
\end{dcases}
\]

\textit{Static equivalence.}
To show static equivalence, for this example,
it is sufficient to observe that an attacker cannot distinguish the messages obtained by breaking down message my applying deconstructors.
\begin{itemize}
\item
For $k^i_j \notin \vartheta_0$ we have
that $\left(\prefix{1^i0^j1}\lambda_1\right)\mathclose{\theta} = nt_{k^i_j}$
and
$\prefix{1^i0^j1}\lambda_1 \rho \circ \theta = nt'_{2f(k^i_j)}$, which are indistinguishable nonces.
\item
For $k^i_\ell \in \varphi_2$ we have
that \begin{align*}
\left(\fst{\prefix{1^i0^{\ell}0}\lambda_2}\right)\mathclose{\theta}
& = M^i_\ell\theta
\text{,} &
\left(\snd{\prefix{1^i0^{\ell}0}\lambda_2}\right)\mathclose{\theta}
&= \mac{ M^i_\ell\theta, km_i}\text{,}\\
\shortintertext{while we have,}
\left(\fst{\prefix{1^i0^{\ell}0}\lambda_2}\right)\mathclose{\rho\sigma}
&= \hat{M}^i_\ell
\text{,} &
\left(\snd{\prefix{1^i0^{\ell}0}\lambda_2}\right)\mathclose{\rho\sigma}
&= \mac{ \hat{M}^i_\ell\rho\sigma, km_{2f(k^i_\ell)}}\text{,}
\end{align*}
which are all irreducible cyphertexts and MACs, since the attack does not have a key; which are furthermore incomparable, since they all involve either different long-term keys or difference nonces.
\item
For $k^i_\ell \in \vartheta^{OK}_3$ we have
that $
\left(\fst{\prefix{1^i0^{\ell}1}\lambda_3}\right)\mathclose{\theta}
= N^i_\ell\theta
$
and
$
\left(\snd{\prefix{1^i0^{\ell}1}\lambda_3}\right)\mathclose{\theta}
= \mac{ N^i_\ell\theta, km_i}
$,
while we have, respectfully,
$
\left(\fst{\prefix{1^i0^{\ell}1}\lambda_3}\right)\mathclose{\rho\sigma}
= \hat{N}^i_\ell
$
and
$
\left(\snd{\prefix{1^i0^{\ell}1}\lambda_3}\right)\mathclose{\rho\sigma}
= \mac{ \hat{N}^i_\ell\rho\sigma, km_{2f(k^i_\ell)+1}}
$,
which are again all irreducible cyphertexts and MACs.
\end{itemize}

\textit{The histories.}
We now define a relation over the 
the most recent event in each thread, thereby respecting causal consistency,
that furthermore do not contain any outputs that are used in another message
either currently or previously active.
We define $\ESS$
for some 
$C1^i_j, C2^i_j, C3^i_j, D1^i_j, D2^i_j$ 
such that
$C1^i_j  \mathrel{=_E} c$, $C2^i_j \mathrel{=_E} c$, $C3^i_j \mathrel{=_E} c$,
$D1^i_j  \mathrel{=_E} d$, $D2^i_j  \mathrel{=_E} d$,
to be the least relation such that
the following hold.
Below, \enquote{is fresh} is to be read \wrt the actions of all events in $\ESS$
 and also with respect to $NT^i_j$ if $k^i_j \in \varphi_2$\footnote{This condition ensures
that if a reader has responded after receiving a message containing an alias,
then the output producing the alias has been removed from the relation,
which notably happens during a successful authentication session. The subsequent condition performs a similar check regarding the final output of an ePassport.}
and
with respect to $Y^i_j$ is $k^i_j \in \vartheta^{OK}_3$.
\begin{align*}
	&
	\text{If }
	k^i_j \in \vartheta_1 \text{ and }\prefix{1^i0^j1}\lambda_1 \text{ is fresh,}
	\text{ then }
	\event{ \co{C1^i_j}(\prefix{1^i0^j1}\lambda_1) }{ \prefix{1^i0^j1}[] } \ESS  \event{ \co{C1^i_j\rho}(\prefix{0^{2f(k^i_j)}1}\lambda_1) }{ \prefix{0^{2f(k^i_j)}1}[] }
	\text{.}
	\\
	&
	\text{If }
	k^i_j \in \varphi_1
	\text{then }
	\event{  D2^i_j\,NT^i_j }{ \prefix{1^i0^j0}[] } \ESS  \event{ D2^i_j\rho\,NT^i_k\rho }{ \prefix{0^{2f(k^i_j)}0}[] }
	\text{.}
	\\
	&
	\text{If }
	k^i_j \in \varphi_2
	\text{ and }
	\prefix{1^i0^j0}\lambda_2
	\text{ is fresh,}
	\text{then }
	\event{ \co{C2^i_j}(\prefix{1^i0^j0}\lambda_2) }{ \prefix{1^i0^j0}[] } \ESS  \event{ \co{C2^i_j}(\prefix{0^{2f(k^i_j)}0}\lambda_2) }{ \prefix{0^{2f(k^i_j)}0}[] }
	\text{.}
	\\
	&
	\text{If }
	k^i_j \in \vartheta^{KO}_2
	\text{ then }
	\event{ D3^i_j\,Y^i_j }{ \prefix{1^i0^j1}[] } \ESS  \event{ D3^i_j\rho\,Y^i_j\rho }{ \prefix{0^{2f(k^i_j)}0}[] }
	\text{.}\\
	&
	\text{If }
	k^i_j \in \vartheta^{OK}_2
	\text{ then }
	\event{ D3^i_j\,Y^i_j }{ \prefix{1^i0^j1}[] } \ESS  \event{ D3^i_j\rho\,Y^i_j\rho }{ \prefix{0^{2f(k^i_j)+1}0}[] }
	\text{.}\\
	&
	\text{If }
	k^i_j \in \vartheta^{OK}_3
	\text{ and }
	\prefix{1^i0^j1}\lambda_3
	\text{ is fresh,}
	\text{ then }
	\event{ \co{C3^i_j}(\prefix{1^i0^j1}\lambda_3) }{ \prefix{1^i0^j1}[] } \ESS  \event{ \co{C3^i_j\rho}(\prefix{0^{2f(k^i_j)+1}0}\lambda_3) }{ \prefix{0^{2f(k^i_j)+1}0}[] }
	\text{.}\end{align*}
Notice that all of the events in $\ESS$ are necessarily concurrent, 
since structural causality is preserved due to the partitioning of states of the readers and ePassports,
and also link causal is respected due to the freshness of any aliases bound by output events that occur.
	
The result follows by checking that all the above constraints are preserved under all transitions performed by the system and matched by the specification, which is routine. 
The most interesting case is where an input in a successful authentication session is received by the ePassport.
That is, we have
$k^i_j \in \vartheta_1$
and
$NT^i_\ell \mathrel{=_E} \prefix{1^i0^j1}\lambda_1$ 
and
$k^i_\ell \in \varphi_2$.
In that case,
the LHS can make a transition labelled with event
$\event{ D3^i_j\,Y^i_j }{ \prefix{1^i0^j1}[] }$
where $Y^i_j \mathrel{=_E} \prefix{1^i0^{\ell}0}\lambda_1$.
By the constraints on $\ESS$ there are no structural dependencies on that event,
at most the causal dependency 
$\event{ D3^i_j\,Y^i_j }{ \prefix{1^i0^j1}[] }
\notIndy
\event{ \co{C2^i_j}(\prefix{1^i0^j0}\lambda_2) }{ \prefix{1^i0^j0}[] }
$
holds (assuming some other input has not used the alias $\prefix{1^i0^j0}\lambda_2$ already thereby removing it from $\ESS$).
The same holds for the RHS with respect to
a transition 
labelled with event
$\event{ D3^i_j\rho\,Y^i_j\rho }{ \prefix{0^{2f(k^i_j)+1}0}[] }$
and
which cancels the event
$\event{ \co{C2^i_j}(\prefix{0^{2f(k^i_j)+1}0}\lambda_2) }{ \prefix{0^{2f(k^i_j)}0}[] }$
only,
again due to the reuse of the alias (recalling that $\rho : \prefix{1^i0^j0}\lambda_2 \mapsto \prefix{0^{2f(k^i_j)+1}0}\lambda_2$).
Thus the updated relation after the transition on the LHS
will be 
$\ESS \setminus \left\{\left( 
\event{ \co{C2^i_j}(\prefix{1^i0^j0}\lambda_2) }{ \prefix{1^i0^j0}[] }
,
\event{ \co{C2^i_j}(\prefix{0^{2f(k^i_j)}0}\lambda_2) }{ \prefix{0^{2f(k^i_j)}0}[] }
\right)
\right\}
\cup
\left\{\left(
\event{ D3^i_j\,Y^i_j }{ \prefix{1^i0^j1}[] }
,
\event{ D3^i_j\rho\,Y^i_j\rho }{ \prefix{0^{2f(k^i_j)+1}0}[] }
\right)\right\}$.
This is, exactly the relation expected when all constraints are the same except that 
$k^i_j \in \vartheta^{OK}_2$
rather than $k^i_j \in \vartheta_1$.
\end{proof}

\end{document}